\documentclass[11pt]{article}
\usepackage[left=1in, right=1in, top=1in, bottom=1in, margin=1in]{geometry}

\usepackage{tcolorbox}
\usepackage{ninecolors}
\usepackage{xcolor}
\usepackage{labelschanged}

\usepackage{amsthm}
\usepackage{amsmath,amssymb,xspace,graphicx,relsize,bm,breqn,multirow}
\usepackage{multicol}

\usepackage{tikz}
\usetikzlibrary{arrows.meta,calc,fit,positioning, arrows.meta,calc,decorations.pathreplacing}

\usepackage{braket}
\usepackage{qcircuit}
\usepackage{adjustbox}

\usepackage{float}
\usepackage[linesnumbered,ruled,vlined]{algorithm2e}
\SetKwInput{KwInput}{Input}
\SetKwInput{KwOutput}{Output}
\SetKwInOut{Promise}{Promise}
\SetKwInput{Goal}{Goal}
\SetKwProg{Fn}{function}{}{}
\SetKwFor{RepTimes}{repeat}{times}{}
\SetKwFunction{Ver}{Verify}
\SetKwFunction{Prep}{PrepareState}
\SetKwComment{Comment}{/* }{ */}
\usepackage{algorithmic}

\usepackage{paralist}

\newtcolorbox{myalgorithm}[1][]{
    colback=gray!10, % Light grey background
    colframe=black, % Black border
    arc=5pt, % Curved corners
    boxrule=0.5pt, % Thin border
    left=0pt, right=0pt, top=0pt, bottom=0pt % Padding inside the box
}
\newcommand{\Exp}{\mathop{\mathbb{E}}}

\usepackage[margin=1in]{geometry}

\newcommand{\E}{\mathbb{E}}

\usepackage{amsthm}
\usepackage{dsfont}
\usepackage{array}
\usepackage{makecell}

\newcommand{\TV}{\ensuremath{\mathsf{TV}}}

\newcommand{\poly}{\ensuremath{\mathsf{poly}}}
\newcommand{\GCz}{\textsf{GC}^0}
\newcommand{\Index}{\textsf{Inx}}
\newcommand{\MIndex}{\textsf{MultInx}}
\newcommand{\IIndex}{\textsf{IterInx}}
\newcommand{\IMIndex}{\textsf{IterMultInx}}

\newcommand{\R}{\ensuremath{\mathbb{R}}}

\newcommand{\id}{\ensuremath{\mathbb{I}}}

\newif\ifcomments
\commentstrue
\ifcomments
\newcommand{\hari}[1]{{\textcolor{teal}{[Hari: {#1}}}]}
\else
\newcommand{\hari}[1]{}
\fi

\newif\ifcomments
\commentstrue
\ifcomments
\newcommand{\arko}[1]{{\textcolor{orange}{[Arko: {#1}}}]}
\else
\newcommand{\arko}[1]{}
\fi

\usepackage{upgreek}
\usepackage{enumerate}

\usepackage[pagebackref]{hyperref}
\usepackage[usestackEOL]{stackengine}
\usepackage{thm-restate,mathrsfs}
\usepackage{enumerate}
\usepackage{array}
\usepackage{parskip}
\def\01{\{0,1\}}
\newcommand{\ketbra}[2]{|#1\rangle\langle#2|}
\hypersetup{
	colorlinks,
	linkcolor={blue!100!black},
	citecolor={red!100!black},
}

\newcommand{\be}{\begin{equation}}
\newcommand{\ee}{\end{equation}}
\newcommand{\ba}{\begin{array}}
\newcommand{\ea}{\end{array}}
\newcommand{\bea}{\begin{eqnarray}}
\newcommand{\eea}{\end{eqnarray}}

\usepackage{mathtools}
\DeclarePairedDelimiter\ceil{\lceil}{\rceil}

\newcommand{\blk}{\textsf{blk}}

\newcommand{\calA}{{\cal A }}

\newcommand{\calC}{{\cal C }}

\newcommand{\calT}{{\cal T }}

\newcommand{\DLM}{\textsf{DLM}}
\newcommand{\LLM}{\textsf{LLM}}
\newcommand{\MASK}{\textsf{mask}}

\newcommand{\AND}{\textsf{AND}}

\usepackage{amsthm}

\def\01{\{0,1\}}

\newcommand{\supp}{\mathsf{supp}}

\newcommand{\CNOT}{\textsf{CNOT}}
\newcommand{\NCz}{\textsf{NC}^0}
\newcommand{\QNCz}{\textsf{QNC}^0}
\newcommand{\ACz}{\textsf{AC}^0}
\newcommand{\QACz}{\textsf{QAC}^0}
\newcommand{\TCz}{\textsf{TC}^0}

\definecolor{citegreen}{HTML}{208054}
\definecolor{citeblue}{HTML}{0055cc}

\NineColors{saturation=high}
\hypersetup{
    breaklinks=true,   % splits links across lines
    colorlinks=true, % false: boxed links; true: colored links
    linkcolor=blue3, % color of internal links
    citecolor=green5, % color of links to bibliography
    urlcolor=blue3, % color of external links
}

\newtheoremstyle{spaced}
  {10pt} % Space above
  {10pt} % Space below
  {\itshape} % Body font
  {} % Indent amount
  {\bfseries} % Head font
  {.} % Punctuation after head
  {0.5em} % Space after head
  {} % Head spec

\theoremstyle{spaced}

\newtheorem{theorem}{Theorem}[section]
\newtheorem{lemma}[theorem]{Lemma}
\newtheorem{claim}[theorem]{Claim}
\newtheorem{corollary}[theorem]{Corollary}

\makeatletter
\def\widebreve{\mathpalette\wide@breve}
\def\wide@breve#1#2{\sbox\z@{$#1#2$}%
     \mathop{\vbox{\m@th\ialign{##\crcr
\kern0.08em\brevefill#1{0.8\wd\z@}\crcr\noalign{\nointerlineskip}%
                    $\hss#1#2\hss$\crcr}}}\nolimits}
\def\brevefill#1#2{$\m@th\sbox\tw@{$#1($}%
    \hss\resizebox{#2}{\wd\tw@}{\rotatebox[origin=c]{90}{\upshape(}}\hss$}

\title{Separating quantum circuits from classical LLMs}
\author{Srinivasan Arunachalam$^*$ \and Arkopal Dutt$^*$ \and  Hari Krovi$^*$ \and Rik Sengupta\thanks{IBM Research. \href{mailto:srinivasan.arunachalam@ibm.com}{srinivasan.arunachalam@ibm.com}, \href{mailto:arkopal@ibm.com}{arkopal@ibm.com}, \href{mailto:hari.krovi@ibm.com}{hari.krovi@ibm.com}, \href{mailto:rik@ibm.com}{rik@ibm.com}.}}
\date{}
\begin{document}
\maketitle

\begin{abstract}
Modern large language models -- transformers and diffusion language models -- are built around two canonical algorithmic tasks: \emph{prediction and generation}. We prove unconditional separations between low-depth quantum computation and the corresponding bounded-resource classical language-model architectures in both~regimes. Concretely, we exhibit the following:
\vspace{2mm}
    \begin{enumerate}
   \item  \textbf{Distributional separation.} We give a distribution that is sampleable by $\textsf{QNC}^0$  circuits (i.e.,  a family of constant-depth  quantum circuits consisting of bounded fan-in gates) that no constant-round diffusion language model ($\textsf{DLM}$) with shallow scheduling and denoising can sample within constant distance, even when allowed sublinear chain-of-thought and output-token revision/remasking events, the very features modern $\textsf{DLM}$s rely~on.
  \vspace{2mm}
\item \textbf{Functional separation.}
We exhibit a function computable in $\land \circ \textsf{QNC}^0[\log\log n]$ (i.e., a family of $O(\log\log n)$-depth $\textsf{QNC}^0$ circuits, where $n$ is the input length, followed by a single classical $\mathsf{AND}$ gate) such that any constant-depth decoder-only transformer computing the function must be \emph{large}: it would have to have width $n^{\Omega(1)}$.
  \vspace{2mm}  
    \end{enumerate}
    Together, our work initiates the study of quantum advantage in the era of large~language~models.
\end{abstract}

\setcounter{tocdepth}{2}
{\footnotesize \tableofcontents}

\newpage 

\section{Introduction}\label{sec:intro}

A central task in quantum computing is finding problems solvable by quantum computers but not solvable by their classical counterparts.  At its strongest, this asks for an \emph{unconditional} complexity-theoretic separation between quantum and classical computing. In the oracular model, exponential separations are well-known, but in the non-oracular model  such separations are far beyond current techniques for general polynomial-time computation. More recently, a fruitful line of work has focused on \emph{restricted} models and asking if one can show unconditional separations between quantum and classical~computation.

One such direction that has received a lot of traction recently is understanding the power of \emph{shallow} circuits, i.e., {constant-depth} quantum circuits in contrast to their classical counterparts. Shallow quantum circuits are well motivated from a practical point of view, since they are potentially implementable  on near-term fault-tolerant quantum computers. Understanding the power of these circuits was spurred by the groundbreaking work of Bravyi, Gosset, and K\"{o}nig~\cite{bravyi2018quantum}, who gave an unconditional relational separation between quantum and classical shallow circuits. Concretely, they considered the class of classical
circuits of polynomial size, \emph{constant} depth, unbounded fan-out, and bounded fan-in gates (which is often referred to as $\NCz$) and the class of quantum circuits of polynomial size and constant depth with single- or two-qubit gates (which is referred to as quantum $\NCz$, denoted $\QNCz$). In~\cite{bravyi2018quantum}, they exhibited a relational problem\footnote{Unlike a functional problem which has a unique output for every input, a \emph{relational} problem can have several outputs for an input; the goal is to produce one such valid output for a given input.} solvable by a $\QNCz$ circuit but not by any $\NCz$ circuit. Subsequent works strengthened this result in several directions, showing that the classical lower bound holds in the average case~\cite{coudron2021trading,gall2018average}, that the quantum upper bound is robust to noise~\cite{bravyi2020quantum}, that stronger separations can be obtained between $\QNCz$ vs $\ACz$~\cite{watts2019exponential},  $\QNCz$ vs $\GCz$~\cite{grewal2024improved}  (we will formally define these classes later), as well as in interactive settings~\cite{grier2020interactive} and sampling separations between $\QNCz$ and classical classes~\cite{watts2023unconditional,grier2025quantum}. These results suggest that, in general, shallow quantum circuits can generate and utilize correlations that are inaccessible to several ``standard'' shallow classical circuit~classes.  

A natural next question is whether such separations  can be shown for more ``contemporary'' classical computational models. In this direction, given the remarkable success of \emph{large language models} ($\LLM$s) in practice, a rapidly growing body of work has sought to understand the expressive power and computational limitations of various models of $\LLM$s. In this regime, transformer models~\cite{vaswani2017attention} have been the dominant architecture for prediction-based language modeling, while diffusion language models~\cite{sohl2015deep,ho2020denoising} ($\DLM$s) have provided an alternative paradigm for content generation.  Although these models can seem superficially different from the traditional circuit complexity classes studied in complexity theory (such as $\NCz$, $\ACz$, and $\TCz$), their theoretical abstractions often retain a key shallow-classical feature: every layer of token updates is generated by bounded-depth classical computation with limited bandwidth. 
This suggests a new form of quantum-classical comparison: rather than only asking whether one can separate $\QNCz$ from the ``theoretically-motivated" circuit classes $\NCz$, $\ACz$, or $\GCz$, we can now also ask whether shallow quantum circuits can solve computational or sampling tasks that are hard for transformers or $\DLM$s, resulting in a new realm of separation results.  Given the recent surge of interest in understanding large language models and the role of quantum computing in machine learning, a particularly timely question is the following:
\begin{quote}
\centering
    \emph{Is there a task solvable by shallow quantum circuits but not by classical $\LLM$s?}
\end{quote}

\subsection{Main results}
We show two results that separate quantum shallow circuits from transformers and $\DLM$s for functional and distributional tasks respectively. In particular, we exhibit the following:
\begin{enumerate}
    \item a function that is computable by a $O(\log \log n)$-depth $\QNCz$ circuit followed by a single classical $\textsf{AND}$ gate, but hard to compute by a shallow-depth transformer unless it is ``large'';\footnote{ We formally define the notion of \emph{small} and \emph{large} transformers in Section~\ref{sec:transformers}.}
    \item  a distribution that is sampleable in $\QNCz$, but hard to approximately sample from a $\DLM$ which uses $\GCz$ denoisers,\footnote{The circuit class $\GCz[k]$ consists of unbounded fan-in and fanout $\textsf{AND},\textsf{OR},\textsf{NOT}$ gates along with the $\textsf{G}[k]$ gate, which  on input $x\in \{0,1\}^n$, can compute an arbitrary function of $x$ if $|x|\leq k$ and evaluates to $0$ otherwise. Throughout this work, we will restrict to $k\leq O(\log n)$.  } has sublinear chain-of-thought/workspace, and a sublinear number of output-token revision/remasking events (which are the features modern $\DLM$s rely~on).
\end{enumerate}
A key distinction from earlier shallow-circuit separations is that our lower bounds are not only against classical circuit classes such as $\NCz,\ACz,\GCz$, but against basic transformer and diffusion language models used in recent theoretical studies of modern language models~\cite{chen2024theoretical,jiang2025diffusion} to understand the strengths and weaknesses of these models. Although our tasks are deliberately complexity-theoretic, the separating models are among the closest formal proxies currently available for the prediction and generation mechanisms used~in~practice. To formulate our main results formally, we need a high-level definition of transformers and $\DLM$s, which we provide in the next sections (Sections \ref{sec:transformers-overview} and \ref{sec:dlms-overview}).  For each of these models, we give an overview of some prior \emph{classical} works in the related area, which serve as a starting point for our results; after that, we state our main theorem for that model, and give a detailed proof sketch of the separation in question.

\subsection{Functional separation}\label{sec:transformers-overview}

\subsubsection{Transformer models (a high level overview)}

In this section, we give a high-level overview of transformers, sufficient to understand the proof overview that follows subsequently. In Section~\ref{sec:transformer-prelims} we give a detailed overview of the model.

\emph{\textbf{Transformers}} form the backbone of most modern $\LLM$s. The core transformer architecture is  specified by a number of parameters: the input length $n$ (the ``context size''), the number $L$ of layers (the ``depth''), the number $H$ of attention heads, the embedding dimension $d$, and a precision parameter $p$. A transformer can be viewed as a sequence of length-preserving functions on $d$-dimensional real vectors across a sequence of \emph{layers}, once the input text has been embedded into one such sequence. The transition between every pair of adjacent layers in the transformer architecture is determined by parameterized \emph{query}, \emph{key}, and \emph{value} vectors. At a high level, the so-called \emph{attention} mechanism takes as input a vector $\mathbf{x} \in \R^n$, and then takes a position-wise weighted average of value vectors (where the weights are a function of the key and query vectors), followed by a position-wise piecewise affine transformation. In this work, we exclusively consider \emph{decoder-only} transformers, where each position attends only to positions \emph{preceding} it, reflecting the causal structure of autoregressive generation. We emphasize that almost every prominent modern $\LLM$, such as GPT~\cite{radford2019language}, Claude~\cite{claude:decoder}, Gemini~\cite{team2024gemini} and LlaMa~\cite{llama:decoder}, use variants of the decoder-only transformer architecture, due to their superiority in next-token prediction.

From the perspective of complexity, the key resource for transformers is not merely the total number of parameters, but the \emph{total amount of information}  communicated through the attention mechanism at each layer. Following recent work, this notion of ``useful'' information transmitted through the layers is captured by the so-called \emph{width} $Hdp$, which is simply the product of the number of heads, the embedding dimension, and the precision used to represent each numeric coordinate.  Following the convention of~\cite{chen2024theoretical,sanford2023representational} we say a transformer is \emph{small} if $Hdp=n^{o(1)}$, and \emph{large} otherwise. In particular, a computational task can be thought of as requiring large width if every constant-depth decoder-only transformer solving the problem requires large internal communication in between the layers. We note that the assumption of constant depth is reasonable in practice.

\subsubsection{Prior works and context}
We now provide some context and prior works in understanding the complexity of transformers, which will set up the main motivation for our results.
\begin{enumerate}[$(i)$]
    \item There have been several recent complexity theoretic works characterizing the power of transformer models in terms of circuits~\cite{chiang2022overcoming,chiang2023tighter,YCAhard,li2024chain}. Depending on the assumptions on the parameters of the transformers in question (i.e., $H,d,p$ defined earlier), their expressive power can be shown to lie somewhere roughly between $\ACz$ (e.g., for constant-precision transformers using unique hard attention) and $\TCz$ (log-precision transformers using softmax attention). Hence, it is conceivable that the results in~\cite{watts2019exponential,grewal2024improved} imply that $\QNCz$ can solve (relational) tasks that the small class of constant-precision transformers cannot solve.\footnote{We believe making this rigorous might be fruitful and leave it for future work.}
    \item Chen, Peng, and Wu~\cite{chen2024theoretical} were the first to show unconditional limitations of decoder-only transformer models (in the regime of \emph{non-constant} precision, i.e., in the $\TCz$ regime). Although their main candidate problem certifying hardness is not quite phrased in these terms, it can actually be viewed as a variant of the well-known \emph{index} function. Their lower bound was obtained via a reduction to a novel communication complexity model.
    \item Finally, a recent result of Grier et al.~\cite{QAC0tc0} shows how to compute arbitrary $\TCz$ functions in $\QACz$, when the quantum circuit is provided with $O(n^2)$ copies of the input (where $n$ is the length of the input),\footnote{Note that in classical complexity we often assume that the bits and gates have \emph{fanout}, i.e., $x_1,\ldots,x_n$ can be fed to several gates,  and so we can get polynomially many copies of the bits $x_1,\ldots,x_n$ for free. However, quantum gates do not have fanout, hence this is not a free operation for quantum circuits. Hence, one needs to account for ``copying" the input $x$ and using $t$ copies of the input increases the depth of the circuit by $O(\log t)$; so, creating $\textsf{poly}(n)$ many copies of the input adds $O(\log n)$ depth to the quantum circuit.} and further shows that computing the index function essentially requires these many copies.
\end{enumerate}
The starting point for our work is the observation that we can avoid the $O(n^2)$ copies for a certain class of problems which are also hard to compute by transformers (which computationally form a subclass of $\TCz$). In particular, we generalize the candidate function of~\cite{chen2024theoretical}, and show that this generalized function requires large width for constant-depth transformers computing it, but is computable by the ``simplest'' non-trivial class $\land\circ \QNCz[\log \log n]$, i.e., an $O(\log \log n)$-depth $\QNCz$ circuit followed by a single classical $\textsf{AND}$ gate.\footnote{We call this the \emph{simplest} class here since we show in addition that, if a function is exactly computed by $\land\circ \QNCz[\log \log n]$, then this function is actually computable by $\ACz$.}

\subsubsection{Proof sketch of the functional separation}

\textbf{Why not a relational separation?} We first explain why the well-known \emph{relational} separations used in~\cite{bravyi2018quantum,watts2019exponential,grewal2024improved}  do not imply separations between $\QNCz$ and transformers of the form that we want. To the best of our knowledge, almost all known existing techniques for proving lower bounds on transformer width are primarily via proving communication complexity lower bounds~\cite{sanford2023representational,peng2024limitations,chen2024theoretical}. In this communication model, the relevant resource is not the total number of bits of communication, but rather the \emph{per-player bandwidth}, which is typically measured \emph{proportional} to the size of that player's input. This creates a basic obstacle for relational problems: the player's output needs to be similar to the input (in order to have a well-formed problem), and if each player is allowed bandwidth proportional to their input length, then enough information can be communicated to solve the corresponding relational communication task in a straightforward way. Thus, the bandwidth regime in which we can hope for a nontrivial lower bound essentially disappears. In contrast, for functional problems,  there are meaningful intermediate bandwidth regimes: the output is a single fixed bit, so it makes sense to ask for a bound on the amount of local information that must be transmitted in order to compute this target bit. This is the reason we focus on a functional separation instead of a relational one.

\paragraph{Our candidate and lower bound.} In light of the above discussion, we turn towards functional problems, and in particular to the one considered by~\cite{chen2024theoretical}, essentially a variant of the \emph{iterated index function}.  Before discussing the iterated version, we first define the index function: on input $(i, \mathbf{x}) \in [n]\times\{0, 1\}^n$, the goal is to output $x_i$. It is a well-known fact that the index function can be computed in $\ACz$; our goal was to express this circuit in terms of quantum gates, and in particular through weak $\QACz$ circuits, which are the same as $\QNCz$ except that we also allow \emph{many-qubit} gates (potentially acting on $n$ qubits) on top of the usual one- and two-qubit gates. There have been several recent works aimed at understanding the power of multi-qubit gates in computing well-known Boolean functions~\cite{green2001counting,rosenthal2020bounds,vasconcelos2024learning,fenner2025tight,gretta2026parity,gretta2026super,joshi2026improved,anshu2025computational,bao2024learning,nadimpalli2024pauli}; however, a common concern acknowledged in these works is that it is \emph{undesirable} to allow multi-qubit gates in quantum circuits, since they require coherent interaction between arbitrary qubits, which is incompatible with locality and noise constraints in physical architectures of current quantum devices. At this point, the following sequence of results yields our separation:
\begin{enumerate}
    \item \textbf{Index with copies.} We first observe that the index function can be computed in $\QACz$ assuming we have \emph{copies} of the input. This is undesirable for two reasons: $(i)$ we require copies and $(ii)$ $\QACz$ circuits can be trivially simulated by $\QNCz$ circuits of depth $O(\log n)$. At this point, we observe that one can overcome this trivial simulation as follows. We observe that the index function can be written as a $\QNCz[\log \log n]$ circuit followed by a (large) $\textsf{AND}$ gate at the end, however we still need copies of the input $i$ in order to compute $x_i$. 
 
    \item \textbf{Generalizing to iterated index, still with copies.} Next, consider the \emph{iterated index function} (as defined by Chen et al.~\cite{chen2024theoretical}) as follows:
    $$
    \textsf{IterInx}: [m]\times [N_1]^m\times [N_2]^{\log N_1}\times \cdots \times [N_\ell]^{\log (N_{\ell-1})}\rightarrow [N_\ell],
    $$
    defined as
    \begin{align}
    \label{eq:defnofIterinx}
    \textsf{IterInx}(i,x^1,\ldots,x^\ell)=x^\ell(x^{\ell-1}(x^{\ell-2}(\cdots x^1(i))\cdots )),
    \end{align}
      for a pre-defined choice of parameters $m,N_1,N_2,\ldots,N_{\ell-1}$. 
    In other words, we iteratively compute index on $(i, x)$, look at $x_i\in [N_1]$, then evaluate $x_2$ on $x_1$ (where we view $x_1$ as a $(\log N_1)$-bit string) and repeat this iteratively for $\ell$ many functions. Using the previous bullet point wherein one can compute index using copies, we show that one can actually even compute $\textsf{IterInx}$ also in $\land \circ \QNCz[\log \log n]$ given copies of the input.
    \item \textbf{Removing the multiple copies.} The next natural step is to remove the assumption of many copies of the input being given. Using a standard trick also used in~\cite{QAC0tc0}, we first define the \emph{multi-index} function as:
    \begin{align}
    \label{eq:multiinx}
    \textsf{MultInx}: (i_1,\ldots,i_N,x)=x_i\cdot [i_1=\cdots=i_N],
    \end{align}
    which is now a total function. Note that by definition, if $i_1=\cdots=i_N$, then we effectively have access to copies of the input. However, the issue is that one needs to check the condition $[i_1=\cdots=i_N]$. For this, we show that this indicator can be checked in $\land\circ \QNCz[\log \log n]$ as well. Furthermore, the same circuit used for checking the index can in fact also compute $x_i$ without increasing the complexity (our construction is subtle in that we ensure that the eventual circuit is not in $\wedge \circ \vee\circ \QNCz[\log \log n]$, but just a single $\textsf{AND}$ gate at the end of the $\QNCz[\log \log n]$ circuit).
    \item \textbf{Generalizing to a hard problem for transformers.} Now that we have a total function $\textsf{MultInx}$, the goal is to prove a lower bound against transformers. As we observed earlier, this problem is \emph{a priori} easy for decoder-only transformers. To make it hard for the transformer model, we define our final candidate function to be the \emph{iterated multi-index} function, denoted $\textsf{IterMultInx}$ (putting together the functional computation Eq.~\eqref{eq:defnofIterinx} along with the predicate of Eq.~\eqref{eq:multiinx}). Here, we have copies of truth tables $\{x^1_i,\ldots,x^\ell_i\}_{i\in [t]}$ (the $t=1$ case is precisely  the $\textsf{IterInx}$ function) and with these copies of the truth table and index position, we show that $\IMIndex$ can be computed in our desired quantum circuit class. A description of this final candidate is quite dense and technical, and so we refer the reader to Section~\ref{sec:multiindex} for the~details. 
    \item \textbf{Putting everything together.} The main challenge in putting together everything is keeping track of several parameters which need to be carefully chosen to witness our separation. Furthermore, the current outline is simplified slightly for ease of readability, but the eventual $\IMIndex$ candidate is actually chosen with a slightly more contrived base index function which is also defined over a larger alphabet. With these technicalities in place, we concretely show that
    \begin{inparaenum}[$(i)$]
        \item the $\textsf{IterMultInx}$ can be computed still in $\land \circ \QNCz[\log \log n]$; 
        \item the $\IMIndex$ problem can be reduced to the $\textsf{IterInx}$ problem, and the hardness for transformers can be therefore imported via the lower bound by Chen et al.~\cite{chen2024theoretical}.
    \end{inparaenum}
     We note here that our lower bound is polynomially weaker than the one in \cite{chen2024theoretical}, since our problem is more~general. 
\end{enumerate}
Finally, we remark that our separation above is \emph{essentially tight}. In particular, we show that for deterministic protocols, any function $f$ that is computed by a $\land \circ \QNCz[\log \log n]$ circuit whose depth is at most $\log \log n$ (i.e., when the implicit constant in the $O(\log \log n)$ depth is at most $1$) can also be computed in $\ACz$, a circuit class that essentially corresponds to transformers with \emph{constant precision} (which is a strong assumption when it comes to transformer computations). Since our goal is to prove lower bounds against non-constant precision transformers, our quantum upper bound of $\land \circ \QNCz[\log \log n]$ (where the implicit constant in the $O(\log \log n)$ depth is strictly greater than $1$) allows us to obtain a function that is in $\QNCz$ but not computable by transformers unless their width is large, i.e., $Hdp=n^{\Omega(1)}$.

\subsection{Distributional separation}\label{sec:dlms-overview}

\subsubsection{Diffusion models (a high level overview)}

In this section, we give a high-level overview of $\DLM$s, sufficient to understand the proof overview that follows. As with transformers in Section~\ref{sec:transformers-overview}, we first give a detailed overview of the model.

\emph{\textbf{Diffusion language models}} ($\DLM$s) are a central generative paradigm for tasks such as image  and text generation. In this work, we will exclusively consider discrete diffusion models. $\DLM$s are organized differently from transformers, and are tailor-made to generate samples from distributions. At a high level, $\DLM$s generate samples by training models via the following process: they introduce random noise to training data, and then train the model to systematically remove the noise, so that the eventual model can generate new noiseless samples from the distribution. More formally, instead of generating tokens strictly from left to right (as in transformers), a $\DLM$ starts with the so-called \emph{masked} sequence $\{\MASK\}^\star$ (which one can think of as maximally random) and iteratively \emph{reveals/unmasks} tokens to produce elements from the alphabet $\Sigma$ (for simplicity in this paper we let $\Sigma=\{0,1\}$). So, at any point in the process, the $\DLM$ maintains a partially revealed sequence whose coordinates are of the form $\Sigma \cup \{\MASK\}$, where $\Sigma$ is the fixed token alphabet and $\MASK$ is an auxiliary token. The sequence contains input coordinates and output coordinates, together with possible auxiliary \emph{workspace} or \emph{chain-of-thought} (CoT) coordinates. The computation itself begins with several masked coordinates, and proceeds through a number of \emph{denoising rounds} which consists of two phases: In the first phase, a \emph{scheduler} decides on the set of currently masked coordinates to be revealed, a decision that can adaptively depend on all the (non-mask) tokens generated so far, including workspace tokens; and  in the second phase, a \emph{denoiser} circuit is applied to each newly revealed coordinate, and a value is~sampled.

The defining \emph{structural} feature of a standard $\DLM$ is
\emph{within-round conditional independence}. Once the current partially
generated sequence is fixed, the coordinates updated in that round are sampled
independently (but potentially dependent on the sequence at the beginning of the round), using the corresponding coordinate-wise denoising marginals.
Thus, after conditioning on the visible transcript and the current
workspace/CoT coordinates, a single denoising step induces a product distribution over
the coordinates it updates.  Global correlations can still arise, but only through the iterative interaction of scheduling, workspace, and future denoising rounds. This makes $\DLM$s powerful enough to build correlations adaptively, while still imposing a strong structural constraint on what can be generated in one round.

\subsubsection{Prior works and context}
The starting point for our result is the work by Jiang, Haghtalab, and Chen~\cite{jiang2025diffusion}, which studies $\DLM$s as parallel samplers and analyzes the  strengths and weakness of revision/remasking/chain-of-thought in $\DLM$ models. In particular, there is one striking feature in their work: they consider the distribution:
\begin{equation*}
U_B^\oplus=\textsf{unif}\{x\in \{0,1\}^B: |x|\equiv 0\pmod 2\}.
\end{equation*}
They show that standard $\DLM$s, without remasking or revision, cannot sample from this distribution
in $O(1)$ steps when the predictor and scheduler are implemented by $\ACz$ circuits. However, their work cannot handle the scenario where the $\DLM$s are augmented with a (bounded) number of chain-of-thought tokens, and a (bounded) number of token revision/remasking \emph{events}. Note that both of these are features that modern $\DLM$s crucially rely on. Although incomparable, the best known distribution sampling quantum-classical separation we know of is due to Bene Watts and Parham~\cite{watts2023unconditional} separating $\QNCz$ from $\NCz$; so it is natural to try and prove sampling lower bounds using the candidate problem from \cite{jiang2025diffusion} against $\DLM$s with $\ACz$ denoisers. Furthermore, several quantum/classical circuit separations~\cite{bravyi2018quantum,watts2019exponential,watts2023unconditional} are in essence based on the observation that one can sample parity-based relations and distributions in $\QNCz$. With this in mind, our main contributions are as follows:
\begin{enumerate}[$(i)$]
    \item we \emph{generalize} the candidate problem of~\cite{jiang2025diffusion} and define a \emph{block-parity function}, which still admits an efficient $\QNCz$ protocol, but is provably hard for $\DLM$s;
    \item using the generalized candidate, we are able to prove lower bounds even when we allow \emph{revision/remasking and chain-of-thought} (with the caveat that we only allow at most sublinear number of tokens for these augmented features);
    \item our lower bounds allow \emph{$\GCz$ denoisers}\footnote{We do not define the class $\GCz$ here and refer the reader to Section~\ref{sec:circuits-prelims}; we simply note that it is a classical circuit complexity class that is strictly more powerful than $\ACz$.}, which are significantly stronger than $\ACz$ denoisers;
    \item  we show hardness even to \emph{approximately sample} from the newly constructed distribution. Indeed, along the way, we show several structural properties of $\DLM$s that we hope might be of independent interest.
\end{enumerate}
The main technical contribution of our work is the classical lower bound, which combines several ideas from theoretical computer science.  First, we use a direct-product style amplification to turn many local parity constraints into a global obstruction.  Second, we view a $\DLM$ execution as a query-like process whose visible behavior, after hiding the workspace/CoT coordinates, can be
expressed as a convex combination of simple decision-tree-like distributions. Finally, we prove that block-parity functions cannot be approximated by such small convex combinations. We make this candidate problem and proof structure clear in the next section.

\subsubsection{Proof sketch of the distributional separation}
We now give a proof sketch of our $\DLM$ lower bound. The task that we consider here is to sample from the distribution:
\begin{equation*}
D_{\mathrm{blk}}
=
\bigotimes_{b=1}^M U_B^\oplus=\{(x^1,\ldots,x^M): x^i\in \{0,1\}^B, |x^i| \equiv 0\pmod 2 \text{ for all }i\}
\end{equation*}
within total variation distance $\eta$ for some fixed constant $\eta>0$. The quantum algorithm to sample from this is easy, and the main ideas have appeared before~\cite{bravyi2018quantum,watts2019exponential,grewal2024improved}. We now sketch the classical lower bound.  Assume towards a contradiction that there exists a constant-round standard
$\DLM$ which samples from $D_\blk$, with workspace length $s$, fully adaptive scheduling, bounded
revision/remasking, and polynomial-size constant depth $\GCz[\log n]$ scheduling and denoising circuits.

For each block $b\in[M]$, define its \emph{completion round} to be the first
round after which all $B$ output coordinates in that block have been generated.
Let $r_b$ be the number of previously ungenerated coordinates of block $b$
that are generated in this round. Each block is completed in one of two~ways.
\begin{itemize}
    \item $r_b \geq 2$, i.e., the $\DLM$ completes the block by unmasking many fresh coordinates at once;
    \item $r_b = 1$, i.e., the $\DLM$ completes the block by generating a single final coordinate, with the remaining $B-1$ coordinates having
    already been fixed in previous rounds.
\end{itemize}

The proof is a win--win argument over these two completion modes. Suppose first
that many blocks are completed with $r_b\ge 2$. For such a block, conditioned
on the coordinates revealed earlier in the block, the $r_b$ fresh coordinates
are not freely distributed: under the target distribution, they must satisfy one
affine parity constraint. Thus a round that completes many such blocks must
simultaneously realize many independent within-block parity constraints. We
prove a structural theorem below showing that this is impossible for the conditional
laws produced by a bounded-workspace $\DLM$. Hence, it follows that only few blocks can be completed in the multi-bit way and  for most blocks, we must have $r_b=1$. But then the $\DLM$ is forced into a
``last-bit prediction" regime: when it completes such a block, the final coordinate
has to be the parity of the previous $B-1$ coordinates. This can be used to show that a shallow $\ACz$ or $\GCz$ circuit predicts parity with constant bias, contradicting the
usual parity lower~bound of Hastad~\cite{haastad2014correlation} or Kumar~\cite{kumar2023tight}.

Formally, the argument proceeds in four main steps.

\begin{enumerate}
    \item \textbf{Mixtures of product distributions induced by $\DLM$s.}
    Fix a round $t$, and condition on the revealed output history $\tau_{\textsf{out}}$,  i.e., the revealed output coordinates for the final sample together with their values, the scheduling history up to round $t$, and the previously generated workspace/chain-of-thought tokens.  Let $w$ denote the revealed current workspace so far. Now, using the independence properties of $\DLM$s, we have that the output distribution of the $t$-th round conditioned on $\tau_{\textsf{out}}$ and $w$ is given by:
    \begin{equation*}
    P(\text{output of the $t$-th round}\mid\tau_{\textsf{out}}, w) = Q_{\tau_{\textsf{out}}, w},
    \end{equation*}
    where $Q_{\tau_{\textsf{out}}, w}$ is a product distribution over the coordinates revealed in round $t$. Hence, averaging over the possible current workspaces gives:
    \begin{align}\label{eq:DLMapproxdist}
    P(\text{output of the $t$-th round}\mid\tau_{\textsf{out}}) = \sum_w\Pr[w\mid\tau_{\textsf{out}}]\cdot Q_{\tau_{\textsf{out}}, w}. 
    \end{align}
    Since the token alphabet has constant size, the current workspace contains at most $s$ tokens, and so there are at most $2^{O(s)}$ possible values of $w$. Consequently, conditioned on
    any visible transcript, a denoising round induces only a small latent mixture of product distributions.
    
    \item \textbf{Multi-bit block completions.}
    Now, fix a transcript branch and consider a block $b$ completed in some round with $r_b \geq 2$. Under the target distribution $U_B^\oplus$, the full block satisfies $\bigoplus_{i=1}^B x_{b,i}=0$. Once some coordinates of the block have already been revealed, the remaining unrevealed coordinates are still uniform subject only to the residual parity constraint determined by the revealed prefix. Therefore, if the $\DLM$ completes the block by generating $r_b$ fresh coordinates simultaneously, the correct conditional law for these fresh coordinates is:
    \begin{equation*}
    U^{a_b}_{r_b}
    :=
    \mathrm{unif}
    \left\{
    x\in\{0,1\}^{r_b}
    :
    \bigoplus_i x_i = a_b
    \right\},
    \end{equation*}
    where $a_b$ depends on the previously revealed coordinates and the parity of their respective coordinates. Importantly, these distributions are not products over coordinates: the
    fresh coordinates inside the block must jointly satisfy a parity condition. Since $D_{\mathrm{blk}}$ is product across blocks, if many blocks are completed simultaneously, then the target conditional law~becomes:
    \begin{align}\label{eq:structureofV}
    V
    =
    \bigotimes_b U^{a_b}_{r_b},
    \end{align}
    namely, a product across blocks of internally parity-constrained factors.
    \item \textbf{Inapproximability of Eq.~\eqref{eq:structureofV}.} Now, our main structural result is that a distribution of the form given in Eq.~\eqref{eq:structureofV} 
    cannot be approximated by a small mixture of coordinate-wise product distributions. Quantitatively, if Eq.~\eqref{eq:DLMapproxdist}, written abstractly as:
    \begin{equation*}
    P=\sum_{h=1}^K \lambda_h Q_h
    \end{equation*}
    is a mixture of product distributions satisfying $\TV(P,V)\le \eta$,
    then we necessarily have:
    \begin{equation*}
    \log K
    \ge
    \Omega\!\left(
    \sum_b (r_b-1)
    \right).
    \end{equation*}
    Intuitively, each block completed with $r_b\ge 2$ contributes an independent intra-block parity constraint: the $r_b$ fresh coordinates must lie in an affine parity coset, i.e., satisfy $\oplus_i x_i=a_b$, rather than in a product distribution. A small latent mixture cannot realize many such constraints simultaneously. Quantitatively, if a round completes a collection of blocks with completion sizes $\{r_b\}$, then approximating the corresponding conditional law requires at least $2^{\Omega\left(\sum_b (r_b-1)\right)}$ product components. On the other hand, after conditioning on the visible transcript, the only hidden information is the workspace/CoT coordinates, which gives at most $K=2^{O(s)}$ mixture components. Hence, we have $\sum_b (r_b-1)=O(s)$ in any typical denoising round. In particular, since every block with $r_b\geq 2$ contributes at least $1$ to this sum, only $O(s)$ blocks can be completed by simultaneous multi-bit generation over all (constantly many) rounds. Therefore, when $s=o(M)$, almost every block must fall in the category where $r_b=1$. Hence, we have reduced the task of $\DLM$ approximating $D_{\blk}$ to the task of predicting a parity bit, which will be the basis for our lower bound.
    
    \item \textbf{Reduction of final-bit completion to parity prediction.} Suppose now that a block satisfies $r_b=1$. At its completion round, the $\DLM$ has already generated $B-1$ coordinates and must generate the last one. Under $U_B^\oplus$, this last coordinate is forced:
    \begin{equation*}
    X_{b,j}=\bigoplus_{i\neq j}X_{b,i}.
    \end{equation*}
    Thus, a successful $\DLM$ must predict a parity bit. Technically, if the previous $B-1$ bits generated by the $\DLM$ were uniformly random, then we are done since we can use known lower bounds of $\GCz,\ACz$ computing parity. However,  the first $B-1$ bits were generated by the $\DLM$ itself, while the parity lower bound is for uniform external inputs. This is the most subtle part of the proof, and showing this requires us to account for the entropy of various blocks along the $\DLM$ process and show that there exists at least one block among the $M$ blocks in which the pre-completion $B-1$ bits have large entropy. In this block $b^\star$, we can use the probabilistic method to show that the pre-completion bits are already \emph{close} to uniform, which means that we can replace these bits by a set of independent uniform bits, This changes the execution distribution by only $o(1)$ in total variation, while leaving the final completion round unchanged. After this replacement, the final generated bit predicts the parity of a uniform $B$-bit input with constant advantage. Since the planted execution can be simulated by a shallow $\GCz[\log n]$ circuit, we obtain a shallow circuit that computes parity on $B-1$ bits with constant advantage, contradicting the known parity lower bounds for $\GCz[\log n]$ circuits \cite{kumar2023tight,grewal2024improved}.
    \item \textbf{Handling Remasking/Revision.}  Finally, we extend the argument to bounded remasking/revision. We observe that this case requires a small modification of the entropy arguments. We cannot unfortunately simply restrict our attention to untouched blocks, since the set of touched blocks may be correlated with the generated values. Instead, we account for the entropy argument by adjusting the probability that it is revised after completion. Since there are at most $R$ revision/remasking events, this adds only an $o(M)$ correction for our parameters. In fact, the same win-win and planting argument as before goes through now. Each output-token remasking/revision event touches exactly one coordinate, and hence belongs to one block. Therefore, at most $R$ distinct blocks can be touched by all such events. On the remaining blocks, the execution is equivalent to one without output-token remasking/revision, and the previous argument applies. Restricting to the untouched blocks eliminates revision entirely, and the previous argument applies verbatim. Thus the lower bound survives as long as  $s+R=o(M)$, which indeed motivates our main choice of parameters.
\end{enumerate}

\subsection{Outlook and further directions}\label{sec:outlook}
We see this work as initiating the task of separating quantum computation and classical $\LLM$s. Naturally, there are several questions arising from this work that are worth further investigating. We hope that this serves as a launching pad for substantial future work from the community. 
\begin{enumerate}
    \item \textbf{Strengthening the separations.} Is there a candidate problem for which we can also show $\DLM$-hardness, even in the presence of linear chain-of-thought, and remasking and/or revision? In the case of transformers, we provide one reason why a more typical relational separation (where the input and output length are the same) seems essentially impossible to obtain, but this does not preclude more contrived search-based or relational separations between $\QNCz$ and decoder-only transformers. Finally, since transformers (without chain-of-thought) are known to be inside $\TCz$, it is conceivable that one can obtain a width-$\Omega(n)$ lower bound without that implying a complexity-theoretic breakthrough.
    \item \textbf{Randomized separation.} This work is only concerned with deterministic functions. If we move to randomized problems (where the output needs to be correct with probability $\geq 2/3$, say), can we also show separations, between, e.g., $\land \circ \QNCz$ and transformers?
    \item \textbf{Quantum-augmented $\DLM$s.}  In the model we consider here, recall that new bits are generated using $\ACz$ or $\GCz$ denoisers. What if we used quantum circuits to generate new bits? Could such quantum-augmented $\DLM$s be more powerful than using $\DLM$s with classical-circuit based denoisers? More generally, is there a deeper connection between quantum dynamics and $\DLM$s?\footnote{A recent work of Layden et al.~\cite{layden2025wavefunction} showed that the dynamics of \emph{continuous}-model $\DLM$s can be efficiently simulated on a quantum computer.}
    \item \textbf{Copy complexity of transformers.} As mentioned earlier, a recent work of Grier et al.~\cite{QAC0tc0} shows how to compute $\TCz$ functions using $\QACz$, with the caveat that one needs $O(n^2)$ copies of the input (assumed to be on $n$ bits). Since transformers can be shown to roughly interpolate between $\ACz$ and $\TCz$ (depending on the precision parameter and the type of attention used), it is not inconceivable that computing these functions would require fewer copies for the $\QACz$ circuit. Could this be formalized?
    \item \textbf{Architecture-independent separations.} Our functional and distributional separations rely on different structural limitations of transformers and diffusion language models. Is there a single task solvable by shallow quantum circuits for which every bounded-resource classical language model—whether autoregressive, diffusion-based, or a hybrid of the two—must incur a polynomial cost in at least one of width, chain-of-thought/workspace, sequential generation rounds, or token revision?
\end{enumerate}

\vspace{2mm}
\textbf{Final remarks.} Finally, we make a couple of important remarks on our results to contextualize our work fairly. The separations that we exhibit here are \emph{entirely} theoretical. Our motivation is to understand how quantum algorithms could be better than classical $\LLM$s at certain ad hoc tasks that are \emph{specifically} designed to show these separations. Secondly,  state-of-the-art AI models rely on significant amounts of chain-of-thought~\cite{wei2022chain}, which is empirically known to improve the ability to reason by forcing the $\LLM$ to perform step-by-step ``scratchpad''-style reasoning, which has been shown to yield better accuracy and more expressivity in practice. Our transformer and $\DLM$ lower bounds are proven under the assumption of \emph{no} chain-of-thought and \emph{sublinear} chain-of-thought respectively, which are very strong assumptions; the lower bounds do not hold without these restrictions in place, and indeed it is entirely conceivable that our problems are easy for classical $\LLM$s with a reasonable amount of chain-of-thought. The goal of this work, rather, is to initiate the study of the interplay between quantum circuits and classical machine learning models, which we view as being timely and important in the era of large language models.
%We hope that this work takes a small but important step in the theoretical study of quantum advantage in the era of large language models.

\vspace{2mm}
\textbf{Acknowledgments.} SA thanks Yunchao Liu for pointing the work of~\cite{chen2024theoretical} and  thanks Hongxun Wu for clarifications about their work~\cite{chen2024theoretical}.  
SA thanks ChatGPT 5.5 for help in proving Theorem~\ref{thm:approximable} and the subsequent corollary. 
RS thanks ChatGPT for some minor corrections in the parameterization of the transformer lower bound, which were since then strengthened by hand. The presentation here has been rewritten and checked by us (any further mistakes are our~own).

\section{Preliminaries and Background}\label{sec:prelims}

\paragraph{Notation.} We use base-$2$ logarithms throughout. We also define $[n] := \{1,\ldots,n\}$ for any natural number $n \in \mathbb{N}$ and $[n]_0 := \{0\} \cup [n]$.

\paragraph{Information theory.} We will need some well-known results in information theory, such as Pinsker's inequality, which can all be found in~\cite{cover1991entropy}. For a distribution $P$ on a finite set $\Omega$, its Shannon entropy is defined as:
\begin{equation*}
H(P):=\sum_{x\in\Omega}P(x)\log\frac{1}{P(x)}.
\end{equation*}
It is a fact that $H(\cdot)\leq \log|\Omega|$. 
For a random variable $X$, we write $H(X)$ for the entropy of its law. For
random variables $X,Y$, the conditional entropy is defined as:
\begin{equation*}
H(X\mid Y) := \E_{y\sim Y}\big[H(X\mid Y=y)\big].
\end{equation*}
We will use the standard facts that:
\begin{equation*}
H(X\mid Y)\leq H(X),
\qquad
H(X_1,\ldots,X_m\mid Y)
\leq
\sum_{i=1}^m H(X_i\mid Y).
\end{equation*} 
For two distributions $P,Q$ on the same finite set $\Omega$, their total variation distance is defined as:
\begin{equation*}
\TV(P,Q) := \frac12\sum_{x\in\Omega}|P(x)-Q(x)| = \sup_{A\subseteq \Omega}|P(A)-Q(A)|.
\end{equation*}
We will also use the \emph{data processing inequality} for total variation: if $f:\Omega\to\Omega'$
is \emph{any} map, then:
\begin{equation*}
\TV(f(P),f(Q))\leq \TV(P,Q).
\end{equation*}
Equivalently, taking marginals cannot increase the total variation distance.

We will need a couple of additional lemmas, as follows.

\begin{lemma}
\label{lem:conditional-tv-average}
Let $P$ and $Q$ be distributions on a finite product space $\mathcal T\times\mathcal Y$.  Then, we have:
\begin{equation*}
\mathbb E_{t\sim Q_T}
[\TV\!\left(P_{Y|T=t}, Q_{Y|T=t}\right)] \leq 2\cdot\TV(P, Q),
\end{equation*}
where $Q_T$ is the marginal distribution. 
The same bound also holds when the expectation is taken with respect to $P_T$ in place of $Q_T$. 
\end{lemma}
\begin{lemma}
\label{lem:fannes}
Let $P$ and $W$ be distributions on a set $\mathcal{X}$ with $\TV(P,W)=\epsilon\leq 1/4$. Then, we have
$|H(P)-H(W)|\le 2\epsilon\log|\mathcal{X}|+1.$
\end{lemma}

\begin{proof}
For $\|P-W\|_1 = 2\epsilon\le 1/2$, we have:
\begin{equation*}
    |H(P)-H(W)|\le \|P-W\|_1\log\frac{|\mathcal{X}|}{\|P-W\|_1} = 2\epsilon\log|\mathcal{X}| + 2\epsilon\log\frac{1}{2\epsilon},
\end{equation*}
and $u\log(1/u)\le (\log e)/e < 1$ for $u\in(0,1]$.
\end{proof}

\begin{lemma}\label{lem:pinsker}
If $\mu$ is a
distribution on $\{0,1\}^{r}$ then
$\TV\big(\mu,\mathsf{unif}(\{0,1\}^{r})\big)
\le \sqrt{\,r-H(\mu)\,}.$
\end{lemma}

\begin{proof}
First we observe that for distributions $\mu,\nu$ on a finite set, the total variation distance is related to the KL divergence as $\TV(\mu,\nu)\le\sqrt{\mathrm{KL}(\mu\|\nu)}$. In particular, Pinsker's inequality  together with $\mathrm{KL}(\mu\,\|\,\mathsf{unif})=r-H(\mu)$ gives the statement of the lemma.
\end{proof}

\subsection{Circuits}\label{sec:circuits-prelims}
\paragraph{Classical circuit classes.} The classes $\NCz$, $\ACz$, and $\TCz$ are well-known in complexity theory. $\NCz$ is the class of all Boolean functions computable in a circuit family containing circuits with bounded \emph{fan-in} (i.e., the number of inputs to every gate is $O(1)$ independent of $n$), constant \emph{depth} (i.e., every circuit in the family has depth $O(1)$ independent of $n$, the number of inputs to the circuit itself), and \emph{basis} $\{\land, \lor, \lnot\}$. $\ACz$ is the class of all Boolean functions defined the same way as $\NCz$, except that we allow the circuits to have unbounded fan-in. $\TCz$ is the class of all Boolean functions defined the same way as $\ACz$, except that we allow the circuits to have \emph{majority} gates (gates that output $1$ if and only if a strict majority of their input bits are $1$).

In addition, we are also interested in $\GCz(k)$, a classical circuit class defined recently~\cite{kumar2023tight,grewal2024improved}, with a special type of gate allowed in addition to $\ACz$. For $k \in \mathbb{N}$, a \emph{$G_k$ gate} takes in an unbounded number of bits as input: on input $x$, if $|x|\leq k$, then $G_k(x)$ can be an \emph{arbitrary} function on the input $x$ and if $|x| > k$, then $G_k(x)$ is some constant in $\01$. These are useful as a soft generalization of \emph{threshold} gates.

\paragraph{Quantum circuit classes.} Parallel (and in some cases comparable) to the classical circuit hierarchy is the hierarchy of \emph{quantum circuits}. The two classes of importance for us are $\QNCz$ and $\QACz$. A $\QNCz$ circuit basis consists of all possible one- and two-qubit gates. For the sake of succinctness, we allow the following in particular: Clifford gates, i.e.:
$$H=\frac{1}{\sqrt{2}}\begin{pmatrix}
1 & 1\\
1 & -1
\end{pmatrix}, \qquad S=\begin{pmatrix}
1 & 0\\
0 & i
\end{pmatrix}, \qquad \mathsf{CNOT}=\begin{pmatrix}
1 & 0 & 0 & 0\\
0 & 1 & 0 & 0\\
0 & 0 & 0 & 1\\
0 & 0 & 1 & 0
\end{pmatrix},
$$
and $S,T$-gates, i.e.:
$$
S=\begin{pmatrix}
1 & 0\\
0 & i
\end{pmatrix},\qquad T=\begin{pmatrix}
1 & 0\\
0 & e^{i\pi/4}
\end{pmatrix}
$$
and Pauli matrices, i.e.:
$$\id=\begin{pmatrix}
1 & 0\\
0 & 1
\end{pmatrix}, \qquad X=\begin{pmatrix}
0 & 1\\
1 & 0
\end{pmatrix}, \qquad Y=\begin{pmatrix}
0 & -i\\
i & 0
\end{pmatrix}, \qquad Z=\begin{pmatrix}
1 & 0\\
0 & -1
\end{pmatrix}.$$
We write these gates as above because all Pauli gates and Clifford gates are simulable, and the $T$ gate makes quantum computation universal.

In order to describe $\QACz$, we first observe what the $\mathsf{CNOT}$ gate does. It performs the following operation:
$$
\mathsf{CNOT}:\ket{a,b}\mapsto\ket{a,a\oplus b} \text{ for all }a,b\in \01.
$$
This is a $2$-qubit gate. If $a=1$ it flips the second bit, whereas if $a=0$, it keeps the second qubit as it is. Consider a $\mathsf{CCNOT}$ gate (also referred to as \emph{Toffoli gate}), defined as:
$$
\mathsf{CCNOT}:\ket{a_1,a_2,b}\rightarrow \ket{a_1,a_2,b\oplus a_1\cdot a_2} \text{ for all } a_1,a_2,b\in \01,
$$
i.e., the gate flips the last bit \emph{if and only if} the first $2$ qubits are both $1$.  In order to describe $\QACz$, we allow the circuits to have \emph{generalized} $\mathsf{CNOT}$ (referred to as Toffoli gates) acting on $t\in [n]$ qubits: 
$$
\mathsf{Toffoli}:\ket{a_1,\ldots,a_t,b}\mapsto\ket{a_1,\ldots,a_t,b\oplus a_1\cdots a_t} \text{ for all }a_1,\ldots,a_t,b\in \01,
$$
i.e., the gate flips the last bit \emph{if and only if}  the first $t$ qubits are all $1$. We can now describe the two quantum circuit classes of interest to us.

\begin{enumerate}
    \item $\QNCz$, the set of circuits acting on: $n$ input bits and $\poly(n)$ auxilliary qubits, constant depth, consisting of Paulis, Cliffords, $S$, and $T$ gates.
    \item $\QACz$, the set of circuits acting on: $n$ input bits and $\poly(n)$ auxilliary qubits, constant depth, consisting of Paulis, Cliffords, $T$, and $\mathsf{Toffoli}$ gates. 
\end{enumerate}
The only difference between $\QNCz$ and $\QACz$ is that the latter contains $\mathsf{Toffoli}$ gates, which are not present in the former.

\paragraph{Notation for superconstant depths.} Suppose we have a circuit family $C$ from a circuit class (or quantum circuit class) $\calC$, such that $C$ has depth $O(f(n))$ (i.e., the depth of each circuit in the family is allowed to be a function of the input length, instead of being a universal constant). In that case, we write $C$ belongs to $\calC[f(n)]$. In other words, $\calC[f(n)]$ is the class of circuit families or quantum circuit families whose depth grows as $O(f(n))$, but whose other parameters are defined by the class $\calC$. Now, when $f(n) = \Omega(\log n)$, typically the superscript of the circuit class goes up according to the degree of the polylogarithmic factor; for instance, $\NCz[\log n] = \mathsf{NC}^1$, $\QACz[\log n] = \mathsf{QAC}^1$, $\ACz[\log^2n] = \mathsf{AC}^2$, and so on. However, we will mainly be concerned with the case when $f(n) = o(\log n)$; in that case, we will simply designate the resulting class as $\calC[f(n)]$. In particular, $\QACz[\log \log n]$ will be the class of quantum circuits that are exactly $\QACz$ circuits, except that their depth can go up to $O(\log \log n)$.

\subsection{Transformer models}
\label{sec:transformer-prelims}
We first describe in the figure below the structure of a transformer model. In the remaining part of this section, we define various aspects of this transformer architecture and discuss how one goes from the input to the output.
\begin{figure}[ht]
    \centering
    \includegraphics[width=\linewidth]{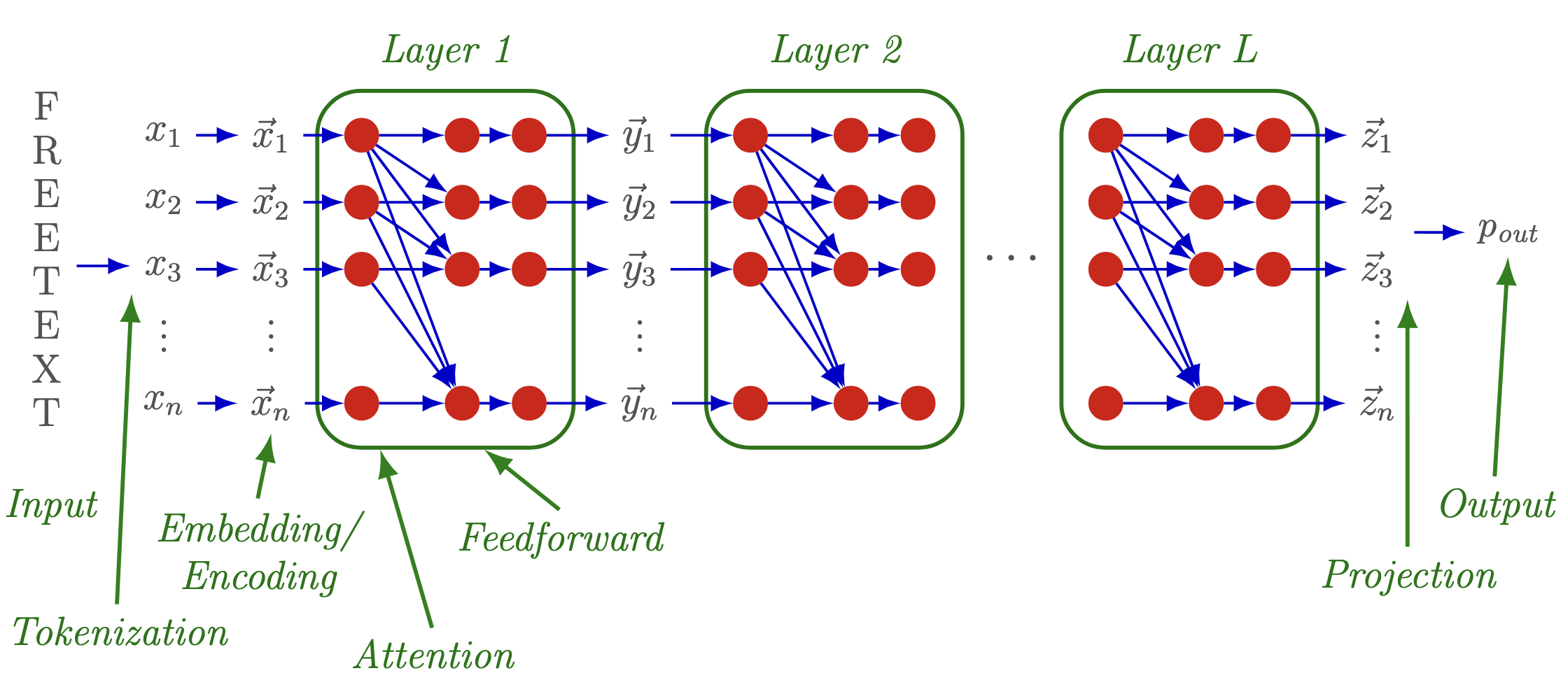}
    \caption{An idealized depiction of the core transformer architecture.}
    \label{fig:transformer-diagram}
\end{figure}
\subsubsection{Definitions and Architecture}
We now give a brief introduction to the basic transformer architecture. Apart from the choice of the attention mechanism itself (which is typically taken to be a variant of \emph{softmax}), there are a few other parameters that define a transformer:
\begin{enumerate}
    \item $L$: the number of layers
    \item $H$: the number of attention heads in each layer
    \item $d$: the embedding dimension
    \item $p$: the bit precision of entries in the internal computations.
\end{enumerate}
The \emph{context length} $n$, which corresponds to the length of the input measured in \emph{tokens}, is \emph{not} a parameter of the transformer; this is what enables a transformer to carry out its computations on unbounded-length inputs. We often refer to $L$ as the \emph{depth} of the transformer, and to the product $Hdp$ as its \emph{width} or \emph{size}. Recent works~\cite{sanford2023representational,chen2024theoretical} have adopted the convention of designating a transformer as ``small'' if $Hdp = n^{o(1)}$ and ``large'' if $Hdp = n^{\Omega(1)}$.

A transformer with $H$ heads and $L$ layers computes a length-preserving sequence-to-sequence function. It starts by breaking up some input text (the \emph{prompt}) into a length-$n$ sequence $w = w_1\ldots w_n$ of \emph{tokens}, and then embeds each token $w_i$ into a $(dH)$-dimensional real vector $x_i^{(0)}$, creating the $n$-sequence:
\begin{equation*}    \mathbf{x}^{(0)}=\Big(x_1^{(0)},\ldots,x_n^{(0)}\Big) \in \Big(\mathbb{R}^{dH}\Big)^n
\end{equation*}
where each entry is represented using at most $p$ bits. Subsequently, over the course of a sequence of length-preserving internal computations through $L$ \emph{layers} in its architecture, the transformer converts the original sequence into the sequence:
\begin{equation*}
\mathbf{x}^{(L)}=\Big(x_1^{(L)},\ldots,x_n^{(L)}\Big) \in \Big(\mathbb{R}^{dH}\Big)^n.
\end{equation*}
Each transformer layer contains $H$ \emph{attention heads}. The computation in the $\ell$-th layer is determined by \emph{query}, \emph{key} and \emph{value} matrices indexed by their corresponding attention head, explicitly given by $Q^{\ell,h},K^{\ell,h},V^{\ell,h} \in \R^{d\times dH}$ for $\ell \in [L], h\in [H]$. These matrices are all determined (i.e., learned) during training time, and are subsequently considered a complete specification to the transformer architecture (which, in this model, carries out a deterministic computation on each given input). We now describe the transition in between the layers of the transformer: if the input to the $\ell$-th layer is $\mathbf{x}^{(\ell-1)}$, then the output $\mathbf{x}^{(\ell)}$ is determined as follows: for every position $i\in [n]$, and attention head $h \in [H]$, define:
\begin{equation*}
    \alpha^{(\ell,h)}_{i,j} := \frac{1}{Z_i}\cdot \exp(\langle  Q^{(\ell,h)} x_i^{(\ell-1)}, K^{(\ell,h)} x_j^{(\ell-1)}\rangle),
\end{equation*}
where $Z_i := \sum_{j}\exp(\langle x_i^{(\ell)} Q^{(\ell,h)}, K^{(\ell,h)} x_j^{\ell}\rangle)$. Note that this is where we have used the currently-used \emph{softmax} attention, instead of other variants such as \emph{unique hard} attention and \emph{average hard} attention that have been theoretically considered as well. Then, set:
\begin{align}
\label{eq:encoderdecodertransformer}
    y_i^{(\ell,h)} := \sum_j \alpha^{(\ell,h)}_{i,j} V^{(\ell,h)} x_j^{(\ell)} \in \R^d.
\end{align}
Now, we can concatenate the vectors $y_i^{(\ell)} := (y_i^{(\ell, 1)}, \ldots, y_i^{(\ell, H)}) \in \R^{dH}$, and then set:
\begin{equation*}
    \mathbf{y}^{(\ell)} := \Big(y_1^{(\ell)},\ldots,y_n^{(\ell)}\Big) \in \Big(\R^{dH}\Big)^n.
\end{equation*}
Finally, the output $\mathbf{x}^{(\ell)}$ of the $\ell$-th layer is obtained via applying an \emph{arbitrary} (learnable) function $g^{\ell}:\R^{dH}\rightarrow \R^{dH}$ coordinate-wise to $\mathbf{y}^{(\ell)}$, i.e.:
\begin{equation*}
\mathbf{x}^\ell :=\Big(g^\ell(y_1^{(\ell)}), \ldots, g^\ell(y_n^{(\ell)})\Big).
\end{equation*}
Usually, the transformer takes the vector $x_n^{(L)} \in \R^{dH}$, applies a linear transformation to it along with a sigmoid function to obtain a probability $p_{\text{out}} \in [0, 1]$, and then rounds this probability $p_{\text{out}}$ to a value in $\{0, 1\}$. In this sense, a transformer can be seen as a \emph{language recognizer} over its input alphabet, and this is what enables us to place transformers inside or outside standard complexity classes.

\paragraph{Decoder transformers.} 
The definition above is for a general transformer. We say a transformer is a \emph{decoder-only transformer} if the sums in Eq.~\eqref{eq:encoderdecodertransformer} and the definition of $Z_i$ are only over all $j$ with $j< i$. In other words, the values of $y_i^{(\ell,h)}$ and $Z_i$ only depend on all the coordinates of $\mathbf{x}$ \emph{until} coordinate $i$, i.e., only on $x_1^{(\ell)},\ldots,x_{i-1}^{(\ell)}$. This requirement for each position to only attend to positions before it is known as \emph{future masking}, and it is what enables the autoregressive capabilities of modern $\LLM$s.

\paragraph{Chain-of-thought.} 
In autoregressive decoder models, the architecture is equipped with the ability to output intermediate tokens during its computation process, which are then fed back to the architecture by appending them to the (increasingly growing) input. This process is called \emph{chain-of-thought} (CoT); it is well known that CoT-equipped transformers are strictly more powerful in terms of expressivity than transformers without CoT. Indeed, most modern LLMs use CoT-style intermediate reasoning.

\subsubsection{Complexity Class Relations}
As we saw, viewing transformers as language recognizers enables us to compare their capabilities with standard complexity classes. In this section, we provide a brief summary of the capabilities of \emph{decoder-only} transformers. Note that encoder-only transformers are theoretically strictly more powerful than decoder-only ones, but recent $\LLM$s focus almost exclusively on decoder-only models, and so we will restrict ourselves to that model. We note that modern architectures allow intermediate tokens generated during the computation by a process called \emph{chain-of-thought} (CoT). Allowing $\ell$ intermediate tokens on an input of length $n$, where $\ell\in \{0,\ldots,\log n\}$ does not increase the complexity class, but increasing $\ell$ further can certainly improve the expressivity of these transformers (in fact, in the limit, arbitrary amount of chain-of-thought lets transformers Turing-complete). So, the same table below applies even for chain-of-thought length up to $\ell=O(\log n)$~\cite{li2024chain}.
\renewcommand{\arraystretch}{1.35}
\begin{table}[!ht]
\begin{center}
\begin{tabular}{|c|c|c|}
\hline
\textbf{Precision} & \textbf{Model dimension} & \textbf{Complexity class} \\
\hline
$O(1)$ & $O(\log n)$ & $\subseteq \ACz$ \\
\hline
$O(1)$ & $\mathrm{poly}(n)$ & $= \ACz$ \\
\hline
$O(\log n)$ & $O(\log n)$ & $\subseteq \TCz$ \\
\hline
$O(\log n)$  & $\mathrm{poly}(n)$ & $= \TCz$ \\
\hline
\end{tabular}
\caption{Summary of main results of~\cite{li2024chain} which simulates the transformer class with the corresponding precision and model dimension in terms of standard circuit complexity classes.}
\end{center}
\end{table}

\subsubsection{Lower Bound via Communication Complexity}
In a recent work, Chen et al.~\cite{chen2024theoretical} introduced an \emph{autoregressive communication model}, and showed that any unconditional lower bound on the communication complexity of a particular problem in this model yields lower bounds on the width $Hdp$ of any constant-depth transformer designed to solve the same problem. The {communication} model is defined formally as follows. There are $L+2$ players (denoted players $-1,0,1,\ldots,L$). For $-1 \leq i \leq L$, the input $z_i$ for player $i$ is describable using $m_i$ tokens. There is a (directional and restricted) communication protocol for these players, which takes place over a certain number of \emph{epochs} and a fixed \emph{bandwidth} $B$, as follows. Let $X^{(0)}_\ell := z_\ell$, for $\ell \in \{-1,\ldots,L\}$, for $\ell\geq 1$, they proceed as follows. Every player $i\in \{-1,\ldots,L\}$ runs the following:

\begin{enumerate}
    \item  For $j>i$, player $i$ sends a message $\Gamma^{(\ell)}_{ij}$ to player $j$, depending solely on its current state $X^{(\ell)}_i$.
    \item Player $j$ replies back to player $i$ with a message $\Pi^{(\ell)}_{ij}$ that satisfies two constraints:
    \begin{inparaenum}[$(i)$]
    \item the message $\Pi^{(\ell)}_{ji}$ depends solely on $\Gamma^{(\ell)}_{ij}$ and player $j$'s  current state $X^{(\ell)}_j$;
    \item the size of the message satisfies $|\Pi^{(\ell)}_{ji}| \leq B\cdot m_{i}$.
    \end{inparaenum}
    \item Upon receiving the messages $\{\Pi^{(\ell)}_{ji}\}_{j>i}$, player $i$ updates its state to $X^{(\ell+1)}:=X^{(\ell)}_i\cup \bigcup_{j>i}\Gamma^{(\ell)}_{ij}$.
\end{enumerate}

After $L$ rounds, player $-1$ should output the first bit of $X^{(L)}$ in order for the protocol to be deemed successful.

We now point out two nuances about this protocol.
\begin{itemize}
    \item \textbf{\emph{Memorylessness.}} Suppose during epoch $\ell$, players $i,j$ interacted by exchanging the messages $\Gamma^{(\ell)}_{ij},\Pi^{(\ell)}_{ji}$, and so player $i$ updated its current state $X^{(\ell+1)}$ by including the transcript message $\Pi^{(\ell)}_{ji}$. For a subsequent epoch, say $\ell+1$, by the protocol requirements, player $j$ is \emph{forced to forget} the previous exchange altogether. Since player $j$ did not update its local state $X^{(\ell)}\rightarrow X^{(\ell+1)}$ using $\Gamma^{(\ell)}_{ij},\Pi^{(\ell)}_{ji}$ and all messages solely depend on their local state, players are therefore forced to forget some information between epochs.
    \item \textbf{\emph{Local one-way-ness.}} Observe that in each epoch, when player $i$ talked to player $j$, by the protocol requirements, player $j$ cannot reply back to player $i$ based on its interaction with the players $\{-1,\ldots,j-1\}$. In any given epoch, player $j$ cannot respond to any of the players by ``mixing" the input messages it received. Instead, it can only reply based on its current state $X^{(\ell)}_j$ and $\Pi^{(\ell)}_{ij}$.
\end{itemize}
The main result proven by Chen et al.~\cite{chen2024theoretical} is the following: they define the \emph{iterated index} function, called $\IIndex$ (which we formally define later in Section~\ref{sec:itreindx}), and show how to give lower bounds on the width of constant-depth decoder-only transformers computing $\IIndex$, using the communication protocol defined above. In particular, they show the following result.
\begin{lemma}\label{lem:cpw-reduction}
If there is an $L$-layer decoder-only transformer with width $Hdp$ that solves $\mathsf{IterInx}$ for all instances of the problem, then there is a deterministic autoregressive protocol on $L + 2$ players that solves $\IIndex$ in $L$ epochs with bandwidth $2Hdp$.
\end{lemma}

\subsection{Diffusion models}\label{sec:dlms-prelims}
We first describe in the figure below the structure of a diffusion language model ($\DLM$). In the remaining part of this section, we define various aspects of this $\DLM$ architecture and discuss how one samples a distribution using this model.
\begin{figure}[!ht]
    \centering
    \includegraphics[width=\linewidth]{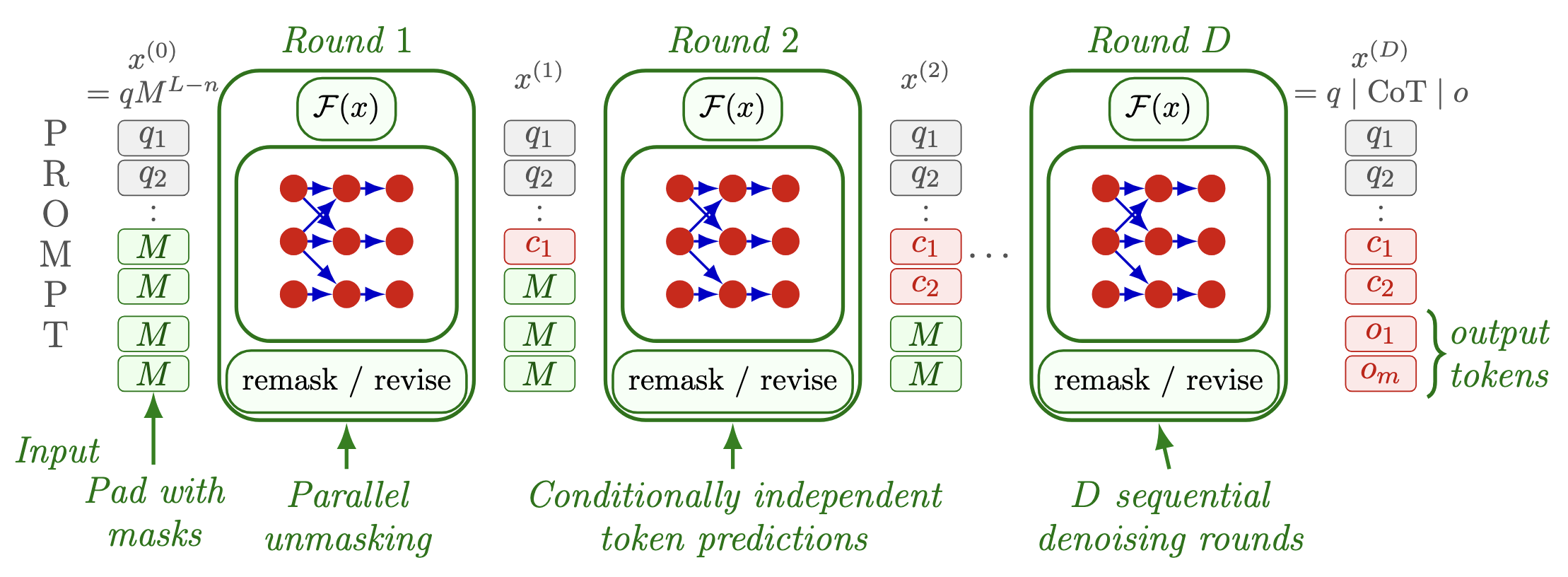}
    \caption{An idealized depiction of the core $\mathsf{DLM}$ architecture.}
    \label{fig:dlm-diagram}
\end{figure}

We now describe the diffusion language model ($\DLM$) framework used in this paper. In abstraction, a $\DLM$ is defined as a generative model for text~\cite{nie2025large,song2025seed} that starts from a corrupted (masked or noisy) sequence of tokens, and iteratively denoises it. It is used typically to generate samples from a learned distribution over text; more precisely, it learns one distribution within a parameterized class of distributions that approximates the data.  Informally, the goal is as follows. Let:
\begin{equation*}
\{D_u:\{0,1\}^n\to[0,1]\}_u
\end{equation*}
be a family of target distributions. The goal of the $\DLM$ is the following task: on input $u$, generate a sampled output string $y\sim D_u$. We also study \emph{approximate} generation, where the $\DLM$ is allowed to output a sample from a distribution $\widetilde D_u$ satisfying:
\begin{equation*}
\TV(\widetilde D_u,D_u) \leq \varepsilon.
\end{equation*}
More formally, we present the usual convention of $\DLM$s in literature (similar to the formalization given in \cite{jiang2025diffusion}). 
Let \(V\) be a finite token vocabulary (for simplicity we let it be $\{0,1\}$ here) and let $\MASK$ denote the so-called \emph{mask token}.  A state in a $\DLM$  algorithm is a sequence:
\begin{equation*}
    x\in (V\cup\{\MASK\})^L.
\end{equation*}
A coordinate equal to $\MASK$ is called \emph{masked} or \emph{unresolved}, while a coordinate in $\Sigma$ has already been assigned a token and is called \emph{resolved}. 
A noiseless sequence is an element of \(V^L\). A \emph{predictor} is a map which, for every partially masked sequence \(x\in (V\cup\{\MASK\})^L\), specifies a distribution:
\begin{equation*}
    p(\cdot\mid x)\in \Delta(V^L),
\end{equation*}
i.e., the coordinate-wise conditional distribution of \(z\sim p(\cdot\mid x)\) is given by:
\begin{equation*}
    p(z\mid x)=\prod_{i=1}^L p^i(z^i\mid x).
\end{equation*}
In the standard model, already unmasked coordinates are kept fixed, i.e., \(p^i(z^i=x^i\mid x)=1\) whenever \(x^i\in V\).

\subsubsection{Masked states}

The structure of a usual $\DLM$ algorithm is as follows: it begins from a partially masked state:
\begin{equation*}
\Big(
\underbrace{x}_{\text{input}},
\underbrace{\MASK^{m-n}}_{\text{workspace/CoT}},
\underbrace{\MASK^n}_{\text{output}}
\Big),
\end{equation*}
and progressively resolves masked coordinates over several denoising rounds (which we define below). After a fixed number of rounds, it outputs the final $n$ coordinates, which are designated as output coordinates. The goal is for the distribution of the final $n$ coordinates to be ``close'' to the target distribution $D_x$. The remaining coordinates are auxiliary workspace coordinates that may be used during the computation and discarded at the end.

We use the term \emph{workspace token} for any auxiliary coordinate generated during the $\DLM$ computation but discarded before the final output is read. In the $\DLM$ literature, such coordinates are often interpreted as chain-of-thought (CoT) tokens when they serve as intermediate reasoning or scratchwork. Since our lower bounds only use the formal role of these coordinates as auxiliary generated tokens available to later rounds, we use the terms \emph{workspace} and \emph{chain-of-thought} interchangeably when no distinction is needed.

\subsubsection{Denoising process}

We now describe the denoising process. This process is a iterative procedure through which masked coordinates are progressively resolved. Suppose the current state of the $\DLM$  at round $t$ is:
\begin{equation*}
X^{(t)}\in(\Sigma\cup\{\MASK\})^m.
\end{equation*}
A single denoising round consists of two conceptually distinct operations:
\begin{enumerate}
\item \emph{scheduling/unmasking}, where the model decides which masked coordinates will be updated in the current round;
\item \emph{token generation}, where the model generates values for the selected coordinates.
\end{enumerate}
We describe both operations below in detail.

\paragraph{Scheduling/unmasking.} 
At round $t$, the model selects a subset
$R_t(X^{(t)})\subseteq[m]$
of currently masked coordinates to reveal. Equivalently, one may think of the scheduler as outputting revealed bits
$S_{t,i}(X^{(t)})\in\{0,1\}$,
where
$S_{t,i}(X^{(t)})=1$ 
indicates that the coordinate $i$ will be updated during round $t$. The revealed set is therefore:
\begin{equation*}
R_t(X^{(t)})
=
\{i:S_{t,i}(X^{(t)})=1\}.
\end{equation*}
The model then samples, independently for all $i\in S_t$:
\begin{equation*}
    X^{(t+1)}_i \sim p_t^i(\cdot\mid X^{(t)}),
\end{equation*}
and leaves all other coordinates unchanged. A key feature of $\DLM$s is that the reveal schedule may depend adaptively on the partially generated sequence so far. Thus, the model does not need to decide the entire reveal order in advance; instead, future reveal decisions may depend on tokens generated in earlier rounds, including both workspace/CoT coordinates and output coordinates.  
An important aspect to consider here is how \emph{complex} the scheduling is. As is common in the $\DLM$ literature, we require that the scheduling rule itself be computationally shallow. Concretely, for each coordinate $i$ and round $t$, the reveal decision $S_{t,i}(X^{(t)})$ is computed by a constant-depth polynomial-size circuit.

\paragraph{Token generation.}
Once the reveal set
$R_t(X^{(t)})$
has been selected, the model generates token values at those coordinates. The crucial and defining structural property of a standard $\DLM$ is the following notion of conditional independence:

\begin{quote}
\emph{Conditioned on the current state, newly revealed coordinates are sampled independently.}
\end{quote}
More precisely, conditioned on the current partially generated state $X^{(t)}$, the model specifies marginal distributions $p_{t,i}(\cdot\mid X^{(t)})$ for all $i\in R_t(X^{(t)})$, and samples:
\begin{equation*}
X_i^{(t+1)} \sim p_{t,i}\left(\cdot\mid X^{(t)}\right)
\end{equation*}
independently across coordinates. Thus, the conditional law of the newly generated coordinates factorizes as:
\begin{equation*}
\mathcal L\left(
(X_i^{(t+1)})_{i\in R_t}
\;\Big|\;
X^{(t)}
\right)
=
\bigotimes_{i\in R_t}
p_{t,i}\left(\cdot\mid X^{(t)}\right).
\end{equation*}
Intuitively, conditioned on the current partially generated sequence, a standard $\DLM$ generates fresh coordinates in parallel using coordinate-wise denoising rules.

\paragraph{Execution as a branching process.}
It is useful to view a $\DLM$ execution as a branching process over partially generated states. Each possible state $\tau=X^{(t)}$ 
at the beginning of round $t$ defines a transcript branch of the computation. Conditioned on reaching a fixed partially generated state $\tau$:
\begin{itemize}
\item the reveal set $R_t(\tau)$ is fixed,
\item the next batch of generated coordinates has the product law:
\begin{equation*}
\mathcal L\left(
(X_i^{(t+1)})_{i\in R_t(\tau)}
\;\Big|\;
X^{(t)}=\tau
\right)
=
\bigotimes_{i\in R_t(\tau)}
p_{t,i}(\cdot\mid \tau).
\end{equation*}
\end{itemize}
We write \(p_t^i(\cdot\mid \tau)\) for the \(i\)-th marginal of the $\DLM$ predictor
$p_t(\cdot\mid \tau)$. Although the predictor specifies a full product
distribution on $V^L$, only the coordinates selected by $R_t(x)$ are actually updated in the current denoising round.
Thus, conditioned on a fixed partially generated state, the next round of generation is a product distribution over the newly revealed coordinates. Globally, however, the overall output distribution may still exhibit highly nontrivial correlations because future rounds depend adaptively on previously generated tokens. Consequently, correlations may accumulate across rounds through the interaction of adaptive scheduling and iterative denoising.

\subsubsection{Constant-round shallow DLMs}
As discussed before, a natural question in this set-up is \emph{how complex is the
sample generation?} In particular, one must also account for the circuits implementing the scheduling and token-generation steps. As is standard in
the $\DLM$ literature, these steps are required to be computationally shallow. Concretely, we assume that the number of denoising rounds is constant, i.e., $T=O(1)$, and that for every round $t$ and coordinate $i$, both the scheduling rule $S_{t,i}(X^{(t)})$ and the denoising marginal $p_{t,i}(\cdot\mid X^{(t)})$
are computable by constant-depth polynomial-size circuits. The work by Jiang et al.~\cite{jiang2025diffusion} studies this model with $\ACz$ circuits, proving both positive and negative results. Since our focus is on hardness, we prove our results even when these shallow circuits are allowed to be $\GCz$ circuits, effectively strengthening the lower bounds. Because the number of rounds is constant, the entire generation process can be viewed as a constant-depth adaptive computation built from
repeated applications of shallow scheduling and coordinate-wise denoising steps.

\subsubsection{Remasking and revision.}
The basic denoising model described above only allows masked coordinates to be resolved over time. Modern diffusion-style language models often consider stronger variants in which previously generated coordinates may later be modified or revisited. Two common mechanisms are \emph{revision} \cite{song2025seed} and \emph{remasking} \cite{nie2025large}. Intuitively speaking, in a remasking step, previously generated coordinates may be returned to the masked state and regenerated in later rounds, whereas in a revision step, the model is allowed to overwrite or update tokens that were generated in earlier~rounds. Formally, in a remasking step, after the denoising update, the model applies a (possibly randomized) remasking policy:
\begin{equation*}
G_t:(V\cup\{\mathsf M\})^L\to 2^{[L]}.
\end{equation*}
For $T_t=G_t(x)$, the coordinates $i\in T_t$ are reset to $\MASK$. In the revision model, we relax the standard condition that an unmasked coordinate must remain unmasked during the algorithm. Thus, $R_t(x)$ may
include coordinates $i$ with $x^i\in V$, and the predictor marginal $p_t^i(\cdot\mid x)$ may assign some probability less than $1$ to the current
token $x^i$. The coordinate-wise conditional independence condition is still maintained:
\begin{equation*}
p_t(z\mid x)=\prod_{i=1}^L p_t^i(z^i\mid x).
\end{equation*}

These mechanisms allow the model to iteratively refine earlier generations rather than committing permanently to the first generated value. Conceptually, remasking and revision substantially strengthen the expressive power of the model: the generation process is no longer monotone in the set of revealed coordinates, and tokens generated early in the computation may later be ``corrected'' based on information obtained in subsequent rounds. We will prove our main lower bound in the presence of both remasking and revision, making it explicit in our lower bounds.\footnote{We remark that \cite{jiang2025diffusion} accounts for the cost of remasking/revision differently from ours: they allow arbitrarily many tokens to be remasked or revised in any particular round, and only bound the total number of \emph{rounds of revision}. Here, we account for the total number of revision/remasking steps throughout the entire $\DLM$ computation.}
\section{Functional separation: quantum versus transformer models}\label{sec:transformers}

\subsection{Motivation for our candidate problem}
In this section, we first give some motivation for our candidate problem that we use to separate quantum shallow circuits from transformer models. The quantum circuits that we use for these motivating functions will be eventually used when discussing our separation.
\subsubsection{Index function}
Our first observation is that the \emph{index} function, denoted by $\Index$ and defined below, is computable with a shallow quantum circuit with minimal classical post-processing without requiring the full power of $\QACz$ circuits, albeit the caveat that it requires copies of the input:
\begin{equation*}
\Index: \{0,1\}^n\times \{0,1\}^N\rightarrow \{0,1\}, \text{ such that } \Index(i,x)=x_i
\end{equation*}
for $x\in \{0,1\}^N$, where $N=2^n$. Formally, we have the following result.
\begin{claim}\label{claim:index_with_copies}
Let $n, N \in \mathbb{N}$ such that $N = 2^n$. Given $N$ copies of the address $i \in \{0,1\}^n$ and one copy of $x \in \{0,1\}^N$, the index function $\Index(i,x)$ can be computed by a quantum circuit with gate complexity $O(N \log N)$ followed by an $\textsf{OR}$ gate. Moreover, the problem is in $\lor \circ \QNCz[{\log \log N}]$.
\end{claim}
Before we prove the result above, we introduce some simple quantum circuits that will be used for the claim, as well as several subsequent arguments.

\paragraph{Quantum circuits for equality checks.}
We will often require a quantum subroutine for computing an \emph{equality check} of two strings $y,z \in \{0,1\}^n$, which we defined as $[y = z]$. We first present a very simple circuit to compute the equality check.
\begin{claim}\label{claim:EQ_circ1}
Let $n \in \mathbb{N}$. Given one copy of $y,z \in \{0,1\}^n$, the equality check $[y=z]$ can be computed by a quantum circuit over $(3n+1)$ qubits with gate complexity $O(n)$ and depth $O(\log n)$.
\end{claim}
\begin{proof}
We define a quantum circuit $\calA^j$ with the following action:
\begin{equation}\label{eq:action1}
\ket{y}\ket{z}\ket{0}^n\ket{0} \xrightarrow{\calA^j} \ket{y}\ket{z}\left(\otimes_{\ell \in [n]} \ket{[y_\ell = z_\ell]}\right) \ket{[y = z]}.    
\end{equation}
This is implemented as followed. Let the $n$-ancilla qubit register be denoted by $s$, and the $\ell$-th qubit by $s_\ell$. Let the last ancilla qubit be denoted by $t$. For each $\ell \in [n]$, we apply a $3$-qubit Toffoli gate with controls over $\ket{y_\ell}$, $\ket{z_\ell}$, with the target being the qubit $s_\ell$. This sets the state of the $n$-ancilla register to $\otimes_{\ell \in [n]} \ket{[y_\ell = z_\ell]}$. We then apply an $(n+1)$-qubit Toffoli gate with controls over all of $s$, with the target being qubit $t$. The final state is then as given in Eq.~\eqref{eq:action1}. 

The circuit $\calA^j$ is over $3n+1$ qubits. All the $3$-qubit Toffoli gates can implemented in parallel and with gate complexity of $O(1)$ each (using $O(1)$ additional ancilla qubits for each). The $(n+1)$-qubit Toffoli gate can be implemented with $O(n)$ one- and two-qubit gate complexity, and depth $O(\log n)$ (using $O(n)$ additional ancillas). Overall, the gate complexity is $O(n)$ and the depth is $O(\log n)$. This completes the proof.
\end{proof}

The construction above can be modified to give a circuit requiring fewer qubits as input, and which we will primarily use.
\begin{claim}\label{claim:EQ_circ2}
Let $n \in \mathbb{N}$. Given one copy each of $y,z \in \{0,1\}^n$, the equality check $[y=z]$ can be computed by a quantum circuit over $(n+1)$ qubits with gate complexity $O(n)$ and depth $O(\log n)$.
\end{claim}
\begin{proof}
We first note that for a fixed $z \in \{0,1\}^n$, we have $[y=z] = \bigwedge_{\ell \in [n]} [y_\ell = z_\ell]$. For all $\ell \in [n]$, the indicator function $[y_\ell = z_\ell]$ can be written as:
\begin{equation}\label{eq:y_z_indicator}
[y_\ell = z_\ell] = 
\begin{cases}
y_\ell, & z_\ell = 1, \\
\overline{y_\ell}, & z_\ell = 0.
\end{cases}    
\end{equation}
Define the (smaller) quantum circuit $\calA^j$ as having the following action:\footnote{Note that we do not use a copy of $z$ as a quantum input, unlike in Claim~\ref{claim:EQ_circ1}.}
\begin{equation*}
\ket{y}\ket{0} \xrightarrow{\calA^j} \ket{y}\ket{[y = z]}.
\end{equation*}
This is implemented as follows. We apply a single-qubit $X$ gate on the qubits of $\ket{y}$ for which $z_\ell = 0$. This is in order to implement $[y_\ell = z_\ell]$ using Eq.~\eqref{eq:y_z_indicator}. We then implement an $(n+1)$-qubit Toffoli gate with controls over the qubits of $\ket{y}$ with the target being the last qubit. We then reset the qubits of $\ket{y}$ by applying a single-qubit $X$ gate on those qubits of $\ket{y}$ for which $z_\ell = 0$ again. Let us denote this last qubit as $\ket{t}$. The circuit is then summarized below in Figure~\ref{fig:circ_Aj_1}.
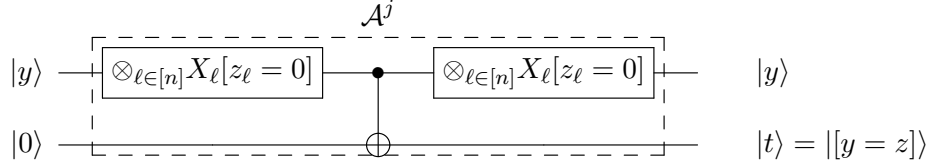
\begin{figure}[ht]
\centering
$$
\Qcircuit @C=1.5em @R=1.2em {
& & \mbox{$\calA^j$} & & \\
\lstick{\ket{y}} & \gate{\otimes_{\ell \in [n]} X_\ell [z_\ell = 0]} & \ctrl{1} & \gate{\otimes_{\ell \in [n]} X_\ell [z_\ell = 0]} & \qw & \rstick{\ket{y}} \\
\lstick{\ket{0}} & \qw & \targ & \qw & \qw & \rstick{\ket{t} = \ket{[y = z]}} \gategroup{2}{2}{3}{4}{.7em}{--} 
}
$$
\caption{Illustration of circuit $\calA^j$ for equality check $[x = y]$.}
\label{fig:circ_Aj_1}
\end{figure}

The circuit $\calA^j$ is applied over $n+1$ qubits accounting for the size of $y$ and the additional ancilla qubit. The circuit $\calA^j$ used $O(n)$ many single-qubit gates applied in parallel to each other and one $(n+1)$-controlled Toffoli gate. The latter has a one- and two-qubit gate complexity of $O(n)$ and requires depth $O(\log n)$ (when given $O(n)$ clean ancillas). This completes the proof.
\end{proof}

\paragraph{Quantum circuit (with copies) for computing index.} 
We are now ready for the proof of Claim~\ref{claim:index_with_copies}.
\begin{proof}[Proof of Claim~\ref{claim:index_with_copies}]
We first observe that the index function can be written as a depth-$3$ DNF circuit as follows:
\begin{equation*}
\Index(i,x)=\bigvee_{j\in \{0,1\}^n} [i=j]\land x_j =\bigvee_{j\in \{0,1\}^n} \left(x_j \land \bigwedge_{\ell\in [n]} [i_\ell=j_\ell]  \right). 
\end{equation*}
We now describe a quantum circuit to compute the index function. Suppose we have $N$ copies of $i$, i.e., we are given $N$ copies of the input $(i,i,i,\ldots,i)$. For one copy of $i$ and a fixed $j \in \{0,1\}^n$, we will consider the quantum circuit $\calA^j$ from Claim~\ref{claim:EQ_circ2} (with instantiations of $y$ as $i$, $z$ as $j$ there) and slightly modify it to output $x_j \wedge [i_j = j]$ (see Figure~\ref{fig:circ_Aj_2}), which has the following action:
\begin{equation*}
\ket{i}\ket{x_j}\ket{0} \xrightarrow{\calA^j} \ket{i}\ket{x_j}\ket{x_j \land [i = j]}.
\end{equation*}
Let us denote the last qubit as $\ket{t_j}$ for a given $j$.
\begin{figure}[ht]
\centering
$$
\Qcircuit @C=1.5em @R=1.2em {
& & \mbox{$\calA^j$} & & \\
\lstick{\ket{i}} & \gate{\otimes_{\ell \in [n]} X_\ell [j_\ell = 0]} & \ctrl{1} & \gate{\otimes_{\ell \in [n]} X_\ell [j_\ell = 0]} & \qw & \rstick{\ket{i}} \\
\lstick{\ket{x_j}} & \qw & \ctrl{1} & \qw & \qw & \rstick{\ket{x_j}}\\
\lstick{\ket{0}} & \qw & \targ & \qw & \qw & \rstick{\ket{t_j} = \ket{x_j \land [i_j = j]}} \gategroup{2}{2}{4}{4}{.7em}{--} 
}
$$
\caption{Illustration of quantum circuit $\calA^j$ used for equality check $[i_j = j]$ and outputting the corresponding bit $x_j$ if true.}
\label{fig:circ_Aj_2}
\end{figure}
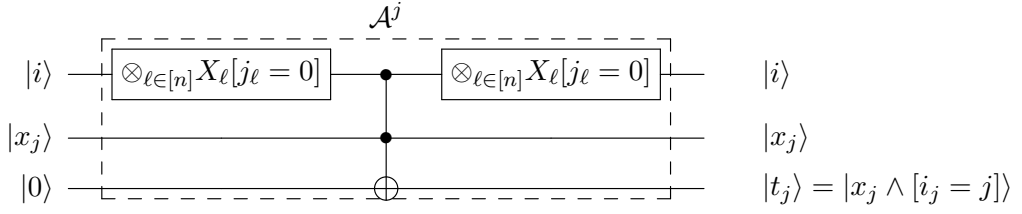

To compute the index function, we now apply the circuit $U$ composed of the smaller circuits $\calA^j$ for $j \in \{0,1\}^n$ in parallel as follows:
\begin{equation*}
\left(\ket{i}\ket{x_0}\ket{0}\right) \ldots \left(\ket{i}\ket{x_{2^n-1}}\ket{0}\right) \xrightarrow{\calA^{0} \otimes \ldots \otimes \calA^{2^n-1}} \left(\ket{i}\ket{x_0}\ket{t_0}\right) \ldots \left(\ket{i}\ket{x_{2^n-1}}\ket{t_{2^n-1}}\right),
\end{equation*}
where we have used the equivalent integer representation of the index $j$ when denoting $x_j$ and $\calA^j$. We now note that each $\ket{t_j}$ is either $\ket{0}$ or $\ket{1}$. We simply measure the qubits over all $\ket{t_j}$ in the computational basis, and obtain the measurement outcome $t_j = x_j \land [i = j]$. Applying an \textsf{OR} gate over all $j \in \{0,1\}^n$ then gives us:
\begin{equation*}
\bigvee_{j\in \{0,1\}^n} t_j  = \bigvee_{j\in \{0,1\}^n} [i=j]\land x_j,
\end{equation*}
which is the desired index function computation. 
\begin{figure}[h!]
\centering
\includegraphics[width=0.6\linewidth]{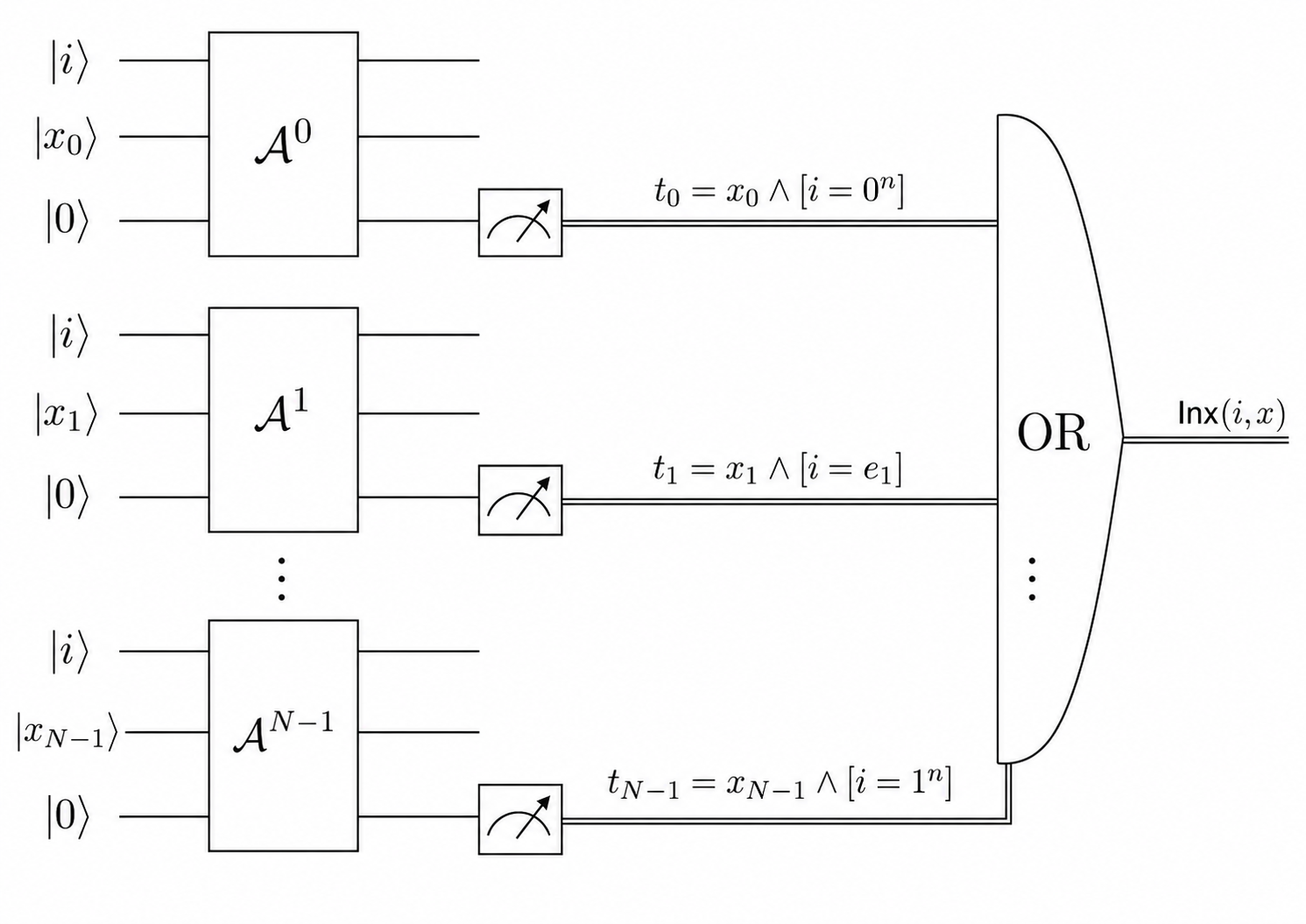}
\caption{Illustration of the circuit used for outputting $\textsf{Inx}(i,x)$ using copies of $\ket{i}$ and the circuit of $\calA^j, \forall j \in [N-1]_0$ from Figure~\ref{fig:circ_Aj_2}.}
\label{fig:circ_Inx}
\end{figure}
For each $\calA^j$, we used $O(n)$ many single-qubit gates applied in parallel to each other and one $(n+1)$-controlled Toffoli gate (Claim~\ref{claim:EQ_circ2}). The latter has a one- and two-qubit gate complexity of $O(n)$ and requires depth $O(\log n)$ (when given $O(n)$ clean ancillas). Summing over all $j \in \{0,1\}^n$ gives us a gate complexity of $O(N \log N)$. Noting that all $\calA^j$ are applied in parallel, the depth of the circuit is $O(\log n) = O(\log \log N)$. This completes the proof.
\end{proof}

\subsubsection{Multi-index function}
The natural next question for us is: can we remove the need for $n$ copies in the $\QACz$ upper bound? For this purpose, we can adopt a simple trick used in~\cite{QAC0tc0}. Consider the index function $\Index:\{0,1\}^n\times \{0,1\}^N\rightarrow \{0,1\}$, defined as $\Index(i,x)=x_i$. It is shown in \cite{QAC0tc0} that one can consider the \emph{multi-index} function $\MIndex:(\{0,1\}^n)^N\times \{0,1\}^N\rightarrow \{0,1\}$ defined as follows:
\begin{equation*}
\MIndex(i_1,\ldots,i_N, x)= [ i_1=\cdots =i_N] \cdot x_{i_1},
\end{equation*}
i.e., it first checks if all the $n$-bit length strings $i_1,\ldots,i_N$ are equal (say equal to $i\in \{0,1\}^n$) and, based on that, outputs $x_i$. The idea is now to say that the players in our protocol are each given $N$ copies of their input, and instead of solving single instances of $\Index$, they solve $\MIndex$ as above. Checking if $i_1=\cdots =i_n$ is doable in $\QACz$, so this does not add any significant complexity to the circuit. Since we also have these many copies ``for free'' by definition, one can also compute $x_i$.

\begin{claim}\label{claim:multiindex_no_copies}
Let $n, N \in \mathbb{N}$ such that $N = 2^n$. Given $\{i_k\}_{k \in [N]}$ and one copy of $x \in \{0,1\}^N$, the function $\MIndex(i_1,\ldots,i_N,x)$ can be computed by a quantum circuit with gate complexity $O(N \log N)$ followed by an $\textsf{AND}$ gate\footnote{Throughout this paper, we assume without loss of generality that $\textsf{NOT}$ gates are free, simply by increasing the size of the circuits by at most a constant factor (using De Morgan's laws)}. Moreover, the problem is in $\land \circ \QNCz[{\log \log N}]$.
\end{claim}
\begin{proof}
Recall that the multi-index function is:
\begin{equation*}
\MIndex(i_1,\ldots,i_N, x)= [ i_1=\cdots =i_N] \cdot x_{i_1}.
\end{equation*}
We now describe a quantum circuit used to compute the multi-index function. We are given the input $(i_1,i_2,\ldots,i_N)$ and $x$. We first define a quantum circuit $U_1$ to check the indicator  $[ i_1=\cdots =i_N]$. Let us introduce $Nn$ ancilla qubits $\{o_{k,\ell}\}_{k \in [N], \ell \in [n]}$ with each initialized to $\ket{0}$ and an ancilla qubit $s$ initialized to $\ket{0}$. We first apply a layer of $Nn$ $\CNOT$s in parallel with a $\CNOT$ gate from each $i_{k,\ell}$ to $o_{k,\ell} \, $ for all $ k \in [N], \ell \in [n]$. This has the effect of setting the state of $o_{k,\ell}$ to that of $i_{k,\ell}$. We then apply another layer of $\CNOT$s in parallel with a $\CNOT$ gate from each $i_{k+1,\ell}$ to $o_{k,\ell}, $ for all $ k \in [N], \ell \in [n]$\footnote{We actually mean from $i_{(k+1)\,\text{mod} N, \ell}$ to $o_{k,\ell}$ so that the bits of $i_1$ are chosen for $k=N$.} At this point, the state of each $\ket{o_{k,\ell}}$ is $\ket{i_{k,\ell} \oplus i_{k+1,\ell}}$. Applying a single-qubit $X$ gate on each $o_{k,\ell}$ then sets the state to be the indicator $[i_{k,\ell} = i_{k+1,\ell}]$, i.e., $\ket{o_{k,\ell}} = \ket{[i_{k,\ell} = i_{k+1,\ell}]}$. We could now measure all the qubits over $\{o_{k,\ell}\}$ in the computational basis and apply an $\AND$ gate over the measurement outcomes to obtain the value of $[i_1 = \ldots = i_N]$. However, we will postpone this step after describing how to obtain $x_{i_1}$ in case $[i_1=\cdots =i_N] = 1$.

To obtain $x_{i_1}$, we proceed similarly as we had in Claim~\ref{claim:index_with_copies}. 
 For each $j \in \{0,1\}^n$, we consider the quantum circuit $\calA^j$ from Claim~\ref{claim:EQ_circ2} (with instantiations of $y$ as $i_j$, $z$ as $j$, $x$ as $x$ and $j$ as the integer corresponding to the given binary representation here) taking as input $\ket{i_j}$, $\ket{x_j}$, and an additional ancilla qubit initialized to $0$, which has the following action:
\begin{equation*}
\ket{i_j}\ket{x_j}\ket{0} \xrightarrow{\calA^j} \ket{i}\ket{x_j}\ket{x_j \land [i_j = j]}.
\end{equation*}
Let us denote that last qubit as $\ket{t_j}$ for a given $j$. To compute the multi-index function, we now apply the circuit $U$ composed of the smaller circuits $\calA^j$ for $j \in \{0,1\}^n$ in parallel as follows:
\begin{align*}
&\left(\ket{i_1}\ket{x_0}\ket{0}\right) \ldots \left(\ket{i_{N}}\ket{x_{N-1}}\ket{0}\right) \ket{0}^{Nn} 
\\
&\qquad \xrightarrow{(\calA^{0} \otimes \ldots \otimes \calA^{N-1})U_1} \\
&\left(\ket{i_1}\ket{x_0}\ket{t_0}\right) \ldots \left(\ket{i_N}\ket{x_{N-1}}\ket{t_{N-1}}\right) \left(\otimes_{k \in [N], \ell \in [n]} \ket{o_{k,\ell}} \right),
\end{align*}
where we have used the equivalent integer representation of the index $j$ when denoting $x_j$ and $\calA^j$. We note that $\ket{t_j} = \ket{x_j \land [i_j = j]}$ for each $j$. We now simply measure the qubits over all $\ket{t_j}$ in the computational basis and obtain the measurement outcome $t_j = x_j \land [i = j]$. We also measure the qubits over all $\ket{o_{k,\ell}}$ in the computational basis to obtain the measurement outcomes $o_{k,\ell}$. Applying an $\AND$ gate over all the measurement outcomes $\{o_{k,\ell}\}_{k \in [N],\ell \in [n]}$ gives us $s$, which has the value:
\begin{equation*}
s = \prod_{k \in [N], \ell \in [n]} [i_{k,\ell} = i_{k+1,\ell}] = [ i_1=\cdots =i_N].
\end{equation*}
Applying an \textsf{OR} gate over all ${t_j}$ measurement outcomes gives us $t$, which has the value:
\begin{equation*}
t = \bigvee_{j\in \{0,1\}^n} t_j  = \bigvee_{j\in \{0,1\}^n} x_j \land [i_j=j].
\end{equation*}
Applying a two-bit $\textsf{AND}$ gate on $s$ and $t$ gives us:
\begin{equation*}
st = [ i_1=\cdots =i_N] \cdot \left(\bigvee_{j\in \{0,1\}^n} x_j \land [i_j=j]\right) = [ i_1=\cdots =i_N] \cdot \left(\bigvee_{j\in \{0,1\}^n} x_j \land [i_1=j]\right), 
\end{equation*}
where we noted in the second equality that $st$ takes a non-zero value only if $[ i_1=\cdots =i_N] = 1$, in which case, we can replace $i_j$ with $i_1$ in the second term. This would then compute the $\MIndex$ function. We note that we can replace the above classical processing with just one $\AND$ gate over the measurement outcomes of $\{t_j\}$ and $\{o_{k,\ell}\}$ with interleaved $\textsf{NOT}$ gates by using DeMorgan's laws. The circuit $U_1$ uses $O(N \log N)$ many $\CNOT$ gates and has depth $2$. For each $\calA^j$, we used $O(n)$ many single-qubit gates applied in parallel to each other and one $(n+1)$-controlled Toffoli gate. The latter has a one- and two-qubit gate complexity of $O(n)$ and requires depth $O(\log n)$ (when given $O(n)$ clean ancillas). Summing over the gates in $\calA^{j} $ for all $j \in \{0,1\}^n$ and $U_1$ gives us a gate complexity of $O(N \log N)$. Noting that all $\calA^j$ are applied in parallel, the depth of the overall circuit $U$ is $O(\log n) = O(\log \log N)$. This completes the proof.
\end{proof}

\subsection{Final candidate function}

\subsubsection{Iterated index function}
\label{sec:itreindx}

\paragraph{Parameter choices.} Fix the transformer parameters $H,d,p,L$, i.e., the number of heads, model dimension, precision and number of layers. We use the same parameter setting as in~\cite{chen2024theoretical}: we will first define $T=Hdp$, which is the \emph{bandwidth} of the players. With this, the remaining parameters are defined as follows:
\begin{align*}
&K(T,L)=(TL)^8,\quad m(T,L)=K^{\sum_{\ell=0}^{L-1}8^{\ell}+1},\\  &n_\ell(T, L)= K^{4 \cdot 8^{L-\ell-1}}\text{ for } \ell \in [L-1],\quad N_\ell(T, L) =m \cdot \prod_{j=1}^{\ell} n_j \text{ for }\ell \in \{0, \ldots, L-1\}.
\end{align*}
We will sometimes drop the $(T,L)$ dependence when it is clear from the context. 

\paragraph{Iterated index function.}
For the iterated index function as defined by~\cite{chen2024theoretical}, the input is a tuple
$(w, z_0, z_1, \ldots, z_L)$, where:
\begin{enumerate}
  \item  $w = (w_1, \ldots, w_{L-1})$ has $L-1$ components, with
  $w_\ell \in [n_\ell]$;
  \item the first argument $z_0 \in [m]$ is referred to as the \emph{seed};
  \item each $\{z_\ell\}_{\ell \in [L]}$  is a function satisfying $z_\ell:[N_{\ell-1}]\rightarrow [N_{\ell-1}]$. We identity the input $[N_{\ell-1}]$ with $[n_{\ell-1}] \times [N_{\ell-2}]$  since $N_{\ell-1} = n_{\ell-1} \cdot N_{\ell-2}$ (so we may
  write $z_\ell(w_{\ell-1}, i_{\ell-1})$ for $\ell \geq 2$). 
\end{enumerate}
With this notation, the input to player $\ell$ has size:
\begin{equation*}
A_{-1}=\prod_{\ell=1}^{L-1}[n_\ell],\quad A_0=[m],\quad  A_\ell:=[N_{\ell-1}]^{N_{\ell-1}}.
\end{equation*}
The goal of the iterated index problem is to simply compute the following function: on input $(w, z_0, z_1, \ldots, z_L) \in A_{-1} \times A_0 \times A_1 \times \cdots \times A_L$, output $\IIndex(w, z_0, \ldots, z_L) := i_L$, where:
\begin{equation*}
i_0=z_0,\quad i_1=z_1(i_0), \quad i_2=z_2(w_1,i_1),\hspace{3mm}\cdots \hspace{3mm}, i_{\ell+1}=z_{\ell+1}(w_\ell,i_\ell) \quad \text{ for }\quad\ell\in [L-1].
\end{equation*}
\textbf{Intuition.} For some intuition, the reason we think of the expression above as an \emph{iterated index} is as follows. Consider player $1$, who receives the truth table of $z_1:[N_0]\rightarrow [N_0]$ (i.e., a message with $(N_0)^{\log N_0}$ bits), and needs to produce the $i_0$-th index in the truth table of $z_1$, which we denote $i_1$. Now, similarly, player $2$ receives the truth table of $z_2:[N_1]\rightarrow [N_1]$ and needs to output the $(w_1,i_1)$-th index of this truth table $z_2$ (which we denote as $z_2(w_1,i_1)$). This procedure iteratively continues for $L-1$ rounds before player $L$ player outputs $z_{L}(w_{L-1},i_{L-1})$ and passes this index to the final player, player $(L + 1)$. This final player $(L+1)$ outputs $z_{L+1}(w_L,i_L)$, which is the intended value of the function.
\begin{theorem}[{\cite[Theorem~4.1]{chen2024theoretical}}]
\label{thm:cpw-main}
There exists $\alpha_L = 2^{-O(L)}$ with the following property. Every deterministic autoregressive protocol  for $\IIndex$ 
on the parameters $(K, m, \{n_\ell\}, \{N_\ell\})$, with
$L$ epochs and bandwidth $B \leq T$  satisfies:
\begin{equation*}
B \;\geq\; n_{\textup{orig}}(T, L)^{\alpha_L},
\end{equation*}
where $n_{\textup{orig}}(T, L) := 1 + 1 + \sum_{\ell \in [L]} N_{\ell-1}
\sim K^{2^{O(L)}}$ is the total token length of the input.\footnote{We will prove this final upper bound expression on $n_{\textup{orig}}(T, L)$ shortly.}
\end{theorem}

\subsubsection{Iterated multi-index function}
\label{sec:multiindex}
\paragraph{Parameter choices.} We retain all parameter choices from the iterated 
index problem above ($K, m, n_\ell, N_\ell$ as functions of $T = Hdp$ and $L$), 
and introduce one additional parameter, the \emph{copy parameter}:
\begin{equation}\label{eq:copy_parameter}
t(T,L) := \prod_{\ell=1}^{L} N_{\ell-1} = \prod_{\ell=0}^{L-1} N_\ell.    
\end{equation}
As before, we drop the $(T,L)$ dependence when the context is clear. 

\paragraph{Iterated multi-index function.}
The iterated multi-index function is the $t$-copy extension of the iterated index 
function. The input is a tuple 
$\left(\{w^{(s)}\}_{s \in [t]}, \{z_0^{(s)}\}_{s \in [t]}, \{z_1^{(s)}\}_{s \in [t]}, 
\ldots, \{z_L^{(s)}\}_{s \in [t]}\right)$, where, for each $s \in [t]$:
\begin{enumerate}
  \item $w^{(s)} = (w_1^{(s)}, \ldots, w_{L-1}^{(s)}) \in A_{-1}$, i.e., 
  $w_\ell^{(s)} \in [n_\ell]$ for $\ell \in [L-1]$;
  \item the seed $z_0^{(s)} \in A_0 = [m]$;
  \item each table $z_\ell^{(s)} \in A_\ell$ is a function 
  $z_\ell^{(s)} : [N_{\ell-1}] \to [N_{\ell-1}]$, with the same identification of 
  $[N_{\ell-1}]$ with $[n_{\ell-1}] \times [N_{\ell-2}]$ as in the iterated index 
  problem.
\end{enumerate}
The per-role input domain in the iterated multi-index problem is the $t$-fold product:
\begin{equation*}
A_{-1}^{\,t} = \left(\prod_{\ell=1}^{L-1}[n_\ell]\right)^t, 
\quad A_0^{\,t} = [m]^{\,t}, 
\quad A_\ell^{\,t} = \left([N_{\ell-1}]^{N_{\ell-1}}\right)^t \text{ for } \ell \in [L].
\end{equation*}
Define the \emph{consistency predicate} as:
\begin{equation*}
E\left(\{w^{(s)}\}, \{z_\ell^{(s)}\}\right) := 
\bigwedge_{s \in [t]} \left( w^{(s)} = w^{(1)} \;\land\; z_0^{(s)} = z_0^{(1)}
\;\land\; \bigwedge_{\ell \in [L]} z_\ell^{(s)} = z_\ell^{(1)} \right) \in \{0,1\},
\end{equation*}
i.e., $E = 1$ iff all $t$ copies of \emph{each} input component agree. The goal of the iterated multi-index problem is to compute the following function: on input 
$\left(\{w^{(s)}\}, \{z_0^{(s)}\}, \ldots, \{z_L^{(s)}\}\right)$, output:
\begin{equation*}
\IMIndex\left(\{w^{(s)}\}, \{z_\ell^{(s)}\}\right) := 
E\left(\{w^{(s)}\}, \{z_\ell^{(s)}\}\right) \cdot 
\IIndex\left(w^{(1)}, z_0^{(1)}, z_1^{(1)}, \ldots, z_L^{(1)}\right).
\end{equation*}
In other words: if every component is internally consistent across the $t$ copies,  output the iterated index value on those (common) inputs; otherwise output $0$. 

\textbf{Intuition.} We again provide some intuition for this candidate problem, and why we call it iterated multi-index. First, observe that if:
\begin{align}
\label{eq:conditionsforiterindex}
w_i^{s}=w_1, z_0^{(s)}=z_0^{(1)}, z_\ell^{(s)}=z_\ell^{(1)} \text{ for all } s\in [t],\ell\in [L],
\end{align}
that means that the players are just get the same input string several times. Once that is the case, the players can just treat their input as $w^{(1)},z_i^{(1)}$ for $i\in \{0,\ldots,L\}$ and just run the usual index protocol. But recall that in order to compute even a \emph{single} instance of the index problem, the $\QNCz$ circuit requires copies of the index, which is not available for free. Having $t$ copies of the individual input has the effect of providing these copies ``for free'', and part of the input for the $\QNCz$ circuit we will construct. It is clear that having just one copy would be equivalent to the iterated index problem. Indeed, the ``diagonal'' version of the problem above is exactly the iterated index problem (whose lower bound will be crucial in proving our lower bound). For the off-diagonal elements, we simply enforce the condition that if any of the equalities in Eq.~\eqref{eq:conditionsforiterindex} are not met, then the overall function just evaluates to $0$.

\subsection{Quantum upper bound}
\paragraph{Iterated index function.}
We have the following main result regarding the computation of the iterated index function.

\begin{theorem}\label{thm:iter-index_with_copies}
Let $L \in \mathbb{N}, L \geq 2$, $n_{\ell}, N_{\ell} \in \mathbb{N}$  for all $\ell \in [L-1]$ such that $N_\ell = 2^{n_\ell}, $ for all $ \ell \in [L-1]$. Given $K^{O(L 8^L)}$ copies of each of $\{w_\ell\}_{\ell \in [L-1]}$, $i_0$, and $\{z_{\ell}\}_{\ell \in [L]}$, the iterated-index function $\IIndex$ can be computed by a quantum circuit with gate complexity $O(K^{L 8^L} \log K)$ followed by an $\AND$ gate. Moreover, the problem is in $\land \circ \QNCz[{\log \log K}]$.
\end{theorem}

\begin{proof}
Denote the input to the algorithm as:
\begin{equation*}
X=(w,z_0,z_1,\ldots,z_L), \text{ where } w=(w_1,\ldots,w_{L-1}).
\end{equation*}
Recall that the iterated-index computation is then the following sequence of computations starting from the initial seed $z_0$:
\begin{equation*}
i_0=z_0,
\quad
i_1=z_1(i_0), \quad i_{\ell+1}=z_{\ell+1}(w_\ell,i_\ell) \quad \text{ for every } \quad \ell \in [L-1].
\end{equation*}
Explicitly, the intermediate computations are:
\begin{equation*}
i_2=z_2(w_1,i_1),\quad
i_3=z_3(w_2,i_2),\quad \ldots,\quad
i_L=z_L(w_{L-1},i_{L-1}).
\end{equation*}
One approach would of course be to sequentially and recursively prepare the intermediate indices
$i_1,i_2,\ldots,i_{L-1}$ on quantum registers, which are the inputs to the truth tables $z_2,z_3,\ldots,z_L$. However, this might increase depth. So, instead, we verify all possible computation histories \emph{in parallel}. 

Note that the input to the final truth table $z_L$ is of size $N_{L-1}$. A candidate computation history up to $z_L$ is a tuple:
\begin{align}\label{eq:defntau}
\tau=(r_1,\ldots,r_{L-1},a_0,a_1,\ldots,a_{L-1}),
\end{align}
where:
\begin{equation*}
a_0\in[N_0],
\quad
a_\ell\in[N_{\ell-1}],\quad r_\ell\in[n_\ell] \quad  \text{ for  } \quad 1\leq \ell\leq L-1.
\end{equation*}
For a fixed $\tau$, let $\textsf{Const}_\tau(X)$ denote the predicate that all of the
following conditions hold:
\begin{align}
&z_0=a_0, \label{eq:Const_tauX_start}\\
&w_\ell=r_\ell
\qquad\text{for every }1\leq \ell\leq L-1,\\
&z_1(a_0)=a_1,\\
&z_{\ell+1}(r_\ell,a_\ell)=a_{\ell+1}
\qquad\text{for every }1\leq \ell\leq L-2.
\label{eq:Const_tauX_end}
\end{align}
In particular, for every $\tau$, $\textsf{Const}_\tau(X)$ checks if $X$ is a valid iterated index transcript. Note that exactly one $\tau^*$ satisfies $\textsf{Const}_{\tau^*}(X)=1$. For this $\tau^*$ and hence when $\textsf{Const}_\tau(X)=1$, we need to compute the output of $z_L(w_{L-1},a_{L-1})$. To this end, for every $b \in [\log(N_{L-1})]$,  define:
\begin{equation}\label{def:C_tau_b}
C_{\tau,b}(X) := \neg\left(\textsf{Const}_\tau(X) \land \neg z_L(r_{L-1},a_{L-1})_b \right),    
\end{equation}
where the subscript $b$ indicates the bit position of $z_L(r_{L-1},a_{L-1})$ (which is viewed as a vector in $\{0,1\}^{\log N_{L-1}}$). For every invalid history $\tau \neq \tau^*$, we have $\textsf{Const}_\tau(X) = 0$, and hence $C_{\tau,b}(X) = 1$. For the unique valid history $\tau^*$, we have $\textsf{Const}_{\tau^*}(X) = 1$, and hence: 
\begin{equation*}
C_{\tau^*,b}(X) = \neg\left(1 \land \neg z_L(r_{L-1}^*,a_{L-1}^*)_b \right) = z_L(w_{L-1},i_{L-1})_b,
\end{equation*}
where we have used the fact that $\textsf{Const}_{\tau^*}(X) = 1$ indicates $r_{L-1}^* = w_{L-1},a_{L-1}^*=i_{L-1}$ in the second equality. Taking an $\textsf{AND}$ over all histories then gives us:
\begin{equation}\label{eq:and_C_taub}
\bigwedge_{\tau} C_{\tau,b}(X) = \left( \bigwedge_{\tau \neq \tau^*} 1 \right) \land C_{\tau^*,b}(X) = C_{\tau^*,b}(X) = z_L(w_{L-1},i_{L-1})_b. 
\end{equation}
Doing this over all $b \in [\log (N_{L-1})]$ in parallel would then yield the desired result of $z_L(w_{L-1},i_{L-1})$.

We now describe the corresponding circuit and analyze the corresponding copy complexity, gate complexity and depth. Let $\mathcal{T}$ denote the set of all candidate histories. There are:
\begin{equation}\label{eq:number_of_taus}
|\mathcal{T}| = N_0 \left(\prod_{\ell=1}^{L-1} N_{\ell-1}\right) \left(\prod_{\ell=1}^{L-1} n_\ell\right)=\prod_{\ell=0}^{L-1}N_\ell    
\end{equation}
choices for $\tau$ and hence the size of $\mathcal{T}$, where we used the definition of $N_{L-1}=N_0\prod_{\ell=1}^{L-1} n_\ell$.
Thus, the entire construction will use $|\mathcal{T}|$ many copies of the input $X$ with one copy for each candidate history $\tau \in \mathcal{T}$. 

For a fixed $\tau \in \mathcal{T}$, we now describe a circuit to compute $C_{\tau,b}(X), \, \forall b \in [\log N_{L-1}]$ (Eq.~\eqref{def:C_tau_b}). Let us define $\ket{X^{(\tau)}}$ as a portion of the copy of the input $X$ assigned to the computation history $\tau = (r_1,\ldots,r_L,a_0,a_1,\ldots,a_{L-1})$ (Eq.~\eqref{eq:defntau}):
\begin{equation}\label{eq:input_assigned_to_tau}
X^{(\tau)} := (w, z_0, z_1(a_0), z_2(r_1,a_1), \ldots, z_\ell(r_\ell,a_\ell), \ldots, z_{L-1}(r_{L-2},a_{L-2}))    
\end{equation}
Once $\tau$ is fixed, all the involved addresses in $z_{\ell+1}(r_{\ell},a_{\ell}), \forall \ell \in [L-2]$ are fixed. We will now describe a circuit $U_\tau$ with the following action on inputs of $\ket{X^{(\tau)}}$, $\ket{z_L(r_{L-1},a_{L-1})}$ along with $2 \log N_{L-1}$ ancilla qubits:
\begin{align}
&\ket{X^{(\tau)}}\ket{0}^{\otimes \log_{N_{L-1}}} \ket{z_L(r_{L-1}(r_{L-1},a_{L-1})}\ket{0}^{\otimes \log_{N_{L-1}}} \nonumber \\
\xrightarrow{U_\tau} & \ket{X^{(\tau)}} \ket{\textsf{Const}_{\tau}(X)}^{\otimes \log_{N_{L-1}}} \ket{z_L(r_{L-1}(r_{L-1},a_{L-1})} \left( \otimes_{b \in [\log N_{L-1}]} \ket{C_{\tau,b}} \right).
\label{eq:def_U_tau}
\end{align}

For the equality check of $\textsf{Const}_\tau(X) := [X^{\tau}=\tau]$, we use the quantum circuit from Claim~\ref{claim:EQ_circ2} (with instantiations of $y$ as $X^{(\tau)}$, $z$ as $\tau$) which we denote by $\calA_\tau$ and has the following action on input $\ket{X^{(\tau)}}$ and an additional ancilla qubit initialized to $\ket{0}$:
\begin{equation}\label{eq:output_const_tau_X}
\ket{X^{(\tau)}} \ket{0} \xrightarrow{\calA_{\tau}} \ket{X^{(\tau)}} \ket{\textsf{Const}_\tau(X)}.    
\end{equation}
Thus, $\calA_{\tau}$ checks whether the involved truth tables are equal to the fixed strings specified by $\tau$. Adding $(\log(N_{L-1}) - 1)$ additional ancilla qubits initialized to $\ket{0}$, we can copy over the value of $\textsf{Const}_\tau(X)$ to these ancilla qubits using a $\CNOT$-tree of depth $O(\log \log N_{L-1})$ and gate complexity $O(\log N_{L-1})$. We denote this subcircuit by $\CNOT_{\mathrm{tree}}$. We have the following so far:
\begin{equation*}
\ket{X^{(\tau)}} \ket{0}^{\log N_{L-1}} \xrightarrow{\calA_{\tau}} \ket{X^{(\tau)}} \ket{\textsf{Const}_\tau(X)} \ket{0}^{\log N_{L-1} - 1} \xrightarrow{\CNOT_{\mathrm{tree}}} \ket{X^{(\tau)}} \ket{\textsf{Const}_\tau(X)}^{\log N_{L-1}}.
\end{equation*}
Recalling the definition of $C_{\tau,b}$ from Eq.~\eqref{def:C_tau_b}, we now proceed as follows. Consider the input of the state $\ket{z_L(r_{L-1},a_{L-1})}$ and $\log N_{L-1}$ many additional ancilla qubits initialized to $\ket{0}$ with each one corresponding to $b \in [\log N_{L-1}]$ and denoted by $t_{\tau,b}$. Let us also denote the qubit corresponding to $z_L(r_{L-1},a_{L-1})_b$ as $s_b$ for all $b \in [\log N_{L-1]}]$. For each $b \in [\log N_{L-1}]$, we then apply an $X$ gate on qubit $s_b$ followed by a $3$-qubit Toffoli gate with controls on $s_b$ and $b$-th copy of $\ket{\textsf{Const}_{\tau}(X)}$ with the target being $t_{\tau,b}$. We then apply another $X$ gate on the target qubit $t_{\tau,b}$ and an $X$ gate on qubit $s_b$ to reset it. We will denote the above sequence of gates as $B_{\tau,b}$ for each $b \in [\log N_{L-1}]$. Noting that we can apply all the subcircuits $B_{\tau,b}$ in parallel, we can define the overall circuit as $B_{\tau} := \otimes_{b \in [\log N_{L-1}]} B_{\tau,b}$ since each of the qubits $s_b$ for $b \in [\log N_{L-1}]$ is unique and can be individually addressed, and we have a copy of $\ket{\textsf{Const}_{\tau}(X)}$ for each such $b$. The circuit $B_\tau$ then has the following action:
{\small
\begin{align*}
& \ket{X^{(\tau)}} \ket{\textsf{Const}_\tau(X)}^{\log N_{L-1}} \ket{z_L(r_{L-1}(r_{L-1},a_{L-1})} \ket{0}^{\otimes \log N_{L-1}} \\
\xrightarrow{B_{\tau}} & \ket{X^{(\tau)}}\ket{\textsf{Const}_{\tau}(X)}^{\otimes \log N_{L-1}} \ket{z_L(r_{L-1}(r_{L-1},a_{L-1})} \left( \mathop{\otimes}_{b \in [\log N_{L-1}]} \ket{\underbrace{\neg\left(\textsf{Const}_\tau(X) \land \neg z_L(r_{L-1},a_{L-1})_b \right)}_{:= C_{\tau,b}(X)}}\right).    
\end{align*}}
The desired circuit $U_\tau$ in Eq.~\eqref{eq:def_U_tau} is then $U_\tau := B_\tau (\CNOT_{\mathrm{tree}} \otimes \id)(\calA_\tau \otimes \id)$ applied to the input states as described in Eq.~\eqref{eq:def_U_tau}.

Applying the circuits $U_\tau$ in parallel for all $\tau \in \calT$ would then simply produce:
\begin{align*}
& \otimes_{\tau \in \calT} \left(\ket{X^{(\tau)}} \ket{0}^{\otimes \log N_{L-1}} \ket{z_L(r_{L-1},a_{L-1})} \ket{0}^{\otimes \log N_{L-1}} \right) \\
\xrightarrow{\otimes_{\tau \in \calT} U_\tau } & \otimes_{\tau \in \calT} \left[ \ket{X^{(\tau)}} \ket{\textsf{Const}_{\tau}(X)}^{\otimes \log N_{L-1}} \ket{z_L(r_{L-1},a_{L-1})} \left( \otimes_{b \in [\log N_{L-1}]} \ket{C_{\tau,b}(X)} \right) \right].
\end{align*}
Measuring the qubits $t_{\tau,b}$ which hold the state $\ket{C_{\tau,b}}$ in the computational basis for all $\tau \in \calT, b \in [\log N_{L-1}]$ and then applying the $\AND$ logical gate over the outcomes corresponding to $\tau \in \calT$ for each $b \in [\log N_{L-1}]$ then gives us the desired result of Eq.~\eqref{eq:and_C_taub}. \footnote{Note for a specified bit $b$, we only require one $\AND$ gate at the end. Alternately if the goal is to output $[\beta = z_L(w_{L-1},i_{L-1})]$, we can add $\beta$ on to $\log N_{L-1}$ qubits, measure these in the computational basis and then perform a large $\AND$ gate over all the measurement outcomes from before along with these.}
% This overall circuit is summarized in Figure~\ref{fig:circ_iterindex_avec_copies}.

We now analyze the number of qubits required and the gate complexity of the overall circuit. For each $\tau \in \calT$, the total number of input qubits involved in defining $\ket{X^{(\tau)}}$ (Eq.~\eqref{eq:input_assigned_to_tau}), $\ket{z_L(r_{L-1},a_{L-1})}$, and the ancilla qubits is:
\begin{equation*}
Q_\tau =
O\left(\log N_0 + \sum_{\ell=1}^{L-1}\log n_\ell + \sum_{\ell=0}^{L-1}\log N_\ell \right)=O\left(\sum_{\ell=0}^{L-1}\log N_\ell\right),
\end{equation*}
where we have noted that the size of the input and output to each $z_\ell$ requires $\ceil{\log N_{\ell-1}}$ bits to specify.
Using the fact that $N_\ell=m \cdot \prod_{j=1}^{\ell} n_j$, $n_\ell(T, L)= K^{4 \cdot 8^{L-\ell-1}}$, $m=K^{8^L}$ and that $L$ is constant from the parameter choices described in Section~\ref{sec:itreindx}, we then have that:
\begin{equation*}
Q_\tau = O\left(L 2^{3L} \log K \right).
\end{equation*}
Taking into account the size of $\calT$ from Eq.~\eqref{eq:number_of_taus} and our choice of the parameters, the total number of qubits used is:
\begin{equation*}
|\calT| Q_\tau = K^{O(L 8^L)}.
\end{equation*}
Moreover, the total number of copies of the input $X$ used to define $\ket{X^{(\tau)}}$ and $\ket{z_L(r_{L-1},a_{L-1}}$ for each $\tau \in \calT$ goes as $|\calT|$ and hence is $K^{O(L8^L)}$.

For each $\tau \in \calT$, the contribution to gate complexity of $U_\tau$ is due to $\calA_\tau$, which requires:
\begin{equation*}
O\left(\sum_{\ell = 0}^{L-1} \log N_\ell\right) = O(L 2^{3L} \log K)
\end{equation*}
one- and two-qubit gates (Claim~\ref{claim:EQ_circ2}) (i.e., $Q_\tau$) due to the size of $X^{(\tau)}$, $\CNOT_{\mathrm{tree}}$ which has gate complexity $O(\log N_{L-1})$, and $B_\tau$ which has gate complexity $O(\log N_{L-1})$. Overall, the gate complexity is then: 
\begin{equation*}
|\calT| O(L 2^{3L} \log K) = O\left(L 2^{3L} K^{L 8^L} \log K\right).
\end{equation*}
Since all the $U_\tau$s are applied in parallel, the depth of the quantum circuit before the application of the $\AND$ gate is the same as the depth of any particular $U_\tau$. The contribution to the depth of $U_\tau$ is $O(L + \log L + \log \log K)$ from $\calA_\tau$ or $O(\log Q_\tau)$ (Claim~\ref{claim:EQ_circ2}), an additional $O(\log \log N_{L-1})$ due to $\CNOT_{\mathrm{tree}}$, and finally an additional $O(1)$ depth due to $B_\tau$. Overall, the depth of the quantum circuit is then:
\begin{equation*}
O\left(L + \log L + \log \log K\right).
\end{equation*}
This places $\IIndex$ in $\land \circ \QNCz[\log \log K]$. This completes the proof.
\end{proof}
\paragraph{Iterated multi-index function.}
We have the following main result regarding the computation of the iterated multi-index function.
\begin{theorem}\label{thm:iter-multiindex_no_copies}
Let $L \in \mathbb{N}, L \geq 2$, $n_{\ell}, N_{\ell} \in \mathbb{N}$ for all $\ell \in [L-1]$ such that $N_\ell = 2^{n_\ell}, $ for all $ \ell \in [L-1]$, and fix $\beta \in [N_{L-1}]$. Set $t=K^{O(L 8^L)}$. Given  $\{w_\ell^{(s)}\}_{\ell \in [L-1]}$, $i_0^{(s)}$, and $\{z_{\ell}^{(s)}\}_{\ell \in [L]}$ for all $s \in [t]$, the output of $[\beta = \IMIndex(\{w^{(s)}_\ell\}_\ell, i_0^{(s)}, \{z^{(s)}_\ell\}_\ell)]$
can be computed in $\land \circ \QNCz[{\log \log K}]$ with total gate complexity $O(K^{L 8^L} \log K)$.
\end{theorem}
\begin{proof}
The proof proceeds similarly to that of Theorem \ref{thm:iter-index_with_copies}. For every $s \in [t]$, denote the $s$-th input as:
\begin{equation}\label{eq:input_itermultinx}
X^{(s)} = (w^{(s)}, z_0^{(s)}, z_1^{(s)}, \ldots, z_L^{(s)}), \quad \text{ where } w^{(s)} = (w_1^{(s)}, \ldots, w_{L-1}^{(s)}).
\end{equation}
Let $\mathbf{X} = \{X^{(s)}\}_{s \in [t]}$ denote the entire input and define:
\begin{equation}\label{eq:equality_check_across_Xs}
E(\mathbf{X}) := \bigwedge_{s=1}^t [X^{(s)} = X^{(1)}]    
\end{equation}
as the consistency predicate across all $X$s. By definition, we have:
\begin{equation}\label{eq:inter_expression_IMIndex}
\IMIndex\left(\{X^{(s)}\}_{s \in [t]}\right) = E(\mathbf{X}) \IIndex(X^{(1)}).    
\end{equation}
As in Theorem~\ref{thm:iter-index_with_copies}, we will proceed by comparing the input against all possible candidate computation histories up to computation of $z_L^{(1)}(w_{L-1}^{(1)}, i_{L-1}^{(1)})$. We will adopt the notation from there by letting $\mathcal T$ denote the set of candidate computation histories:
\begin{equation*}
\tau=(a_0,a_1,\ldots,a_{L-1},r_1,\ldots,r_{L-1}),
\end{equation*}
where:
\begin{equation*}
a_0\in[N_0],
\qquad
a_\ell\in[N_{\ell-1}]\quad\text{for }1\leq \ell\leq L-1, \text{ and } r_\ell\in[n_\ell]\quad\text{for }1\leq \ell\leq L-1.
\end{equation*}
The number of such histories is:
\begin{equation*}
|\mathcal T| = N_0 \left(\prod_{\ell=1}^{L-1}N_{\ell-1}\right) \left(\prod_{\ell=1}^{L-1}n_\ell\right) \prod_{\ell=0}^{L-1}N_\ell = t,
\end{equation*}
which is precisely the copy parameter from before. Let us order the candidate histories and let $\tau^{(s)}$ denote the $s$th computation history in $\calT$ where $s \in [t]$. We can then denote the components of $\tau^{(s)}$ as $\tau^{(s)} = (a_0^{(s)}, a_1^{(s)}, \ldots, a_{L-1}^{(s)}, r_1^{(s)}, \ldots, r_{L-1}^{(s)})$.

For a fixed history $\tau^{(s)} \in\mathcal T$ and hence fixed $s \in [t]$, define $\textsf{Const}_{\tau^{(s)}}(X^{(s)})$ as the predicate that all of the following conditions hold:
\begin{align}
z_0^{(s)} &= a_0^{(s)}, \nonumber \\
w_\ell^{(s)}&=r_\ell^{(s)}
\qquad\text{for every }1\leq \ell\leq L-1, \nonumber \\
z_1^{(s)}(a_0^{(s)})&=a_1^{(s)}, \nonumber \\
z_{\ell+1}^{(s)}(r_\ell^{(s)},a_\ell^{(s)})&=a_{\ell+1}^{(s)}
\qquad\text{for every }1\leq \ell\leq L-2. \label{eq:Const_tau_s_b}
\end{align}
Thus $\textsf{Const}_{\tau^{(s)}}(X^{(s)})$ checks whether the copy $X^{(s)}$ is consistent with the history
$\tau^{(s)}$ up to the computation of $z_L^{(s)}, \forall s \in [t]$. Exactly as in the proof of Theorem~\ref{thm:iter-index_with_copies}, note that exactly one $\tau^* \in \calT$ satisfies $\textsf{Const}_{\tau^*}(X^{(s)})=1$. For this $\tau^*$ (and hence when $\textsf{Const}_{\tau^*}(X^{(s)})=1$), we need to compute the output of $z_L^{(s)}(w_{L-1}^*,a_{L-1}^*)$ (where we have decorated the components of $\tau^*$ with superscript $*$). To this end, for every $b \in [\log(N_{L-1})]$,  define:
\begin{equation}\label{def:C_tau_s_b}
C_{\tau^{(s)},b}(X^{(s)}) := \neg\left(\textsf{Const}_{\tau^{(s)}}(X^{(s)}) \land \neg z_L^{(s)}(r_{L-1}^{(s)},a_{L-1}^{(s)})_b \right),
\end{equation}
where the subscript $b$ indicates the bit position of $z_L^{(s)}(r_{L-1}^{(s)},a_{L-1}^{(s)})$ (which is viewed as a vector in $\{0,1\}^{\log N_{L-1}}$). For every invalid history $\tau^{(s)} \neq \tau^*$, we have $\textsf{Const}_{\tau^{(s)}}(X^{(s)}) = 0$, and hence $C_{\tau^{(s)},b}(X^{(s)}) = 1$. For the unique valid history $\tau^*$, we have $\textsf{Const}_{\tau^*}(X^{(s)}) = 1$, and hence:
\begin{equation}\label{eq:Ctaustar_b2}
C_{\tau^*,b}(X) = \neg\left(1 \land \neg z_L^{(s)}(r_{L-1}^*,a_{L-1}^*)_b \right) = z_L^{(s)}(w_{L-1}^{(s)},i_{L-1}^{(s)})_b,    
\end{equation}
where we have used the fact that $\textsf{Const}_{\tau^*}(X^{(s)}) = 1$ indicates $r_{L-1}^* = w_{L-1}^{(s)},a_{L-1}^*=i_{L-1}^{(s)}$ in the second equality. 

We next define our approach to evaluate $E(\mathbf{X})$ (Eq.~\eqref{eq:equality_check_across_Xs}). First, let $R$ denote the number of bits used to encode one copy $X^{(s)}$. For every $s\in [t-1]$ and bit position $j\in [R]$, define:
\begin{equation*}
D_{s,j}(\mathbf X)
:=
[X^{(s)}_j=X^{(s+1)}_j].
\end{equation*}
It is then immediate that:
\begin{equation}\label{eq:EX_as_DXs}
\bigwedge_{s=1}^{t-1}\bigwedge_{j=1}^{R}D_{s,j}(\mathbf X) = E(\mathbf{X}),    
\end{equation}
which evaluates to $1$ if and only if all copies $X^{(s)}, \, \forall s \in [t]$ are identical. For the $b$-th output bit, define:
\begin{equation}\label{eq:compute_IterMultInd}
F_b(\mathbf{X}) := \left(\bigwedge_{s=1}^{t-1}\bigwedge_{j=1}^{R}D_{s,j}\right)
\land \left(\bigwedge_{s \in [t]}C_{\tau^{(s)},b}(X^{(s)})\right).
\end{equation}
We claim that:
\begin{equation}\label{eq:func_Fb}
F_b(\mathbf{X}) = \IMIndex\left(\{X^{(s)}\}_{s \in [t]}\right)_b.    
\end{equation}
To see this, suppose first that the $t$ input copies are not all identical. Then, there exist $s \in [t-1]$ and $j \in [R]$ such that $D_{s,j} = 0$ and consequently $F_b(\mathbf{X})=0$, which is consistent with the evaluation of $\IMIndex$ on inconsistent copies. Now, suppose that all copies are identical i.e., $X^{(s)} = X^{(1)}, \, \forall s \in [t]$. The first factor in $F_b(\mathbf{X})$ (Eq.~\eqref{eq:compute_IterMultInd}), which is nothing but $E(\mathbf{X})$, is then $1$. Moreover, there is also a unique computation history $\tau^\star$ for which $\textsf{Const}_{\tau^*}(X^{(1)})=1$, as observed earlier. We then have:
\begin{equation*}
\bigwedge_{s \in [t]}C_{\tau^{(s)},b}(X^{(s)}) = \bigwedge_{s \in [t]}C_{\tau^{(s)},b}(X^{(1)}) = z_L^{(1)}(w_{L-1}^{(1)}, i_{L-1}^{(1)})_b,
\end{equation*}
where the second equality follows from the fact that all the input copies are equal (indeed, they are all $X^{(1)}$), and the final equality follows as in Eq.~\eqref{eq:and_C_taub} in the proof of Theorem~\ref{thm:iter-index_with_copies}. This proves the claimed identity of Eq.~\eqref{eq:func_Fb}.

We now describe a quantum circuit to compute $F_b(\mathbf{X})$. We first introduce a circuit to evaluate $D_{s,j}(\mathbf{X})$ and hence the expression in Eq.~\eqref{eq:EX_as_DXs}. We will describe $U_1$ that has the following action on the inputs $\{X^{(s)}\}_{s \in [t]}$ along with $(t-1)R$ many ancilla qubits initialized as $\ket{0}$:
\begin{equation}
\left( \otimes_{s \in [t]} \ket{X^{(s)}} \right) \ket{0}^{(t-1)R} \xrightarrow{U_1} \left( \otimes_{s \in [t]} \ket{X^{(s)}} \right) \left(\bigotimes_{s \in [t-1]} \bigotimes_{j \in [R]} \ket{D_{s,j}} \right)
\end{equation}
For this, we will re-use the circuit introduced in the proof of Claim~\ref{claim:multiindex_no_copies} for checking equality of different copies of indices. Let us denote the $(t-1)R$ ancilla qubits as $\{o_{s,j}\}_{s \in [t-1], j \in [R]}$. We first apply a layer of $(t-1)R$ $\CNOT$s in parallel with a $\CNOT$ gate from each $\ket{X^{(s)}_{b}}$ to $o_{s,j}$ for all $s \in [t-1], j \in [R]$. This has the effect of setting the state of $o_{s,j}$ to that of $X^{(s)}_{j}$. We then apply another layer of $\CNOT$s in parallel with a $\CNOT$ gate from each $X^{(s+1)}_j$ to $o_{s,j}$ for all $s \in [t-1], j \in [R]$. At this point, the state of each qubit $o_{s,j}$ is $\ket{X^{(s)}_j \oplus X^{(s+1)}_j}$. Applying a single-qubit $X$ gate on each $o_{s,j}$ then sets of the state to be the indicator $[X^{(s)}_{j} = X^{(s+1)}_j]$ i.e., $\ket{o_{s,j}} = \ket{[X^{(s)}_j = X^{(s+1)}_j]} = \ket{D_{s,j}(\mathbf{X})}, \, \forall s \in [t-1], j \in [R]$. At this point, we could measure all the qubits over $\{o_{s,j}\}$ in the computational basis and apply an $\AND$ gate over the measurement outcomes to obtain the value of $E(\mathbf{X})$ (Eq.~\eqref{eq:EX_as_DXs}). However, we will postpone this step until after describing how to obtain $z^{(1)}_L(w_{L-1}^{(1)}, i_{L-1}^{(1)})$ in the case $E(\mathbf{X}) = 1$.

For a fixed $s \in [t]$ and hence a fixed $\tau^{(s)} \in \mathcal{T}$, we now describe a circuit to compute $C_{\tau^{(s)},b}(X)$ for all $b \in [\log N_{L-1}]$ (Eq.~\eqref{def:C_tau_s_b}). Define $\ket{Y^{(s)}}$ as a portion of the copy of the input $X^{(s)}$ assigned to the computation history $\tau^{(s)} = (r_1^{(s)},\ldots,r_L^{(s)},a_0^{(s)},a_1^{(s)},\ldots,a_{L-1}^{(s)})$:
\begin{equation}\label{eq:input_assigned_to_tau_s}
Y^{(s)} := (w^{(s)}, z_0^{(s)}, z_1^{(s)}(a_0^{(s)}), z_2^{(s)}(r_1^{(s)},a_1^{(s)}), \ldots, z_\ell^{(s)}(r_\ell^{(s)},a_\ell^{(s)}), \ldots, z_{L-1}^{(s)}(r_{L-2}^{(s)},a_{L-2}^{(s)})).  
\end{equation}
Once $\tau^{(s)}$ is fixed, all the addresses involved in $z_{\ell+1}^{(s)}(r_{\ell}^{(s)},a_{\ell}^{(s)})$ for all $\ell \in [L-2]$ are fixed. We can then use the circuit $U_\tau$ from the proof of Theorem~\ref{thm:iter-index_with_copies}, denoted here as $U_{\tau^{(s)}}$, to clarify which copy of the input we work with, and which has the following action inputs of $\ket{Y^{(s)}}$, $\ket{z_L^{(s)}(r_{L-1}^{(s)},a_{L-1}^{(s)})}$ along with $2 \log N_{L-1}$ ancilla qubits:
\begin{align}
&\ket{Y^{(s)}}\ket{0}^{\otimes \log_{N_{L-1}}} \ket{z_L^{(s)}(r_{L-1}^{(s)},a_{L-1}^{(s)})}\ket{0}^{\otimes \log_{N_{L-1}}} \nonumber \\
\xrightarrow{U_{\tau^{(s)}}} & \ket{Y^{(s)}} \ket{\textsf{Const}_{\tau^{(s)}}(X^{(s)})}^{\otimes \log_{N_{L-1}}} \ket{z_L^{(s)}(r_{L-1}^{(s)},a_{L-1}^{(s)})} \left( \otimes_{b \in [\log N_{L-1}]} \ket{C_{\tau^{(s)},b}(X^{(s)}} \right).
\label{eq:def_U_tau_s}
\end{align}

Applying the circuits $U_{\tau^{s}}$ in parallel for all $s \in [t]$ (and hence all $\tau^{(s)} \in \calT$) would then simply produce:
\begin{align*}
& \otimes_{s \in [t]} \left(\ket{Y^{(s)}} \ket{0}^{\otimes \log N_{L-1}} \ket{z_L^{(s)}(r_{L-1}^{(s)},a_{L-1}^{(s)})} \ket{0}^{\otimes \log N_{L-1}} \right) \\
\xrightarrow{\otimes_{s \in [t]} U_{\tau^{(s)}} } & \otimes_{s \in [t]} \left[ \ket{Y^{(s)}} \ket{\textsf{Const}_{\tau^{(s)}}(X^{(s)})}^{\otimes \log N_{L-1}} \ket{z_L^{(s)}(r_{L-1}^{(s)},a_{L-1}^{(s)})} \left( \otimes_{b \in [\log N_{L-1}]} \ket{C_{\tau^{(s)},b}(X^{(s)})} \right) \right].
\end{align*}
Denote the overall circuit as $U := (\otimes_{s \in [t]} U_{\tau^{(s)}})U_1$, which then has the following action:
\begin{align}
&\left( \otimes_{s \in [t]} \ket{X^{(s)}} \right) \ket{0}^{(t-1)R} \ket{0}^{2 \log N_{L-1}} \nonumber \\
\xrightarrow{U} & \left( \otimes_{s \in [t]} \ket{X^{(s)}} \right) \left(\bigotimes_{s \in [t-1]} \bigotimes_{j \in [R]} \ket{D_{s,j}} \right) \ket{\textsf{Const}_{\tau^{(s)}}(X^{(s)})}^{\otimes \log N_{L-1}} \left( \otimes_{b \in [\log N_{L-1}]} \ket{C_{\tau^{(s)},b}(X^{(s)})} \right) \label{eq:def_U_IterMultInd},
\end{align}
where we have noted that $\ket{Y^{(s)}}$, $\ket{z_L^{(s)}(r_{L-1}^{(s)}, a_{L-1}^{(s)}}$ are a part of the input $\ket{X^{(s)}}$, and collected all the ancilla qubits required for $U_1$ and $U_{\tau^{(s)}}$ in the first line. For a specified $b \in [\log N_{L-1}]$, measuring the qubits which hold the state $\ket{D_{s,j}}, \forall s \in [t-1], j \in [R]$ and $\ket{C_{\tau^{(s)},b}}, \forall s \in [t]$ in the computational basis and then applying the $\AND$ logical gate over the outcomes then gives us the desired result of Eq.~\eqref{eq:func_Fb}.\footnote{Alternately, if given a $\beta \in [N_{L-1}]$, we can also output $[\beta = \IMIndex(\mathbf{X})]$ by adding $\beta$ on $\log N_{L-1}$ ancilla qubits, measuring these in the computational basis and performing a classical $\AND$ gate on these along with all the measurement outcomes for each $b \in [\log N_{L-1}]$ from before.} 
% This overall circuit is summarized in Figure~\ref{fig:circ_iterindex_avec_copies}.

We can now once again analyze the qubit count, depth, and gate complexity of $U$. For each $s \in [t]$, the total number of input qubits involved in defining $\ket{X^{(s)}}$ is:
\begin{equation*}
R =
O\left(\log N_0 + \sum_{\ell=1}^{L-1}\log n_\ell + \sum_{\ell=0}^{L-1}\log N_\ell \right)=O\left(\sum_{\ell=0}^{L-1}\log N_\ell\right),
\end{equation*}
where we have used the fact that the size of the input and output to $z_\ell$ require $\ceil{\log N_{\ell-1}}$ bits each to specify. Using the fact that $N_\ell=m \cdot \prod_{j=1}^{\ell} n_j$, $n_\ell(T, L)= K^{4 \cdot 8^{L-\ell-1}}$, $m=K^{8^L}$ and that $L$ is assumed to be a constant, we then have that:
\begin{equation*}
R = O\left(L 2^{3L} \log K \right).
\end{equation*}
Taking into account the size of $t$ and our choices for the parameters, the total number of qubits used is:
\begin{equation*}
tR = K^{O(L 8^L)}.
\end{equation*}
The contribution to the gate complexity of $U$ is due to $U_1$, which has gate complexity $O(tR) = K^{O(L 8^L)}$, and $U_{\tau^{(s)}}$ for all $s \in [t]$, each requiring:
\begin{equation*}
O\left(\sum_{\ell = 0}^{L-1} \log N_\ell\right) = O(L 2^{3L} \log K)
\end{equation*}
one- and two-qubit gates (from the proof of Theorem~\ref{thm:iter-index_with_copies}). Overall, the gate complexity of $U$ is therefore: 
\begin{equation*}
O(t L 2^{3L} \log K + K^{L8^L}) = O\left(L 2^{3L} K^{L 8^L} \log K\right).
\end{equation*}
Since all the $U_s$'s are applied in parallel, the depth of $U$ before application of the $\AND$ gate is same as depth of any particular $U_s$, which is $O\left(L + \log L + \log \log K\right)$, in addition to the depth of $U_1$ which is $O(1)$. Overall, the depth of $U$ is therefore:
\begin{equation*}
O\left(L + \log L + \log \log K\right).
\end{equation*}
This places $\IMIndex$ in $\land \circ \QNCz[\log \log K]$, and completes the proof.
\end{proof}

\subsection{Classical lower bound}
\subsubsection{Communication problem and token count} We use the autoregressive model used in \cite{chen2024theoretical} to show unconditional lower bounds on transformer architectures, but adapt it to apply to the $\IMIndex$ function. Our autoregressive model consists of $L+2$ \emph{macro-players}, who are indexed $-1, 0, 1, \ldots, L$. The protocol takes place over $L$ \emph{communication epochs}, which are inherently directional. For each index $i$, macro-player $i$ receives as input the concatenation of $t$ \emph{sub-player} copies. In particular:
\begin{enumerate}
  \item the input of macro-player $-1$ is $(w^{(1)}, \ldots, w^{(t)}) \in A_{-1}^{\,t}$;
  \item the input of macro-player $0$ is $(z_0^{(1)}, \ldots, z_0^{(t)}) \in A_0^{\,t}$;
  \item for $\ell \in [L]$, the input of macro-player $\ell$ is $(z_\ell^{(1)}, \ldots, z_\ell^{(t)}) \in A_\ell^{\,t}$.
\end{enumerate}
At the end of the protocol, i.e., at the end of epoch $L$, macro-player $-1$ must output the value of $\IMIndex$ on the given instance.

\paragraph{Token counts.} We now estimate the total number of tokens taken up by each player in the communication protocol. For this, it will be helpful to use the parameters as counted in the original protocol used in \cite{chen2024theoretical}, which used the iterated index problem. Let $m_{\textup{orig}}^{(i)}$ denote the token count of player $i$ in the original iterated index problem, i.e.:
\begin{equation*}
m_{\textup{orig}}^{(i)} = \begin{cases} N_{i-1} & i \in [L], \\ 1 & i \in \{-1, 0\}. \end{cases}
\end{equation*}
In our problem, i.e., the iterated multi-index problem, on the other hand, macro-player $i$ holds $t$ concatenated copies of the same set of tokens. So now, his token count becomes $m_{\textup{new}}^{(i)} := t \cdot m_{\textup{orig}}^{(i)}$, which means the total input length (i.e., number of tokens) becomes:
\begin{equation}
\label{eq:nnewnorig}
n(T, L) := \sum_{i=-1}^{L} m_{\textup{new}}^{(i)} = t \cdot n_{\textup{orig}}(T, L),
\end{equation}
where $n_{\textup{orig}}(T, L) := \sum_{i=-1}^{L} m_{\textup{orig}}^{(i)}$ is the total (token) length of the original iterated index instance. 

Now, recall that the autoregressive communication model relies on a \emph{bandwidth} parameter $B \in \mathbb{N}$ (which relates it to the width of the transformer in Lemma~\ref{lem:cpw-reduction}). By definition, no response that player $i$ receives from player $j$ can exceed the bandwidth bit-count. Specifically, we must have, for each $\ell \in [L]$, and any pair of players $i$ and $j$:
\begin{equation}
\left| \Pi_{j,i}^{(\ell)} \right| \leq 2B \cdot m^{(i)}.
\label{eq:bandwidth}
\end{equation}
As we stated earlier, for the protocol to succeed, at the end of epoch $L$, player $-1$ must output a value that depends only on $X_{-1}^{(L)}$. We are now ready to prove our lower bound, after a couple of lemmas we state first.

\begin{lemma}
\label{lem:tbound}
The copy complexity $t = t(T, L)$, as defined in Section \ref{sec:multiindex}, satisfies the following~bound:
\begin{equation*}
t \leq C_L \cdot T^{c_L},
\text{ where } c_L = 2^{O(L)}, \text{ } C_L = 2^{2^{O(L)}}.
\end{equation*}
In particular, for every constant $L$, we have that $t$ is a fixed polynomial in $T$.
\end{lemma}

\begin{proof}
We first bound each $N_\ell$ in terms of $K$. By definition, $N_\ell = m \cdot \prod_{j=1}^\ell n_j$, so the exponent of $K$ in $N_\ell$ is:
\begin{equation*}
\sum_{j=0}^{L-1} 8^{j+1} + \sum_{j=1}^{\ell} 4 \cdot 8^{L-j-1}.
\end{equation*}
Observe that this is a sum of two geometric series, each with ratio $8$. Clearly, both sums are $2^{O(L)}$, and so, for all $\ell$, we have $N_\ell = K^{2^{O(L)}}$. Therefore, we have:
\begin{equation*}
t = \prod_{\ell=0}^{L-1} N_\ell = \left( K^{2^{O(L)}} \right)^L = K^{L \cdot 2^{O(L)}} = K^{2^{O(L)}}.
\end{equation*}
Substituting $K = (TL)^8$, we obtain:
\begin{equation*}
t = \left( (TL)^8 \right)^{2^{O(L)}} = T^{8 \cdot 2^{O(L)}} \cdot L^{8 \cdot 2^{O(L)}} = C_L\cdot T^{c_L},
\end{equation*}
where we have set $C_L := L^{8 \cdot 2^{O(L)}}$ and $c_L := 8 \cdot 2^{O(L)}$ and absorbed all other implicit constants independent of $L$ into them appropriately. Note that both $C_L$ and $c_L$ only depend on $L$.
\end{proof}

\subsubsection{Transformer lower bound}
\begin{lemma}
\label{lem:reduction}
Fix $T \geq 1$. Suppose that there is a deterministic autoregressive protocol $\Pi_{\textup{multi}}$ that solves $\IMIndex$ on bandwidth $T$ in $L$ epochs for every instance. Then, there is a deterministic autoregressive protocol $\Pi_{\textup{orig}}$ that solves $\IIndex$ on bandwidth $T' := Tt$ in $L$ epochs for every instance.
\end{lemma}

\begin{proof}
Consider an arbitrary instance  $(w, z_0, z_1, \ldots, z_L)$ of $\IIndex$, where the tuple is distributed across players $-1, 0, 1, \ldots, L$ in the usual way (i.e., player $i$ holds the $i$-th component of the vector). Now, we can have the players in the protocol $\Pi_{\textup{orig}}$ simulate $\Pi_{\textup{multi}}$ on the \emph{diagonal} instance $\left(\{w^{(s)}\}_{s \in [t]}, \{z_0^{(s)}\}_{s \in [t]}, \{z_1^{(s)}\}_{s \in [t]}, 
\ldots, \{z_L^{(s)}\}_{s \in [t]}\right)$ of $\IMIndex$:
\begin{equation*}
w^{(s)} := w, \qquad z_0^{(s)} := z_0, \qquad z_\ell^{(s)} := z_\ell  
\end{equation*}
for all $s \in [t],\ell \in [L]$.  Concretely:
\begin{enumerate}
  \item For all $i \in \{-1, 0, 1, \ldots, L\}$, player $i$ \emph{locally} constructs the $t$-fold concatenation of its input, and treats the resulting string as the input to macro-player $i$ in $\Pi_{\textup{multi}}$. Note that this is purely local; no communication between players is needed.
  \item The players now run $\Pi_{\textup{multi}}$ verbatim, acting as the macro-players in that protocol. In particular, whenever a message $\Gamma_{i,j}^{(\ell)}$ or response $\Pi_{j,i}^{(\ell)}$ is sent in $\Pi_{\textup{multi}}$ between the macro-players, the same bits are passed in $\Pi_{\textup{orig}}$ between the same players.
  \item At the end of epoch $L$, player $-1$ in $\Pi_{\textup{orig}}$ outputs the output of macro-player $-1$ in $\Pi_{\textup{multi}}$.
\end{enumerate}

Note that the resulting protocol is correct: on the diagonal instance the consistency predicate $E = 1$ (i.e., all the $x$s given to individual players are the same), and so the output of $\Pi_{\textup{multi}}$ equals $\IIndex(w, z_0, z_1, \ldots, z_L)$, as required. Let us examine the bandwidth used by $\Pi_{\textup{orig}}$. Suppose player $i$ in $\Pi_{\textup{orig}}$ has $m_{\textup{orig}}^{(i)}$ tokens; he locally concatenates $t$ copies of this, creating $m_{\textup{new}}^{(i)} = t\cdot m_{\textup{orig}}^{(i)}$ tokens, and therefore each of his communicated strings is blown up by a factor of $t$. By the bandwidth constraint of the autoregressive model applied to $\Pi_{\textup{multi}}$, we have:
\begin{equation*}
\left| \Pi_{j,i}^{(\ell)} \right| \leq 2T \cdot m_{\textup{new}}^{(i)} = 2(Tt) \cdot m_{\textup{orig}}^{(i)},
\end{equation*}
which is precisely within the bandwidth-$(Tt)$ budget of $\Pi_{\textup{orig}}$. The number of epochs remains unchanged at $L$.
\end{proof}

We are now ready to state and prove the lower bound. The instance family is parameterized by the transformer's bandwidth itself: the parameters $K, m, \{n_\ell\}, \{N_\ell\}, t, n$ are defined exactly as above (as functions of $T = Hdp$ and $L$). This is the standard parameterization of~\cite{chen2024theoretical}, in which the hard instance is tailored to the candidate model size. To prove the theorem, we first observe that on top of the parameters involved in $T$-parameterized $\IMIndex$, our quantum upper bound also had an extra parameter $\beta \in [N_{\ell-1}]$ specified in the problem and the goal was to check if $\beta$ equals the output of the $T$-parameterized $\IMIndex$. Since we prove lower bounds against \emph{determiistic} transformer models (for which we showed the $\wedge\circ \QNCz[\log \log n]$ upper bound. However, note that checking this equality amounts to \emph{computing} the output of $T$-parameterized $\IMIndex$, and hence below we prove lower bounds on transformers computing this function.\footnote{Concretely, we remark that when reducing the transformer lower bound to the communication lower bound, we will provide this $\beta$ parameter to the $L+2$-th player, who needs to check if $\beta$ equals the output of $T$-parameterized $\IMIndex$ which depends on parameters defined for the previous $L-1$ players.}

\begin{theorem}
\label{thm:maintransformer}
Let $L \geq 2$ be constant. Let $\mathscr{T}$ be any $L$-layer decoder-only transformer (with $H$ heads, embedding dimension $d$, and precision $p$, with width $T := Hdp$) which solves every instance of the $T$-parameterized $\IMIndex$ family. Then, the width $T$ of $\mathscr{T}$ satisfies:
$$
T \geq n(T, L)^{2^{-O(L)}},
$$
where $n(T, L)$ is the number of tokens (i.e., context length) needed to specify the instance family.
\end{theorem}

\begin{proof}
Since $\mathscr{T}$ solves the $T$-parameterized $\IMIndex$ family on every instance, by Lemma~\ref{lem:cpw-reduction}, there is an $L$-epoch autoregressive protocol $\Pi_{\textup{multi}}$ of bandwidth $T$ solving the same task.\footnote{Actually, Lemma~\ref{lem:cpw-reduction} is about $\IIndex$ rather than $\IMIndex$, but we can adapt the same result for $\IMIndex$ as well with no change in the proof.} By Lemma~\ref{lem:reduction} (applied with bandwidth $T$ and copy parameter $t$), there is an $L$-epoch autoregressive protocol $\Pi_{\textup{orig}}$ of bandwidth $Tt$ solving $\IIndex$ on the same parameter $T$. Applying  Theorem~\ref{thm:cpw-main} to $\Pi_{\textup{orig}}$, we get:
\begin{equation}
Tt \geq n_{\textup{orig}}^{\alpha_L}, \text{ where } \alpha_L = 2^{-O(L)}.
\label{eq:cpw-applied}
\end{equation}
By Lemma~\ref{lem:tbound}, $t = C_L T^{c_L}$ with $c_L = 2^{O(L)}$, and $C_L = 2^{2^{O(L)}}$. Substituting into~\eqref{eq:cpw-applied}, we get:
\begin{align}
\label{eq:implicationofnewprotocol}
C_L \cdot T^{c_L + 1} \geq n_{\textup{orig}}^{\alpha_L} 
%\implies T \geq \left( n_{\textup{orig}}^{\alpha_L} / C_L \right)^{1/(c_L+1)}.,
\end{align}
By definition of $n$ and Eq.~\eqref{eq:nnewnorig}, we have $n = t \cdot n_{\textup{orig}} = C_L T^{c_L} \cdot n_{\textup{orig}}$. Hence, we get:
\begin{equation*}
n_{\textup{orig}} \geq \frac{n}{C_L T^{c_L}}.
\end{equation*}
Substituting this into~\eqref{eq:implicationofnewprotocol} yields:
$$
T^{c_L + 1} \geq \frac{1}{C_L} \cdot \left( \frac{n}{C_L T^{c_L}} \right)^{\alpha_L} = \frac{n^{\alpha_L}}{C_L^{1 + \alpha_L} \cdot T^{c_L \alpha_L}}\implies T^{c_L + 1 + c_L \alpha_L} \geq \frac{n^{\alpha_L}}{C_L^{1 + \alpha_L}}.
$$
The exponent on the left-hand side is $c_L(1 + \alpha_L) + 1 \leq 2 c_L + 1 = 2^{O(L)}$ (since $\alpha_L \leq 1$). For constant $L$, $C_L^{1+\alpha_L}$ is a constant independent of $n$. Taking roots and absorbing the constant into the asymptotic, we get:
\begin{equation*}
T \geq n^{\alpha_L / (c_L + 1 + c_L \alpha_L)} = n^{2^{-O(L)}}.
\end{equation*}
Since $T = Hdp$ by definition, the claimed bound is immediate.
\end{proof}

\subsection{Obstructions in further improvement}
In this section, we show that if a function $f$ is \emph{deterministically} computable by $\QNCz$ circuit of depth $\log\log n$ (given an advice state) and followed by an $\textsf{AND}$ gate, then this $f$ is already in $\ACz$. Note that $\ACz$ functions roughly capture constant-precision transformers (which are quite weak from an expressivity point of view); indeed, we wish to have lower bounds against at least log-precision transformers. However, since our main result already says that an $O(\log\log n)$-depth $\QNCz$ circuit with a single $\textsf{AND}$ gate computes functions outside of such log-precision transformers, this shows that even increasing the depth of our $(\log\log n)$-circuit by a constant multiple has an effect on its computational power. This implies that our separation is essentially optimal.\footnote{As we mentioned in the introduction, our result crucially uses that $f$ is exactly computed by $\land \circ \QNCz/\textsf{qpoly}$. If $f$ was only approximated by these circuits, then such a result is unclear.} 
Indeed, if we increase the depth from $\log \log n$ to  $c\cdot \log \log n$ for some constant $c>1$, then we break past $\ACz$ to a function that is hard for transformers, i.e., needs width $n^{\Omega(1)}$.
\begin{theorem}
Let $n\geq 1$ and $d\leq \log \log n +O(1)$. 
Every function $f:\{0,1\}^n\rightarrow \{0,1\}$ computable by  $\land \circ \QNCz[d]/\textsf{qpoly}$ can be computed in $\ACz$.
\end{theorem}
\begin{proof}
Let $U$ be the depth-$d$ $\QNCz$ circuit, let $|\psi_n\rangle$ be the
advice state, and let the input state be $|x\rangle|0\rangle|\psi_n\rangle$.
Let $q_1,\ldots,q_m$ be the designated output qubits, where
$m=\textsf{poly}(n)$. For each $j$, let $Q_j$ denote the projector
onto the event that output qubit $q_j$ is measured as $1$, and define:
\begin{equation*}
P_j := U^\dagger Q_j U.
\end{equation*}
The projectors $P_j$ commute, since the corresponding output measurements commute. The final $\textsf{AND}$ gate accepts with probability:
\begin{equation*}
    \left\langle x,0,\psi_n \right|
        \prod_{j=1}^m P_j
    \left| x,0,\psi_n \right\rangle .
\end{equation*}
Because the computation is exact, this acceptance probability is either $0$ or $1$, and it equals $f_n(x)$. Since $\prod_j P_j$ is a projector, the acceptance probability is $1$ if and only if:
\begin{equation*}
    \prod_{j=1}^m P_j |x,0,\psi_n\rangle = |x,0,\psi_n\rangle .
\end{equation*}
For commuting projectors, this is equivalent to:
\begin{equation*}
    P_j |x,0,\psi_n\rangle = |x,0,\psi_n\rangle
    \qquad\text{for every }j\in[m].
\end{equation*}
Thus
$    f_n(x)=1
$ if and only if $\bigwedge_{j=1}^m g_j(x)=1$. The light cone of the output qubit $q_j$ contains at most $2^d$ input qubits, since the circuit has bounded fan-in and depth $d$. Hence $P_j$ acts
nontrivially only on those at most $2^d$ input qubits, together with some ancilla and advice qubits. The ancilla and advice states are fixed
independently of $x$. Therefore the truth value of $g_j(x)$ depends only on the input bits in this light cone. Thus each $g_j$ is a Boolean function of at most $2^d$ input bits. Since any Boolean function on $k$ bits can be written as a DNF of size at most $k2^k$, applying this fact with $k\le 2^d$ yields the fact that each $g_j$ has a depth-$2$
$\ACz$ circuit of size $2^{O(2^d)}$. Taking the $\textsf{AND}$ over all
$m=\textsf{poly}(n)$ such circuits gives an $\ACz$ circuit for $f_n$ of size $\poly(n,2^{2^d})$. So, if $d\leq \log \log n$, then $f_n\in\ACz$, as desired.
\end{proof}
\section{Distributional separation: quantum circuits versus DLMs}\label{sec:dlms}
In this section, we begin by describing our candidate problem for the distributional separation, and then prove our upper and lower bounds, with the latter being the bulk of the technical content.

\subsection{Main result}
\paragraph{Candidate problem.} We first define our candidate problem. 
For $B \in \mathbb{N}$, consider the set of length-$B$ strings which have even parity, and define the uniform distribution on this set as follows:
\begin{equation*}
 U_B^\oplus := \textsf{unif} \{x\in\{0,1\}^B : |x|=0\}.
\end{equation*}
For $n=MB$, define the \emph{block-parity} distribution:
\begin{equation*}
  D_{\blk}:=\bigotimes_{b=1}^M U_B^\oplus,
\end{equation*}
i.e., $D_{\blk}$ is uniform on the set of strings $\{x^1, \ldots, x^M\} \subseteq \{0,1\}^B$ such that $\bigoplus_{i \in [B]} \, x^j_i=0$ for all $j\in [M]$. Our proof will require us to keep careful track of these parameters, and we will use the following settings:
\begin{align}
\label{eq:parametersetting}
  B=n^{0.01},\qquad M=n^{0.99},\qquad s,R=n^{0.9},
\end{align}
where $s$ is the number of CoT/workspace tokens, and $R$ is the total number of output-token revision/remasking events. Note that by this definition, $D_\blk$ is defined on $n$ bits and all the other parameters such as CoT, remasking, revision will be $o(n)$. These parameter choices will become clear when discussing our lower bounds.\footnote{We remark that one could have also chosen $B,M$ such that $B\cdot M=n$ and $s,R = o(M)$, so we could have also handled $s,R=n^{1-\varepsilon}$ for arbitrarily small constant $\varepsilon$. We chose the parameters above for the sake of simplicity in presentation.} With these definitions we prove our main result.

\begin{theorem}
\label{thm:mainDLM}
 The distribution $D_{\textsf{blk}}$ satisfies the following:
\begin{enumerate}
\item $D_{\blk}$ can be sampled exactly by a $\QNCz$ circuit.
\item  No constant-round shallow  $\DLM(\GCz, \log n)$ with $s=n^{0.9}$ CoT tokens and $R=n^{0.9}$ output-token remasking/revision events can sample $D_{\blk}$ within constant total variation distance.
\end{enumerate}
\end{theorem}
We now make a few remarks about our candidate problem and the lower bound. The $\DLM$  samplers we consider have no input; on any (vacuous) input, their goal is just to produce an output that is a uniform sample from the distribution $D_{\blk}$. The $\DLM$ that we consider here is the same one described in Section~\ref{sec:intro}, wherein each token generated within the same round is generated in a conditionally independent manner, given the current full transcript. The constraints on this $\DLM$ are as follows:
\begin{enumerate}[$(i)$]
\item We allow only $s =n^{0.9}$ hidden transcript bits. 
\item The output schedule may be \emph{fully adaptive} and may depend on hidden token values.
%\item The first lower bound (Theorem~\ref{thm:no-revision}) forbids remasking/revision. 
\item For each transcript cell $\tau$,  the corresponding token is computed by a polynomial-size $\GCz(k)$ circuit with depth $O(1)$, where $k = O(\log n)$ (i.e., consists of $\textsf{GC}_k$ gates for $k=O(\log n)$).
\end{enumerate}

\subsection{Quantum upper bound}\label{sec:qub_sample}
We first show that the block-parity distribution is easy for shallow quantum circuits. The key ingredient is the poor man's cat state construction of Watts-Kothari-Schaeffer-Tal~\cite{watts2019exponential}, which can be prepared in $\QNCz$.

\begin{lemma}\label{lem:qub_sample_block_parity}
For every $B$, the distribution $U_B^\oplus$ can be sampled exactly by a $\QNCz$ circuit. Consequently, $D_{\blk}$ can be sampled exactly by a $\QNCz$ circuit.
\end{lemma}
\begin{proof}
We first describe how to prepare the poor man's cat state in $\QNCz$. We start with $B$ qubits initialized to
$\ket{0}^{\otimes B}$. Applying Hadamard to all qubits gives
$\ket{+}^{\otimes B}
=
\frac{1}{\sqrt{2^B}}\sum_{x\in\{0,1\}^B}\ket{x}$. Next, we measure the commuting parity observables
$Z_iZ_{i+1}$, for $
i=1,\dots,B-1$.\footnote{By this we mean, apply the two-qubit operator $\Pi_{s_i}=\ketbra{00}{00}+\ketbra{11}{11}$ and $\Pi_{1+s_i}=\ketbra{01}{01}+\ketbra{10}{10}$ onto qubits $i,i+1$.}
Let the measurement outcomes be
$s_i\in\{0,1\}$, 
where $s_i=0$ corresponds to eigenvalue $+1$ (i.e.,
$x_i\oplus x_{i+1}=0$), 
and $s_i=1$ corresponds to eigenvalue $-1$ (i.e.,
$x_i\oplus x_{i+1}=1$). These $B-1$ parity constraints determine the string uniquely up to its bitwise complement. More precisely, there is some string
$z\in\{0,1\}^B$, such that $z$ and $\bar z$ are the \emph{only} two solutions to all constraints (in fact, they would be of the form $z_1=0,z_2=s_1,z_3=s_1+s_2,\ldots$ and $z_1=1,z_2=1+s_1,z_3=1+s_1+s_2,\ldots$ where all the additions are modulo $2$). Therefore, after applying these $B-1$ projective measurements, the post-measurement state is:
\begin{equation}\label{eq:post_meas_state}
\frac{\ket z+\ket{\bar z}}{\sqrt2},
\end{equation}
which is referred to as the \emph{poor man's cat state}. The value of $z$ depends on the measurement outcomes $(s_1,\dots,s_{B-1})$, but this will not matter for the sampling task. Now, finally, we apply Hadamard to all $B$ qubits, and measure in the computational basis. For any outcome $y\in\{0,1\}^B$, the amplitude equals:
\begin{equation*}
\left\langle y\middle|
\textsf{Had}^{\otimes B}
\frac{\ket z+\ket{\bar z}}{\sqrt2}
\right\rangle
\propto
(-1)^{y\cdot z}
+
(-1)^{y\cdot\bar z}.
\end{equation*}
Since
$y\cdot\bar z
=
y\cdot z+\sum_{i=1}^B y_i
\pmod 2$, 
we obtain:
\begin{equation*}
(-1)^{y\cdot z}
+
(-1)^{y\cdot\bar z}
=
(-1)^{y\cdot z}
\left(
1+(-1)^{|y|}
\right).
\end{equation*}
Therefore, the amplitude vanishes whenever $|y|$ is odd, and has equal magnitude for every even-parity string $y$. Hence, the measurement outcome is distributed exactly as
$y\sim U_B^\oplus$. 

Let us quickly describe the corresponding quantum circuit. Consider the input of $\ket{+}^B$ on the $B$ system qubits, which we will denote by $q_i$. Let us add $B-1$ many ancilla qubits, which we denote by $a_i, \forall i \in [B-1]$, each initialized to $\ket{0}$. To carry out the measurement corresponding to $Z_i Z_{i+1}, \forall i \in [B-1]$, we will first apply a layer of $\CNOT$s with control over $q_i$ to the target $a_i$ for all $i \in [B-1]$ followed by another layer of $\CNOT$s with control over $q_{i+1}$ to the target $a_i$ for all $i \in [B-1]$. At this point, the state over the system and ancilla qubits is $\frac{1}{\sqrt{2^B}} \sum_x \ket{x} \Big(\otimes_{i \in [B-1]} \ket{x_i \oplus x_{i+1}} \Big)$. Measuring the qubits $a_i, \forall i \in [B-1]$ in the computational basis then results in the post-measurement state of Eq.~\eqref{eq:post_meas_state} over qubits $q_i, \forall i \in [B]$. Finally, applying Hadamards to the qubits $\{q_i\}$ and then measuring all $q_i$ in the computational basis yields outcomes $y \sim U_B^{\oplus}$ as shown above. This circuit is illustrated in Figure~\ref{fig:poor_man_cat_state}.
\begin{figure}[ht]
    \centering
    \[
\Qcircuit @C=0.75em @R=0.55em {
\lstick{q_1:\ket{0}}
    & \gate{H}
    & \ctrl{1}
    & \qw
    & \qw
    & \gate{H}
    & \meter
    & \cw
    & \rstick{y_1}
\\
\lstick{a_1:\ket{0}}
    & \qw
    & \targ
    & \targ
    & \meter
    & \cw
    & \cw
    & \cw
    & \rstick{s_1}
\\
\lstick{q_2:\ket{0}}
    & \gate{H}
    & \ctrl{1}
    & \ctrl{-1}
    & \qw
    & \gate{H}
    & \meter
    & \cw
    & \rstick{y_2}
\\
\lstick{a_2:\ket{0}}
    & \qw
    & \targ
    & \targ
    & \meter
    & \cw
    & \cw
    & \cw
    & \rstick{s_2}
\\
\lstick{q_3:\ket{0}}
    & \gate{H}
    & \ctrl{1}
    & \ctrl{-1}
    & \qw
    & \gate{H}
    & \meter
    & \cw
    & \rstick{y_3}
\\
\lstick{a_3:\ket{0}}
    & \qw
    & \targ
    & \targ
    & \meter
    & \cw
    & \cw
    & \cw
    & \rstick{s_3}
\\
\lstick{q_4:\ket{0}}
    & \gate{H}
    & \ctrl{1}
    & \ctrl{-1}
    & \qw
    & \gate{H}
    & \meter
    & \cw
    & \rstick{y_4}
\\
\lstick{a_4:\ket{0}}
    & \qw
    & \targ
    & \targ
    & \meter
    & \cw
    & \cw
    & \cw
    & \rstick{s_4}
\\
\lstick{q_5:\ket{0}}
    & \gate{H}
    & \qw
    & \ctrl{-1}
    & \qw
    & \gate{H}
    & \meter
    & \cw
    & \rstick{y_5}
}
\]
\caption{Illustration of circuit for preparing the poor man's cat state and sampling from $U_5^{\oplus}$.}
\label{fig:poor_man_cat_state}
\end{figure}
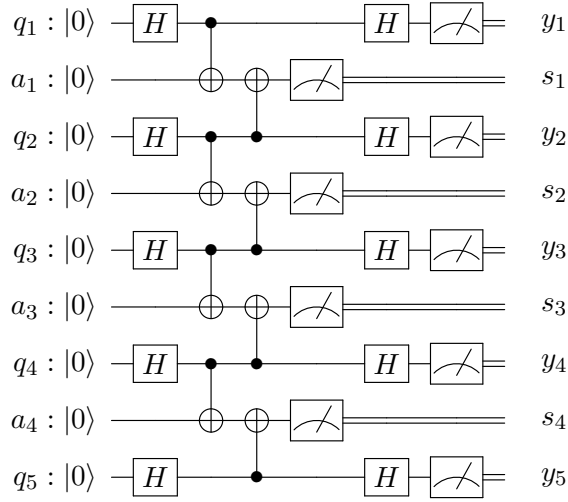
We observe that the poor man's cat state (Eq.~\eqref{eq:post_meas_state}) can be prepared in depth $4$, and overall this gives an exact $\QNCz$ sampler for $U_B^\oplus$. Running this sampler independently and in parallel on all $M$ blocks gives an exact $\QNCz$ sampler for $D_{\blk}$, proving the lemma statement. The overall circuit has depth $4$, gate complexity $O(BM)$, and is over $O(MB)$ qubits.
\end{proof}

\subsection{Lower bound: important lemmas}
For the lower bound, we will use as a black box the following theorem of Kumar~\cite{kumar2023tight}, that shows that $\GCz[\log m]$ circuits cannot approximate the parity function with constant bias, extending the well-known result of Hastad~\cite{haastad2014correlation} that $\ACz$ functions cannot approximate parity with constant~bias. 

\begin{theorem}[\cite{kumar2023tight}]\label{thm:parity-bound}
If $C \colon \{0, 1\}^m \to \{0, 1\}$ is a polynomial-size constant  depth $\GCz[\log m]$  circuit, then:
\begin{equation*}
\Pr_{x \sim \{0, 1\}^m}\bigl[C(x) = x_1 \oplus \cdots \oplus x_m\bigr] \le \frac{1}{2} + 2^{-\Omega(m/\poly\log(m))}. 
\end{equation*}
\end{theorem}

\subsubsection{$\DLM$ run as a combination of distributions}
We now show the core structural fact underlying the lower bound. Recall that, by the model definition, a $\DLM$ follows a sequence of denoising steps before the $n$ output tokens (which we will view as $M$ blocks of $B$ bits) need to produce a string distributed according to $D_\blk$.  Consider a $\DLM$ round $t$, and consider the parity blocks that are completed for the first time in this round. For such a block $b\in [M]$, let $r_b$ denote the number of coordinates in block $b$ that are newly generated in the $t$-th round; equivalently, $r_b$ is the number of block-$b$ coordinates that were still masked in each of the first $t-1$ rounds. Now, observe that  since $D_\blk$ is uniform on every block $b\in [M]$, the newly generated coordinates in the completion round should closely emulate the uniform distribution over an affine parity coset. In particular, the target conditional law for the block $b$ is the distribution:\footnote{Throughout this section, we interchangeably use $\oplus x_i$ and $|x|=\sum_i x_i$ but emphasize that we using the latter, the summation is modulo $2$.}
\begin{equation*}
U_{a_b}^{r_b}
:=
\textsf{unif}\Big\{
x\in\{0,1\}^{r_b}
:
\bigoplus_{i=1}^{r_b} x_i = a_b
\Big\},
\end{equation*}
where $a_b\in \{0,1\}$ is set according to the parity value of the previously revealed coordinates of the block in the first $t-1$ rounds. Now, suppose there are $q\leq M$ such  blocks completed in the $t$-th round. Then, the overall \emph{target} conditional law is the product distribution:
\begin{align}
\label{eq:defnofV}
V=\bigotimes_{b=1}^q U_{a_b}^{r_b}.
\end{align}
In other words, $V$ is a product distribution across the simultaneously completed blocks, but each of the factors $U^{r_b}_{a_b}$ contains an internal parity constraint among the coordinates of that block (determined by $a_b$). Let 
$
D:=\sum_{b=1}^q (r_b-1)
$
be the total dimension of the parity-coset distribution $V$. Consequently, we have:
\begin{equation*}
|\operatorname{supp}(V)|=\prod_{b=1}^q 2^{r_b-1}=2^D.
\end{equation*}
Let $N:=\sum_{b=1}^q r_b$ denote the total number of generated coordinates. The following theorem precisely
quantifies the obstruction created by these ``inter-block'' constraints; roughly speaking, even though the $q$ blocks in $V$ are independent of one another, a distribution with many parity-constrained factors cannot be approximated by a small mixture of fully coordinate-wise product distributions.\footnote{Although this does not seem \emph{a priori} related to $\DLM$s as stated, every round of a $\DLM$ in fact generates product distributions, which is the connection that we will utilize later.}

\begin{theorem}
\label{thm:approximable}
Let $q \leq M \in \mathbb{N}$ be the number of blocks completed in the $t$-th round. For $r_b\geq 2, a_b\in \{0, 1\}$,  define:
\begin{equation*}
  U^{r_b}_{a_b}=\textsf{unif}\{x\in\{0, 1\}^{r_b}:|x|=a_b \pmod 2\}. 
\end{equation*}
Define $V$ be as in Eq.~\eqref{eq:defnofV},  $N = \sum_{b=1}^q r_b$, and $D = N - q$. Suppose $W$ is a distribution on $\{0,1\}^N$ with $\supp(W)\subseteq\supp(V)$ and $H(W)\geq D-\kappa$ for some $\kappa\geq 0$.\footnote{ Note that the choice $W=V$ corresponds to $\kappa=0$.}  Then, there exists a constant $\eta>0$ such that whenever $P$ is a mixture of $K$ product distributions (i.e., 
$P=\sum_{h=1}^K\lambda_h Q_h$ for product distributions $Q_h$) satisfying $\textsf{TV}(P, {W})\le\eta$, then there exist constants $ c_1,c_2$ satisfying:
\begin{equation*}
  \log K\ge c_1 D -\,\kappa\,-c_2. 
\end{equation*}
% Equivalently,  if $K\le2^s$,  then
% $$
%   D\le C_\eta(s+\log N). 
% $$
\end{theorem}
\begin{proof}
Let $\Pi_b$ denote the set of fresh coordinates of block $b$ revealed in the
current round, so $|\Pi_b|=r_b$.  In Eq.~\eqref{eq:defnofV}, the $b$-th
tensor factor $U_{a_b}^{r_b}$ is indexed by these coordinates, after choosing
an arbitrary ordering of $\Pi_b$.  Thus, identifying
$\{0,1\}^{r_b}$ with $\{0,1\}^{\Pi_b}$, we have:
\begin{equation*}
\supp(V)
=
\left\{
x\in\{0,1\}^{\Pi}:
\bigoplus_{i\in \Pi_b}x_i=a_b
\text{ for every }b
\right\},
\qquad
\Pi:=\bigcup_{b=1}^q \Pi_b.
\end{equation*}
Then $V$ is uniform on $S$, and $|S|=\prod_b 2^{r_b-1}=2^D$ (since each linear constraint on $x\in \{0, 1\}^{r_b}$ is satisfied by exactly half the variables). 
Since $V$ is supported on $S$,  we have $V(S)=1$ and $V(S^c)=0$. Since $\supp(W)\subseteq S$, we also have $W(S)=1$ and $W(S^c)=0$. Therefore, we have:
\begin{align}
\label{eq:upperboundonPSc}
P(S^c)
=
|P(S^c)-{W}(S^c)|
\le
\textsf{TV}(P, {W})
\le
\eta. 
\end{align}
Next,  since $P$ and ${W}$ are distributions on $\{0, 1\}^N$, Lemma~\ref{lem:fannes} implies:
\begin{equation*}
|H(P)-H({W})|
\le
\eta\cdot\log(2^N-1)+h_2(\eta)
<
\eta N+h_2(\eta), 
\end{equation*}
where $h_2(\cdot)$ denotes the binary entropy function. Therefore, we have:
\begin{align}
\label{eq:lowerboundonH(P)}   
H(P)
\ge
H({W})-O(\eta N)-O(1){\;\geq\; D-\kappa-O(\eta N)-O(1)}\geq (1-c_1\eta)D{\,-\,\kappa}-O(1), 
\end{align}
where the second inequality uses the fact that $H(W)\geq D-\kappa$ (an equality when $W=V$, since $V$ is uniform on a set $S$ of size $2^D$), and the last inequality uses the fact that $N=\sum_b r_b\leq 2\sum_b(r_b-1)=2D$. On the other hand, we can sample from $P$ by first choosing $h\in [K]$ and then sampling $Q_h$; therefore, we have:
\begin{align}
\label{eq:firstupperboundonHP}
  H(P)\le \log K+\E_h H(Q_h). 
\end{align}
Let $\alpha_h:=Q_h(S^c)$,  
so that $\alpha_h$ is the probability that the product distribution $Q_h$ violates at least one
of the parity constraints defining $S$.  Since $P=\sum_{h=1}^K \lambda_h Q_h$, 
we have:
\begin{align}
    \label{eq:expectationonh}
\eta\geq P(S^c)
=
\sum_h \lambda_h \alpha_h
=
\E_h[\alpha_h], 
\end{align}
where the first inequality used Eq.~\eqref{eq:upperboundonPSc}.  Hence  $\E_h[\alpha_h]\le \eta$.  Fix parameters
$\tau=10^{-3}$ and $\eta=1/20(c_1+2\cdot 10^3)$, where $c_1$ is the constant that appears in Eq.~\eqref{eq:lowerboundonH(P)}.  
We now split the mixture components into ``good'' and ``bad'' components according to whether $\alpha_h\le \tau$ or $\alpha_h\geq \tau$.  

\paragraph{Bad components. } By Markov's inequality on Eq.~\eqref{eq:expectationonh}, first observe that:
\begin{equation*}
\Pr_h[\alpha_h>\tau]
\le
\E_h[\alpha_h]/{\tau}
\le
\eta/\tau. 
\end{equation*}
So, the total weight over all possible $h$ on the bad components is at most $\eta/\tau$.  Now every distribution on $\{0, 1\}^N$ has entropy at most $N$. Hence, the entropy contributed by each of the bad components to the overall average entropy $\E_h[H(Q_h)]$ is at most $(\eta/\tau)\cdot N \leq 2\eta/\tau\cdot D$ (where we have used the fact that $N\leq 2D$, established earlier).  

\paragraph{Good components. } It remains to bound the entropy of the good components, namely, those satisfying
$
Q(S^c)\le \tau. 
$
Fix such a product distribution
$
Q=\prod_i Q_i,
$
and for each block $b$, let $\varepsilon_b$ denote the probability that $Q$ violates the parity
constraint in block $b$.  Equivalently:
$$
1-\varepsilon_b
=
Q\!\left[
\bigoplus_{i\in \Pi^b}x_i=a_b
\right]. 
$$
Since $Q$ is a product across the coordinates, and the blocks $\Pi^b$ are disjoint, this implies the parity events across different blocks are actually independent. Therefore, we have:
\begin{equation*}
Q(S)
=
\prod_b (1-\varepsilon_b). 
\end{equation*}
Since $Q(S^c)\le\tau$, we have $Q(S)\ge 1-\tau$, i.e., $\prod_b(1-\varepsilon_b)\ge 1-\tau$. Again, by taking logarithms and using the identity $-\log(1-u)\geq u$ for $u\in [0, 1)$,  we get that:
\begin{equation*}
\sum_b \varepsilon_b
\le
-\log(1-\tau)\leq 2\tau, 
\end{equation*}
using the fact that $\tau$ is a sufficiently small constant (e.g., our choice of $10^{-3}$).  We now analyze a single block $\Pi^b$. 
Write $p_i:=Q_i(1)$, and $t_i:=\min\{p_i, 1-p_i\}$.
Recall that for independent Bernoulli variables, $X_i\sim\mathrm{Bernoulli}(p_i)$, we have:
\begin{equation*}
\Pr\!\left[\bigoplus_{i\in \Pi_b}X_i=a_b\right]
=
\frac12\left(1+(-1)^{a_b}\prod_{i\in \Pi_b}(1-2p_i)\right).
\end{equation*}
which follows from the fact that:
\begin{equation*}
\E\left[(-1)^{\oplus_{i\in \Pi_b}X_i}\right]
=
\prod_{i\in \Pi_b}\E[(-1)^{X_i}]
=
\prod_{i\in \Pi_b}(1-2p_i).
\end{equation*}
Since the left hand side is precisely $1-\varepsilon_b$, we have that
\begin{equation*}
1-\varepsilon_b = \prod_{i \in \Pi_b}(1-2t_i) \geq 1 - 2 \sum_{i \in \Pi_b} t_i,
\end{equation*}
where the second inequality uses the identity
$\prod_i(1-a_i)\ge 1-\sum_i a_i$ for $a_i\in[0,1]$. Taking logarithms on both sides and using the identity $- \log(1-u) \geq u$ for $u \in [0,1)$ again as before, we obtain
\begin{equation*}
\sum_{i\in \Pi^b} t_i
\le
C\varepsilon_b
\end{equation*}
for some absolute constant $C$ provided $\varepsilon_b\le1/10$.  Hence,  the entropy of the block equals:
\begin{equation*}
\sum_b H(Q_{\Pi^b})
=
\sum_b\sum_{i\in \Pi^b} h_2(p_i)
=
\sum_b\sum_{i\in \Pi^b} h_2(t_i)\leq \sum_b\sum_{i\in \Pi^b} t_i\log\frac et_i\leq C\sum_b\varepsilon_b\log\frac{er_b}{\varepsilon_b}\leq C\cdot \tau \log eN/\tau, 
\end{equation*}
using $h_2(u)=h_2(1-u)$ (by definition of binary entropy), and also the concavity of the map $u\mapsto u\log(e/u)$, and then finally the fact that $\sum_b e_b\leq \tau$. Therefore, every good component satisfies $H(Q_h)\le C\tau\log\frac{eN}{\tau}$.  

Combining the good and bad components, we obtain:
\begin{equation*}
\E_h[H(Q_h)]
\le
C\tau\log ({eN}/{\tau})
+
{2\eta D}/{\tau}. 
\end{equation*}
Finally,  using Eq.~\eqref{eq:firstupperboundonHP} and Eq.~\eqref{eq:lowerboundonH(P)},  we get:
\begin{align*}
(1-c\eta)D{\,-\,\kappa}-O(1)\leq H(P)&\le \log K+\E_h[H(Q_h)] \leq \log K+ C\tau\log\frac{eN}{\tau}
+\frac{2\eta}{\tau}D,
\end{align*}
which lets us conclude that:
\begin{align}
\log K
&\ge
\left(
1-c_1\eta-{2\eta }/{\tau}
\right)D
{\,-\,\kappa}
-
C\tau\log({eN}/{\tau})
-
O(1)\\
&\geq \left(
1-c_1\eta-{2\eta }/{\tau}
\right)D
{\,-\,\kappa}
-
C\tau\log({2eD}/{\tau})
-
O(1)\\
&\geq c_3 D
-\kappa -c_4,
\end{align}
using the fact that $\eta$ and $\tau$ were chosen to be constants, $c_3$ and $c_4$ are new constants, and $D$ is sufficiently large. The theorem follows by relabeling the constants $c_3,c_4$ as $c_1,c_2$ respectively.
\end{proof}

\subsubsection{$\DLM$s for $D_\blk$ induce high entropy distributions}

\paragraph{Notation.} We first define some necessary and convenient notation. Let $\eta:=\TV(P_Y,D_{\blk})$.
\begin{itemize}
\item $P_Y$ is the \emph{output} distribution of the $\DLM$: the law of the final $n$-bit string, ignoring the workspace. By assumption,
$\TV(P_Y,D_{\blk})\le\eta$.
\item $\mathbf{P}$ is the distribution of the \emph{entire run}: every revealed token
(workspace and output), in order. 
\item Let $\widehat{D}=D_{\blk}(\cdot\mid A)$ for $A=\supp(P_Y)$, i.e., the induced distribution of $D_{\blk}$ conditioned on the bits in $A$ having been revealed by $P_Y$ already. %One can $\TV(\widehat{D},D_{\blk})\le\eta$.
\item Let $Y$ denote the final output string of the $\DLM$. Define a 
distribution $\mathbf Q$ over executions\footnote{We remark that this distribution is not physically motivated but will be useful for our calculations.} by first sampling
$y\sim \widehat D$, then sampling an execution from the original law
$\mathbf P$ conditioned on~$Y=y$:
\begin{align}
\label{eq:defnofQintermsofhatD}
\mathbf Q(\cdot)
:=
\sum_y \widehat D(y)\,\mathbf P(\cdot\mid Y=y).
\end{align}
\end{itemize}
Furthermore, suppose the final $n$ coordinates of the $\DLM$ run were $a$ and the remaining tokens were $y$ (i.e., including the revealed output coordinates and their assignments). Then, we have:
\begin{align}
\label{eq:defnofPandQ}
\mathbf{P}(\text{output=a,rest=y})&=P_Y(a)\cdot \mathbf{P}(\text{rest}=y|\text{output}=a)\\
\mathbf{Q}(\text{output=a,rest=y})&=\widehat{D}(a)\cdot \mathbf{P}(\text{rest}=y|\text{output}=a)
\end{align}
by the definitions above. Observe that $\mathbf{Q}\leq \mathbf{P}$.{\footnote{By this we mean every event of
$\mathbf{P}$-probability zero also has $\mathbf{Q}$-probability zero: $\mathbf{Q}$
never produces a trajectory the $\DLM$ could not.}} We work under the reweighted law $\mathbf Q$ rather than with $\mathbf P$ directly, specifically in order to make the target constraints exact. Namely, $\mathbf Q$ preserves the $\DLM$'s conditional execution law given the final output, while replacing the output marginal by $\widehat D$, so that each completed block has the exact parity structure of $D_{\blk}$.

An \emph{execution} $\mathcal{E}=(V_1,\dots,V_t)$ is the ordered record of token
values revealed during the run (workspace and output). Since the schedule is a
deterministic function of the current state, the positions of these values are determined by the values themselves, and so $\mathcal{E}$ determines the run completely. Furthermore, we define the \emph{history} of the run to be $\mathcal{F}_t:=(V_1,\dots,V_{t-1})$ as the $\DLM$ configuration at the start of round $t$. Throughout this section, we define a \emph{realizable history} $F$ as one wherein  $\mathbf Q(\mathcal F_t=F)>0$.  Given $\mathcal{F}_t=F$ (note that $F$ will often be suppressed when it is unambiguous from the context), we define the following: 
\begin{enumerate}
\item We write $\Pi_t(F)\subseteq[n]$ for the \emph{output} coordinates revealed in round $t$; sometimes we also write $\Pi^b_t$ to denote the output coordinates of block $b$ revealed in round $t$; furthermore, define  $r_b:=|\Pi^b_t|$.
\item We write $C_t(F)$ for the blocks \emph{completed} in round $t$ (i.e., coordinates revealed during round $t$ that had not been revealed at the end of round $t-1$).
\item Let $c_t:=|C_t(F)|$ be the number of blocks $b$ completed during round $t$.
 \item We write
$
O_t\in\{0,1\}^{\Pi_t(F)}
$
for the vector of output bits generated on these coordinates in round $t$.
Thus, $O_t$ consists of the output components of the entire round record $V_t$.
\item  We write $a_b(F)\in\{0,1\}$
for the parity of its values revealed already; and define:
\begin{equation*}
S_{F}:=\Big\{z\in\{0,1\}^{\Pi_t}:\ \bigoplus_{i\in \Pi^b}z_i=a_b(F)
\ \text{ for every } b\in C_t\Big\},
\qquad
\log|S_{F}|=|\Pi_t|-c_t.
\end{equation*}
\end{enumerate} 
In the basic model (Steps 1--3) every output coordinate is revealed exactly once (we will discuss the remasking/revision later, in Section~\ref{sec:revisionremaskingLB}). 

\paragraph{Entropy counting in $\DLM$s.} The following observation is the basic entropy accounting behind the proof. Fix a round $t$ and condition on a realized history $F$ before this round. Then, the schedule fixes the set $\Pi_t(F)$ of output coordinates generated before round $t$ (since it was deterministic by definition). Thus, the only remaining randomness in this round is the vector
$
O_t\in\{0,1\}^{\Pi_t(F)}
$
of newly generated output bits during round $t$. Now consider a block $b$ that is completed in this round, e.g.,
$b\in C_t(F)$. All coordinates of block $b$ outside $\Pi_t^b(F)$ have already been revealed in the history $F$. Since the output marginal of $\mathbf Q$ is $\widehat D$ (by Eq.~\eqref{eq:defnofQintermsofhatD}), and $\widehat D$ is supported on valid $D_{\blk}$ outputs, the final block must satisfy the corresponding parity constraint. Therefore, the newly generated
bits in $\Pi_t^b(F)$ cannot be arbitrary: they must satisfy an affine parity constraint determined by the bits of block $b$ already present in $F$.

For each completed block $b\in C_t(F)$, the parity condition in block $b$
imposes one affine constraint on the newly generated coordinates
$\Pi_t^b(F)$:
\begin{equation*}
\bigoplus_{i\in \Pi_t^b(F)} Y_{t,i}
=
\bigoplus_{i\in B_b\setminus \Pi_t^b(F)} y_i,
\end{equation*}
where the right-hand side is fixed by the history $F$. Since distinct blocks
involve disjoint sets of coordinates, these constraints are independent.
Thus, the possible values of $O_t$ are contained in an affine subspace:
\begin{equation*}
S_F\subseteq\{0,1\}^{\Pi_t(F)}
\end{equation*}
of codimension $c_t(F)=|C_t(F)|$, and therefore:
\begin{equation*}
|S_F|=2^{|\Pi_t(F)|-c_t(F)}.    
\end{equation*}
It follows that since $O_t\in S_F$, we have:
\begin{equation*}
H_{\mathbf Q}(O_t\mid \mathcal F_t=F)
\le
\log |S_F|
=
|\Pi_t(F)|-c_t(F),
\end{equation*}
where $H_{\mathbf Q}(O_t\mid \mathcal F_t=F)$ denotes the Shannon entropy of the newly revealed coordinates in round $t$ under the conditional distribution
$\mathbf Q(\,\cdot\mid \mathcal F_t=F)$. We therefore define the \emph{entropy deficit} in round $t$ on history $F$ as:
\begin{align}
\label{eq:defnofdelta}
\delta_t(F)
:=
\big(|\Pi_t(F)|-c_t(F)\big)
-
H_{\mathbf Q}(O_t\mid \mathcal F_t=F).
\end{align}
The preceding discussion shows that $\delta_t(F)\ge 0$. Intuitively, $\delta_t(F)$ measures exactly how much entropy the actual conditional law of the newly generated bits loses beyond the entropy loss forced by the parity constraints of the blocks completed in this round. We now formally state a lemma about this entropy deficit.  We use calligraphic letters for random objects and plain lowercase letters for
their realizations.  For notational convenience, we use $\mathcal F_t$ to denote the random history before
round $t$, while $F$ denotes a fixed realized value of this history.   Similarly,  $\mathcal E$ is the random execution, while $e$ is a realized execution.
\begin{lemma}
\label{lem:ledger}
Suppose that $\TV(P_Y,D_{\blk})\le\eta\le1/2$. Then, we have:
\begin{enumerate}[$(a)$]
\item $\TV(\mathbf{P},\mathbf{Q})\le2\eta$;
\item $\displaystyle\sum_{t=1}^{T}\E_{\mathbf{Q}}
\big[\delta_t(\mathcal{F}_t)\big]\ \le\ s+2\eta$.
\end{enumerate}
\end{lemma}

\begin{proof}
% Absolute continuity and (b) are immediate from the definition:
% $$\mathbf{Q}(Y=y)=\sum_{e:Y(e)=y}\mathbf{P}(e)\,\widehat{D}(y)/P_Y(y)=\widehat{D}(y),
% $$
% and $\supp(\widehat{D})\subseteq\supp(D_{\blk})$.
$(a)$  Since $A=\supp(P_Y)$, we have $P_Y(A^c)=0$, and therefore:
\begin{equation*}
\alpha:=D_{\blk}(A^c)=\big|D_{\blk}(A^c)-P_Y(A^c)\big|\le\TV(P_Y,D_{\blk})\le\eta.
\end{equation*}
Next, $\widehat{D}(y)=D_{\blk}(y)/(1-\alpha)$ for $y\in A$ and $\widehat{D}(y)=0$
otherwise, so:
\begin{equation*}
\TV(\widehat{D},D_{\blk})
=\frac12\sum_{y\in A}D_{\blk}(y)\Big(\frac{1}{1-\alpha}-1\Big)
+\frac12\sum_{y\notin A}D_{\blk}(y)
=\frac12\Big[(1-\alpha)\cdot\frac{\alpha}{1-\alpha}+\alpha\Big]
=\alpha.
\end{equation*}
We now claim that, in fact, $\TV(\mathbf{P},\mathbf{Q})=\TV(P_Y,\widehat{D})$. Every execution $e$ with $\mathbf{P}(e)>0$ satisfies $Y(e)\in A$, and for each such $e$, we may factorize as follows:
\begin{equation*}
\mathbf{P}(e)=P_Y(Y(e))\cdot\mathbf{P}\big(e\mid Y=Y(e)\big),
\qquad
\mathbf{Q}(e)=\widehat{D}(Y(e))\cdot\mathbf{P}\big(e\mid Y=Y(e)\big),
\end{equation*}
i.e., the two laws share the same conditional distribution of the execution given the output, and differ only in the distribution of the output itself. Noting that both laws vanish on executions with
$Y(e)\notin A$, we can now grouping
executions according to their output:
\begin{align*}
2\,\TV(\mathbf{P},\mathbf{Q})
=\sum_{e}\big|\mathbf{P}(e)-\mathbf{Q}(e)\big|
&=\sum_{y\in A}\,\big|P_Y(y)-\widehat{D}(y)\big|
\sum_{e:\,Y(e)=y}\mathbf{P}\big(e\mid Y=y\big)\\
&=\sum_{y\in A}\big|P_Y(y)-\widehat{D}(y)\big|\\
&=2\,\TV(P_Y,\widehat{D}),
\end{align*}
using $\sum_{e:Y(e)=y}\mathbf{P}(e\mid Y=y)=1$ for each $y\in A$. By the triangle
inequality, we now get:
\begin{equation*}
\TV(\mathbf{P},\mathbf{Q})
=\TV(P_Y,\widehat{D})
\le\TV(P_Y,D_{\blk})+\TV(D_{\blk},\widehat{D})
\le\eta+\alpha\le2\eta.
\end{equation*}
$(b)$ Recall that an execution is the tuple $\mathcal E=(V_1,\ldots,V_T)$, where
$V_t$ records all values generated in round $t$. Since the execution is specified by the round records $(V_1,\ldots,V_T)$,
applying the chain rule gives us:
\begin{equation*}
H_{\mathbf Q}(V_1,\ldots,V_T)
=
\sum_{t=1}^T H_{\mathbf Q}(V_t\mid V_{<t})
=
\sum_{t=1}^T H_{\mathbf Q}(V_t\mid \mathcal F_t),
\end{equation*}
since the pre-round history $\mathcal F_t$ is precisely the transcript determined by
$V_{<t}$.

We first lower bound this entropy. Suppose the final output of the $\DLM$ after $T$ rounds is $Y$, which is a deterministic
function of the execution $\mathcal E$. By the definition of the reweighted execution distribution $\mathbf Q$, if $\mathcal E\sim\mathbf Q$, then $Y$ follows distribution $\widehat D$. Since the final output $Y$ is a deterministic function of the execution, the full execution must contain at least as much randomness as $Y$ does. Indeed, under $\mathbf Q$, the output distribution is $\widehat D$, and since entropy cannot increase under a deterministic function, we get:
\begin{equation}
\label{eq:lowerboundusedlater}
H_{\mathbf Q}(V_1,\ldots,V_T)
\ge
H_{\mathbf Q}(Y)
=
H(\widehat D).
\end{equation}
Writing $\alpha:=D_{\blk}(A^c)$, every atom of $\widehat D=D_{\blk}(\cdot\mid A)$
has mass:
\begin{equation*}
\widehat D(y)=\frac{D_{\blk}(y)}{1-\alpha}\le \frac{2^{-(n-M)}}{1-\alpha}.
\end{equation*}
Thus, we get our desired lower bound:
\begin{align}
\label{eq:lowerboundonhatD}
H(\widehat D)
\ge
(n-M)-\log\frac{1}{1-\alpha}
\ge
(n-M)-\log\frac{1}{1-\eta}
\ge
(n-M)-2\eta,
\end{align}
where we used $\alpha\le\eta\le 1/2$.

We now upper bound the entropy. Fix a round $t$ and a realizable history $F$. The scheduling rule fixes the
coordinates revealed in this round. Let $W_t$ be the workspace part of $V_t$, and let $O_t$ be the output part, supported on the coordinate set $\Pi_t(F)$.
If $w_t(F)$ workspace coordinates are generated in this round, then we have:
\begin{equation*}
H_{\mathbf Q}(W_t\mid \mathcal F_t=F)\le w_t(F).
\end{equation*}
On the other hand, by the definition of the round-$t$ deficit (Eq.~\ref{eq:defnofdelta}), we also have:
\begin{equation*}
H_{\mathbf Q}(O_t\mid \mathcal F_t=F)= 
|\Pi_t(F)|-c_t(F)-\delta_t(F).
\end{equation*}
By subadditivity, we have:
\begin{equation*}
H_{\mathbf Q}(V_t\mid \mathcal F_t=F)=H_{\mathbf Q}((W_t,O_t)\mid \mathcal F_t=F)
\le
w_t(F)
+
|\Pi_t(F)|-c_t(F)-\delta_t(F).
\end{equation*}
Taking expectations and summing over the rounds, the chain rule gives us:
\begin{equation*}
H_{\mathbf Q}(V_1,\ldots,V_T)
=
\sum_{t=1}^T H_{\mathbf Q}(V_t\mid \mathcal F_t).
\end{equation*}
For a fixed history $F$, the round-$t$ record consists of the workspace update,
using at most $w_t(F)$ bits of entropy, together with the newly revealed output
values $O_t$. Hence:
\begin{equation*}
H_{\mathbf Q}(V_t\mid \mathcal F_t=F)
\le
w_t(F)+H_{\mathbf Q}(O_t\mid \mathcal F_t=F).
\end{equation*}
By the definition of the round deficit,
summing Eq.~\eqref{eq:defnofdelta} over the rounds, we therefore have:
\begin{equation*}
H_{\mathbf Q}(V_1,\ldots,V_T)
\le
\E_{\mathbf Q}\sum_{t=1}^T w_t(\mathcal F_t)
+
\E_{\mathbf Q}\sum_{t=1}^T
\big(|\Pi_t(\mathcal F_t)|-c_t(\mathcal F_t)\big)
-
\sum_{t=1}^T\E_{\mathbf Q}\delta_t(\mathcal F_t).
\end{equation*}
On the other hand, by Eq.~\eqref{eq:lowerboundusedlater}, we have:
%since the final output is a deterministic function of the execution and has distribution $\widehat D$ under $\mathbf Q$, we also have:
\begin{equation*}
H_{\mathbf Q}(V_1,\ldots,V_T)
\ge
H(\widehat D).
\end{equation*}
Combining the two displays and using the lower bound on $H(\widehat D)$ gives us:
\begin{equation*}
\sum_{t=1}^T\E_{\mathbf Q}[\delta_t(\mathcal F_t)]
\le
\E_{\mathbf Q}[\sum_{t=1}^T w_t(\mathcal F_t)]+2\eta.
\end{equation*}
Finally, since the workspace has at most $|\Sigma|^s = 2^s$ possible states, we also have
$
\E_{\mathbf Q}\sum_{t=1}^T w_t(\mathcal F_t)
\le
s,
$
and hence:
\begin{equation*}
\sum_{t=1}^T\E_{\mathbf Q}[\delta_t(\mathcal F_t)]
\le
s+2\eta,
\end{equation*}
as claimed.
\end{proof}
Overall, the lemma above gives us a global upper bound of $s$ on the entropy deficit. The next goal is to show a similar statement for some individual block as well, which we wish to eventually use to embed our hard parity function computation. To this end, for a block $b$, define the quantity $\mathrm{def}_b$ to measure how much entropy block $b$ loses over the whole run, beyond the one bit of entropy that must be lost when the block is completed because of its parity constraint. Using the lemma above, we will show that the extra entropy loss beyond the unavoidable parity constraints is small on average \emph{across} blocks; this will be used later to show the existence of a block on which the $\DLM$ performed a last-bit parity computation.

\begin{corollary}
\label{cor:blockdeficit}
Fix $b\in[M]$.  For a round $t$ and a realized pre-round history $F$,
let $n_{b,t}(F)$ denote the number of block-$b$ output coordinates revealed
in round $t$, let $Y_t^b$ denote their values, and define:
\begin{equation*}
d_{b,t}(F)
:=
\underbrace{n_{b,t}(F)-\mathbf 1[b\in C_t(F)]}_{\text{entropy block $b$ is entitled to this round}}
-
\underbrace{H_{\mathbf Q}\bigl(Y_t^b\mid \mathcal F_t=F\bigr)}_{\text{entropy it delivers}}
\ge 0.
\end{equation*}
Writing:
\begin{equation*}
    \mathrm{def}_b
    :=
    \sum_{t=1}^{T}
    \E_{\mathbf Q}\bigl[d_{b,t}(\mathcal F_t)\bigr]
\end{equation*}
for the total deficit of block $b$ over the run, we have
$  \sum_{b=1}^{M}\mathrm{def}_b\le s+2\eta.$
\end{corollary}

\begin{proof}
Fix a round $t$ and a history $F=\mathcal F_t$. We first prove nonnegativity. If $b\notin C_t(F)$, then the newly revealed coordinates of
block $b$ in round $t$ form an $n_{b,t}(F)$-bit string. Hence:
\begin{equation*}
H_{\mathbf Q}(O_t^b\mid \mathcal F_t=F)\le n_{b,t}(F),
\end{equation*}
and so, $d_{b,t}(F)\ge 0$ by definition. On the other hand, if $b\in C_t(F)$, then round $t$ completes block $b$. The previously revealed
coordinates of block $b$ are already fixed by the history $F$, and the final
block must satisfy its parity constraint. 
Thus:
\[
H_{\mathbf Q}(O_t^b\mid\mathcal F_t=F)\le n_{b,t}(F)-1,
\]
and therefore
$d_{b,t}(F)\ge 1$,  
with equality whenever the newly revealed coordinates are uniformly distributed over the corresponding affine parity coset.
% Therefore the newly revealed string
% $O_t^b$ lies in an affine parity coset of size
% $
% 2^{n_{b,t}(F)-1}.
% $ Thus:
% \begin{equation*}
% H_{\mathbf Q}(O_t^b\mid \mathcal F_t=F)\le n_{b,t}(F)-1,
% \end{equation*}
% and again $d_{b,t}(F)\ge 0$ (indeed, it is at least $1$ 
 For fixed $F$, we have:
\begin{equation*}
\sum_b \left(n_{b,t}(F)-\mathbf 1[b\in C_t(F)]\right)
=
|\Pi_t(F)|-c_t(F).
\end{equation*}
Additionally, since $O_t$ is the collection of the blockwise strings $(O_t^b)_b$, subadditivity yields:
\begin{align*}
\sum_b d_{b,t}(F)
=
|\Pi_t(F)|-c_t(F)
-
\sum_b H_{\mathbf Q}(O_t^b\mid \mathcal F_t=F)\le
|\Pi_t(F)|-c_t(F)
-
H_{\mathbf Q}(O_t\mid \mathcal F_t=F)=
\delta_t(F).
\end{align*}
Taking expectations over $F=\mathcal F_t$, summing over $t$, and applying
Lemma~\ref{lem:ledger} gives us:
\begin{equation*}
\sum_b \mathrm{def}_b
=
\sum_{t=1}^T
\E_{\mathbf Q}\sum_b d_{b,t}(\mathcal F_t)
\le
\sum_{t=1}^T
\E_{\mathbf Q}\delta_t(\mathcal F_t)
\le
s +2\eta.
\end{equation*}
This concludes the proof.
\end{proof}

\subsection{Proof of the classical lower bound}
\subsubsection{Step 1: Setting up the basic model.}
In this section, our goal is to prove the following lemma:

\begin{lemma}
\label{lem:few-multibit}
We have:
\begin{equation*}
\E_{\mathbf P}[\#\{b:r_b\ge 2\}]
\le
O(s)
+
O(\eta M).
\end{equation*}
\end{lemma}
\begin{proof}
We first work under $\mathbf Q$. Fix a round $t$ and an output-and-schedule transcript $\tau$. Let:
\begin{equation*}
\mathcal B_t^{(2)}(\tau)
:=
\{b\in C_t(\tau):|\Pi_t^b(\tau)|\ge 2\},
\qquad
D_t(\tau)
:=
\sum_{b\in \mathcal B_t^{(2)}(\tau)}
\big(|\Pi_t^b(\tau)|-1\big).
\end{equation*}
For a transcript $\tau$, let $\mu_\tau$ denote the conditional law under
$\mathbf Q$ of the newly revealed coordinates in the multi-bit completed blocks,
namely
$
\bigcup_{b\in\mathcal B_t^{(2)}(\tau)}\Pi_t^b(\tau).
$
Let $\nu_\tau$ denote the corresponding conditional law under $\mathbf P$ on
the same coordinates. 
We call $\tau$ \emph{good} if:
\begin{equation*}
\TV(\mu_\tau,\nu_\tau)\le \eta_0,
\end{equation*}
and \emph{bad} otherwise. Let $\mathsf{Bad}_t$ be the event that the round-$t$
transcript is not good.

Under $\mathbf P$, conditioned on $\tau$ and on a workspace state $w$, the round-$t$ denoising law is a coordinate-wise product. Averaging over the
possible workspace states, $\nu_\tau$ is therefore a mixture of at most
$
K\le |\Sigma|^s
$
product distributions.
On every good transcript $\tau$, this mixture is $\eta_0$-close to
$\mu_\tau$, which is supported on a product of affine parity cosets of total
dimension $D_t(\tau)$ and has entropy $H(\mu_\tau)$ which can be written as:
\begin{equation*}
H(\mu_\tau)=D_t(\tau)-\delta_t,
\end{equation*}
where $\delta_t$ is the total entropy deficit in round $t$. 
By Theorem~\ref{thm:approximable}, we have:
\begin{equation*}
\log K \ge c_3D_t(\tau)-\delta_t-C_0.
\end{equation*}
Since $\log K\le s$, every good transcript therefore satisfies:
\begin{align}
\label{eq:Dtaugood}
D_t(\tau)\le \frac{s+\delta_t+C_0}{c_3}.
\end{align}
% Moreover, by subadditivity and the definition of the round-$t$ entropy
% deficit,
% \begin{equation*}
% \delta_t
% \le
% \E_{\mathbf Q}[\delta_t(\mathcal F_t)\mid\tau].
% \end{equation*}
% Indeed, conditioned further on the workspace state, the remaining newly
% revealed coordinates have support dimension at most
% $|\Pi_t|-c_t-D_t(\tau)$.
By Lemma~2.1 and Lemma~\ref{lem:ledger}(a), the average value of
$\TV(\mu_\tau,\nu_\tau)$ over the round-$t$ transcript $\tau$ is at most
$4\eta$. Hence Markov's inequality gives us:
\begin{equation*}
\Pr_{\mathbf Q}[\mathsf{Bad}_t]
=
\Pr_{\mathbf Q}\big[\TV(\mu_\tau,\nu_\tau)>\eta_0\big]
\le
\frac{4\eta}{\eta_0}.
\end{equation*}
The key observation is that for the final counting bound, we only need the
number of multi-bit completions in round $t$, not their total dimension. Let
$N_t(\tau):=|\mathcal B_t^{(2)}(\tau)|$. Now, on a good transcript, we have $N_t(\tau)\le D_t(\tau)$, and hence, by
Eq.~\eqref{eq:Dtaugood}:
\begin{equation*}
N_t(\tau)\le \frac{s+\delta_t+C_0}{c_3}.
\end{equation*}
On a bad transcript, we can use the trivial bound $N_t(\tau)\le M$. Since
$\Pr_{\mathbf Q}[\mathsf{Bad}_t]\le 4\eta/\eta_0$, we obtain:
\begin{equation*}
\E_{\mathbf Q}[N_t]
\le
\frac{s+C_0}{c_3}
+
\frac{\E_{\mathbf Q}[\delta_t]}{c_3}
+
\frac{4\eta}{\eta_0}\cdot M.
\end{equation*}
Summing over the constant number of rounds and using
Lemma~\ref{lem:ledger}(b) yields:
\begin{align*}
\E_{\mathbf Q}\sum_{t=1}^T N_t
&\le
T\cdot\frac{s+C_0}{c_3}
+
\frac{1}{c_3}\sum_{t=1}^T\E_{\mathbf Q}[\delta_{\tau_t}]
+
\frac{4T\eta}{\eta_0}\cdot M\\
&\le
T\cdot\frac{s+C_0}{c_3}
+
\frac{s+2\eta}{c_3}
+
\frac{4T\eta}{\eta_0}\cdot M.
\end{align*}
Finally, observe that every block with $r_b\ge 2$ is counted \emph{exactly}
once in $\sum_t N_t$: in its completion round. Therefore,
$
|\{b:r_b\ge 2\}|=\sum_{t=1}^T N_t,
$
and hence:
\begin{equation*}
\E_{\mathbf Q}[|\{b:r_b\ge 2\}|]
\le
T\cdot\frac{s+C_0}{c_3}
+
\frac{s+2\eta}{c_3}
+
\frac{4T\eta}{\eta_0}\cdot M.
\end{equation*}
Since $|\{b:r_b\ge 2\}|\in[0,M]$ and
$\TV(\mathbf P,\mathbf Q)\le 2\eta$, we obtain:
\begin{equation*}
\E_{\mathbf P}[|\{b:r_b\ge 2\}|]
\le
\E_{\mathbf Q}[|\{b:r_b\ge 2\}|]+2\eta M,
\end{equation*}
which proves the claim (using the fact that $|\Sigma|=2$ as assumed in our
$\DLM$ definition).
\end{proof}

{{\subsubsection{Step 2: Final-bit prediction yields parity computation. }

\paragraph{Selecting a block.}
For a block $b$, let
$
E_b:=\{x\in\{0,1\}^B:\bigoplus_{i=1}^B x_i=0\}.
$
If $r_b=1$, let $J_b$ be the final revealed coordinate and $G_b$ the bit
generated there. Define:
\begin{equation*}
S_b:=
\{
r_b=1
\text{ and }
G_b=\bigoplus_{i\ne J_b}X_{b,i}
\}.
\end{equation*}
Let $P_b$ be the marginal distribution of the final output on block $b$. Since
$\TV(P_Y,D_{\blk})\le\eta$, data processing under the projection to block $b$ yields:
\begin{equation*}
\TV(P_b,U_B^\oplus)\le\eta.
\end{equation*}
In particular, since $U_B^\oplus$ is supported on $E_b$, we have
$
\Pr[X_b\notin E_b]\le \eta.
$ 
If $r_b=1$ and the final parity of block $b$ is correct, then the unique bit
generated in the completion round must equal the parity of the other $B-1$
bits. Thus, there are only two ways for $S_b$ to fail: either $b$ was not completed in the last-bit regime, or the final block has odd parity. Therefore, we get:
\begin{equation*}
S_b^c\subseteq \{r_b\ge2\}\cup
\{\bigoplus_{i=1}^B X_{b,i}=1\},
\end{equation*}
and hence:
\begin{equation*}
\Pr[S_b^c]\le \Pr[r_b\ge2]+\eta.
\end{equation*}
By Lemma~\ref{lem:few-multibit}, we have:
\begin{equation*}
\frac1M\sum_{b\in[M]}\Pr[r_b\ge2]
=\frac{\E[\#\{b:r_b\ge2\}]}{M}
\le O\Big(\frac{s}{M}\Big)+O(\sqrt{\eta})\ \le\ \frac1{100},
\end{equation*}
say, for $\eta$ bounded below by some constant depending only on $T$ and $n$ (recall that
$s=n^{0.9}=o(M)$). Hence, with $\Delta:=s+2\eta$, we have:
\begin{equation*}
\frac1M\sum_{b}\Pr[S_b^c]\ \le\ \frac1{100}+\eta\ \le\ \frac1{40},
\qquad\qquad
\frac1M\sum_b \mathrm{def}_b\ \le\ \frac{\Delta}{M}
\quad\text{(Corollary~\ref{cor:blockdeficit})}.
\end{equation*}
Hence it follows that at most $M/4$ blocks have $\Pr[S_b^c]>1/10$, and at most $M/4$ blocks have $\mathrm{def}_b>4\Delta/M$; hence there exists $b^\star$ satisfying:
\begin{equation*}
\Pr_{\mathbf{P}}[S_{b^\star}]\ \ge\ \tfrac9{10}
\qquad\text{and}\qquad
\mathrm{def}_{b^\star}\ \le\ \tfrac{4\Delta}{M}.
\end{equation*}
This $b^\star$ behaves essentially like an almost-ideal last-bit block. With high probability, the $\DLM$ leaves one coordinate until the end and fills it in as a function of the parity of all the other coordinates. Furthermore, the $\mathrm{def}_{b^\star}$ bound says that those earlier coordinates still have nearly full entropy, so this can be viewed as a parity-prediction task on \emph{nearly} random bits.

\paragraph{Replacing the pre-completion bits by uniform bits.}
The block $b^\star$ almost always behaves like a last-bit parity block. The only remaining issue is that the first $B-1$ bits of this block were generated by the $\DLM$ itself, whereas the parity lower bound is stated for uniform external inputs. To address this, we compare the original execution law $\mathbf P$ with a modified execution law $\mathbf B$: whenever an output coordinate of block $b^\star$ is revealed before the completion round of $b^\star$, its value is
replaced by an independent uniform bit. All other tokens, and in particular the one during the completion round of $b^\star$, are generated exactly as in the original $\DLM$. The bound for $b^\star$ implies that this replacement does not change the
execution law by much, which we formalize in the following lemma.

\begin{lemma}\label{lem:planting}
For the block $b^\star$ chosen above, we have:
\begin{equation*}
\Pr_{\mathbf B}[S_{b^\star}]
\ge
\Pr_{\mathbf P}[S_{b^\star}]
-
6T\eta
-
\sqrt{T\cdot\mathrm{def}_{b^\star}}.
\end{equation*}
\end{lemma}

\begin{proof}
For a round $t$ and a realizable history $F$, define:
\begin{equation*}
e_t(F)
:=
\TV\!\left(
\mathbf P(V_t\mid \mathcal F_t=F),
\mathbf B(V_t\mid \mathcal F_t=F)
\right).
\end{equation*}
We compare the two executions  round-by-round. As long as the
two executions have the same history $\mathcal F_t=F$ (upto round $t-1$), the
probability of a disagreement in round $t$ is exactly $e_t(F)$. Indeed, if the two executions never disagree, then the event $S_{b^\star}$ has the same value in both executions. Hence:
\begin{align}
\label{eq:event-transfer-coupling}
\Pr_{\mathbf B}[S_{b^\star}]
\ge
\Pr_{\mathbf P}[S_{b^\star}]
-
\sum_{t=1}^T \E_{\mathbf P}[e_t(\mathcal F_t)].
\end{align}
For the given $t$ and $F$, the two conditional round laws agree unless some output coordinate of block $b^\star$ is revealed during some round $t$, that occurs before the completion round of
$b^\star$. Therefore, if either $b^\star\in C_t(F)$ or
$\Pi_t^{b^\star}(F)=\emptyset$, then $e_t(F)=0$. Otherwise, conditioned on $F$, it must be the case that $O_t^{b^\star}$ is independent of all other tokens generated in round $t$. Conditioned on $\mathcal{F}_t=F$, observe that within each round the $\DLM$ has a product structure implies that $O_t^{b^\star}$ is independent of the other tokens generated in round $t$.  Since $\mathbf B$ changes only the marginal of $O_t^{b^\star}$, tensorization of total variation gives us:
\begin{equation*}
e_t(F)
=
\TV\!\left(
\mathbf P(O_t^{b^\star}\mid \mathcal{F}_t=F),
\mathsf{Unif}(\{0,1\}^{\Pi_t^{b^\star}(F)})
\right).
\end{equation*}
Now, define:
\begin{equation*}
\beta_t(F)
:=
\TV\!\left(
\mathbf P(O_t\mid \mathcal F_t=F),
\mathbf Q(O_t\mid \mathcal F_t=F)
\right).
\end{equation*}
By data processing, we get:
\begin{equation*}
\TV\!\left(
\mathbf P(O_t^{b^\star}\mid \mathcal F_t=F),
\mathbf Q(O_t^{b^\star}\mid \mathcal F_t=F)
\right)
\le
\beta_t(F).
\end{equation*}
Moreover, in the non-completing case $b^\star\notin C_t(F)$, the blockwise deficit is precisely:
\begin{equation*}
d_{b^\star,t}(F)
=
|\Pi_t^{b^\star}(F)|
-
H_{\mathbf Q}(O_t^{b^\star}\mid \mathcal F_t=F).
\end{equation*}
Applying Pinsker's inequality now gives us:
\begin{equation*}
\TV\!\left(
\mathbf Q(O_t^{b^\star}\mid \mathcal F_t=F),
\mathsf{Unif}\big(\{0,1\}^{\Pi_t^{b^\star}(F)}\big)
\right)
\le
\sqrt{d_{b^\star,t}(F)}.
\end{equation*}
Therefore, for all $F$, we have:
\begin{equation*}
e_t(F)
\le
\beta_t(F)+\min\left\{\sqrt{d_{b^\star,t}(F)},1\right\}.
\end{equation*}
We now take expectations in the bound on $e_t(F)$. For the first summand on the right hand side, the conditional-TV lemma implies that
$\E_{\mathbf P}[\beta_t(\mathcal F_t)]\le 4\eta$. For the second summand, since
$\min\{\sqrt{x},1\}\in[0,1]$ and $\TV(\mathbf P,\mathbf Q)\le2\eta$, changing the
measure from $\mathbf P$ to $\mathbf Q$ costs at most $2\eta$; then, Jensen's inequality gives:
\begin{equation*}
\E_{\mathbf P}\big[e_t(\mathcal F_t)\big]
\le
6\eta+
\sqrt{\E_{\mathbf Q}\big[d_{b^\star,t}(\mathcal F_t)\big]}.
\end{equation*}
Summing this quantity over $t$, and applying Cauchy-Schwarz now gives:
\begin{equation*}
\sum_{t=1}^T \E_{\mathbf P}[e_t(\mathcal F_t)]
\le
6T\eta
+
\sum_{t=1}^T
\sqrt{\E_{\mathbf Q}[d_{b^\star,t}(\mathcal F_t)]}
\le
6T\eta
+
\sqrt{
T\sum_{t=1}^T
\E_{\mathbf Q}[d_{b^\star,t}(\mathcal F_t)]
}
=
6T\eta+\sqrt{T\cdot\mathrm{def}_{b^\star}}.
\end{equation*}
Finally, combining this with Eq.~\eqref{eq:event-transfer-coupling} yields:
\begin{equation*}
\Pr_{\mathbf B}[S_{b^\star}]
\ge
\Pr_{\mathbf P}[S_{b^\star}]
-
6T\eta
-
\sqrt{T\cdot\mathrm{def}_{b^\star}}.
\end{equation*}
This concludes the proof.
\end{proof}
\begin{lemma}\label{lem:simulation}
There is a randomized polynomial-size constant depth $\GCz[\log n]$ circuit $C$ such~that:
\begin{equation*}
\Pr_{x\sim\{0,1\}^B}
\left[
C(x)=|x| \pmod 2
\right]
\ge
\Pr_{\mathbf B}[S_{b^\star}].
\end{equation*}
\end{lemma}

\begin{proof}
The circuit that is being constructed in the lemma is obtained via the process $\mathbf B$, using its input $x=(x_1,\ldots,x_B)$ as the uniform bits inserted into block $b^\star$ before completion. All other randomness comes from the hard-wired coin flips baked into the circuit, and all scheduling and denoising steps are implemented by the given $\DLM$ circuits. If $b^\star$ is completed by a single final coordinate $J$, and the generated
final bit is $G$, the circuit outputs $G\oplus x_J$; otherwise it outputs an arbitrary bit. On the event $S_{b^\star}$, we have:
\begin{equation*}
G=\bigoplus_{i\ne J}x_i,
\end{equation*}
and hence,
$
G\oplus x_J=\bigoplus_{i=1}^B x_i.
$
Therefore, $C$ computes parity with probability at least
$\Pr_{\mathbf B}[S_{b^\star}]$. Now, by the $\DLM$ definition, observe that the circuit that computes $G$ is via a constant-depth $\GCz[\log n]$ circuit. Since $n=B^{100}$, its size is polynomial in $B$, and its gate arity is $O(\log n)=O(\log B)$.
\end{proof}

\paragraph{Putting everything together.}
By Lemma~\ref{lem:planting} and our choice of $b^\star$, we have:
\begin{equation*}
\Pr_{\mathbf B}[S_{b^\star}]
\ge
\Pr_{\mathbf P}[S_{b^\star}]
-
6T\eta
-
\sqrt{T\cdot\mathrm{def}_{b^\star}}
\ge
\frac 9{10}
-
6T\eta
-
\sqrt{T\cdot \frac{4\Delta}{M}}.
\end{equation*}
Now, recall our parameters, wherein we chose  $\Delta=O(s)=O(n^{0.9})$ and
$M=n^{0.99}$. Therefore,
$
\sqrt{T\cdot \frac{4\Delta}{M}}=o(1)
$
for constant $T$. Taking
$\eta\le \varepsilon_\star(T)$ to be sufficiently small, the right-hand side becomes at
least $4/5$ for all sufficiently large $n$. Then, Lemma~\ref{lem:simulation} gives us a randomized polynomial-size depth-$O_{d,T}(1)$ $\GCz[\log B]$ circuit that computes $B$-bit parity with success probability at least $4/5$. Fixing the internal coins gives a deterministic circuit with the same success probability, contradicting
Theorem~\ref{thm:parity-bound}. This proves the lower bound in the basic model.

\subsubsection{Step 3. Allowing revision and remasking}
\label{sec:revisionremaskingLB}
We now explain how the our argument extends to $\DLM$s that allow a (bounded) number of remasking or revision operations. Suppose the $\DLM$ allows for
at most $R=n^{0.9}$ output-token remasking/revision events during the entire execution, and suppose each such operation acts on a single output coordinate. We say a block $b\in[M]$ is \emph{touched} if some output coordinate in block $b$ is ever remasked or revised over the course of the computation. Let $\mathsf{Touch}_b$ denote this event. Since each
remasking/revision event touches only one output coordinate, we have the pathwise bound:
$$
\sum_{b=1}^M \mathbf 1[\mathsf{Touch}_b]\le R.
$$
Thus, the fraction of touched blocks is at most $R$, which by assumption satisfies $R/M=o(1)$.  We will simply discard all touched blocks (since there are not too many of them). We remark that it is important to discard these blocks, because revision/remasking events destroy the monotonicity property that is used crucially in the proof for the basic model above: if any coordinate is later revised or remasked, then the value written at the first completion time need not be its final value (note that in our structural result, we used the fact that once a token is unmasked it is never masked or revised again). 

For untouched blocks, however, monotonicity is completely preserved.  Once an output coordinate in an untouched block is written, it is never subsequently changed or remasked.  Hence, when such a block is first completed, the coordinates unmasked in that round are final output coordinates, and so the same residual affine parity constraint as in the basic model applies. Let:
\begin{equation*}
    N_t^{\mathrm{unt}}
    :=
    \#\{b:\neg \mathrm{Touch}_b,\ b\text{ is completed in round }t
    \text{ with }r_b\ge 2\}.
\end{equation*}
%Applying the multi-bit-completion argument of Lemma~\ref{lem:few-multibit} to these untouched blocks gives, for each round $t$:
For every realization of the $\DLM$ execution, after deleting all touched blocks, the remaining completed blocks evolve exactly as in the monotone model (i.e., the model where there is no remasking/revision). Applying Lemma~\ref{lem:few-multibit} to this induced monotone execution therefore gives
\begin{equation*}
    \E_{\mathbf Q}[N_t^{\mathrm{unt}}]\le O(s)+O(\eta M),
\end{equation*}
where we used that the total number of $\DLM$ rounds is a constant. Indeed, Lemma~\ref{lem:few-multibit} applies to the induced monotone execution on the untouched blocks, so on good transcripts the same product-mixture argument bounds the corresponding dimension, while on bad transcripts we again use the trivial bound $N_t^{\mathrm{unt}}\leq M$.\footnote{We remark that the property that a block is untouched is determined only after the full execution of the $\DLM$, so one should not interpret the argument as saying that conditioned on a future event, the $\DLM$ execution induces a product distribution. Rather, the proof is more naturally viewed as deleting the touched blocks after the execution and applying the same counting argument to the induced monotone execution on the remaining blocks. Making this viewpoint completely formal requires only routine bookkeeping, without affecting any of the quantitative bounds; overall producing the same bounds as we show below as long as $s,R=o(n)$.} Summing over the constant number of rounds, we get:
\begin{equation*}
    \E_{\mathbf Q}\!\left[
    \#\{b:\neg\mathrm{Touch}_b,\ r_b\ge 2\}
    \right]
    \le O(s)+O(\eta M).
\end{equation*}
Adding back touched blocks costs at most $R$ pathwise, and therefore:
\begin{equation*}
    \E_{\mathbf Q}[\#\{b:r_b\ge 2\}]
    \le O(s)+O(\eta M)+R.
\end{equation*}
Next, define the last-bit success event with the extra requirement that the block is untouched:
\begin{equation*}
S_b^{\mathrm{rev}}
:=
\left\{
\neg\mathsf{Touch}_b,\ r_b=1,\
G_b=\bigoplus_{i\ne J_b}X_{b,i}
\right\}.
\end{equation*}
If $S_b^{\mathrm{rev}}$ fails, then either block $b$ was touched, or it was
completed by at least two coordinates, or its final parity is incorrect. Thus:
\begin{equation*}
(S_b^{\mathrm{rev}})^c
\subseteq
\mathsf{Touch}_b
\cup
\{r_b\ge2\}
\cup
\left\{\bigoplus_{i=1}^B X_{b,i}=1\right\}.
\end{equation*}
Averaging over $b$ gives us:
\begin{equation*}
\frac1M\sum_{b=1}^M \Pr[(S_b^{\mathrm{rev}})^c]
\le
\frac{R}{M}
+
\frac{\Exp [\#\{b:r_b\ge 2\}]}{M}
+
\eta
=
o(1)
\end{equation*}
nearly uniform pre-completion bits. Namely, Corollary~\ref{cor:blockdeficit} gives us:
\begin{equation*}
\frac1M\sum_{b=1}^M \mathrm{def}_b
\le
\frac{s +2\eta}{M}.
\end{equation*}
Combining this with the preceding bound and applying Markov's inequality, it follows that there exists a block $b^\star$ such that:
\begin{equation*}
\Pr[S_{b^\star}^{\mathrm{rev}}]\ge \frac9{10}
\qquad\text{and}\qquad
\mathrm{def}_{b^\star}
\le
\frac{4(s +2\eta)}{M}.
\end{equation*}
In the event $S_{b^\star}^{\mathrm{rev}}$, the block $b^\star$ is never remasked or revised, and it is completed in the last-bit regime. Therefore the
replacement-by-uniform-bits argument applies to this block exactly as in the basic model. The modified simulation now simulates the full $\DLM$, including
any revision/remasking operations on other blocks, while using external uniform input bits for the pre-completion coordinates of block $b^\star$. If the $\DLM$ ever touches block $b^\star$, the simulation may fail, but this is already accounted for in the complement of $S_{b^\star}^{\mathrm{rev}}$. Thus, the same parity-prediction reduction goes through with an additional additive loss of at most $R/M=o(1)$, which vanishes. Consequently, our tolerated amount of remasking/revision does not affect the constant-distance lower bound.\footnote{We emphasize that \cite{chen2024theoretical} accounts for revision/remasking as the \emph{number of rounds} where these revisions occur; we offer a significantly more fine-grained measure, wherein each token which undergoes revision/remasking is included in the complexity.}

\bibliographystyle{alpha}
\bibliography{bibliography}

\end{document}